\documentclass{article}

    \usepackage[preprint]{neurips_2024}

\usepackage[utf8]{inputenc} %
\usepackage[T1]{fontenc}    %
\usepackage{hyperref}       %
\usepackage{url}            %
\usepackage{booktabs}       %
\usepackage{amsfonts}       %
\usepackage{nicefrac}       %
\usepackage{microtype}      %
\usepackage{xcolor}         %
\usepackage{amsmath}
\usepackage{amssymb}
\usepackage{mathtools}
\usepackage{amsthm}
\usepackage{stackengine}
\usepackage{lipsum}
\usepackage{dsfont}
\usepackage{graphicx}
\usepackage{tabularx}
\usepackage{float}
\usepackage{mathtools}
\usepackage{mathrsfs}
\usepackage{enumitem}
\usepackage{algorithm}
\usepackage{xspace}
\usepackage{xparse}
\usepackage{cleveref}
\usepackage{multirow}
\usepackage{graphicx}
\usepackage[export]{adjustbox}
\usepackage[framemethod=TikZ]{mdframed}
\usepackage[noend]{algpseudocode}

\newtheorem{theorem}{Theorem}[section]
\newtheorem{lemma}[theorem]{Lemma}

\newtheorem{definition}[theorem]{Definition}

\newtheorem{claim}[theorem]{Claim}
\newtheorem{fact}[theorem]{Fact}

\renewenvironment{quote}
  {\list{}{\rightmargin=0.3cm \leftmargin=0.3cm}%
   \item\relax}
  {\endlist}

\newcommand{\cA}{\mathcal{A}}

\newcommand{\cD}{\mathcal{D}}
\newcommand{\cE}{\mathcal{E}}

\newcommand{\cI}{\mathcal{I}}

\newcommand{\cM}{\mathcal{M}}

\newcommand{\cS}{\mathcal{S}}

\newcommand{\cU}{\mathcal{U}}

\newcommand{\cX}{\mathcal{X}}
\newcommand{\cY}{\mathcal{Y}}

\newcommand{\N}{\mathbb{N}}
\newcommand{\R}{\mathbb{R}}

\newcommand{\PAREN}[1]{{\left( {#1} \right)}}

\newcommand{\bigparen}[1]{{\big( {#1} \big)}}
\newcommand{\paren}[1]{{( {#1} )}}

\newcommand{\bracket}[1]{{[ {#1} ]}}
\newcommand{\BigBracket}[1]{{\big[ {#1} \big]}}

\newcommand{\card}[1]{\left| {#1} \right|}
\newcommand{\set}[1]{\left\{ {#1} \right\}}
\newcommand{\eps}{\varepsilon}

\DeclareMathOperator*{\argmax}{arg\,max}

\NewDocumentCommand\p{ m g }{
  \ensuremath{
    \IfNoValueTF{#2}
    {\Pr [ #1 ]}
    {\Pr_{#1}[#2]}
  }
}
\RenewDocumentCommand\P{ m g }{
  \ensuremath{
    \IfNoValueTF{#2}
    {\Pr \left[#1\right]}
    {\Pr_{#1}\left[#2\right]}
  }
}
\DeclareMathOperator*{\Expectation}{\mathbb{E}}
\NewDocumentCommand\E{ m g }{
  \ensuremath{
    \IfNoValueTF{#2}
    {\Expectation \left[#1\right]}
    {\Expectation_{#1}\left[#2\right]}
  }
}
\DeclareMathOperator*{\Variance}{\mathbb{V}\mathrm{ar}}
\NewDocumentCommand\Var{ m g }{
  \ensuremath{
    \IfNoValueTF{#2}
    {\Variance \left[#1\right]}
    {\Variance_{#1}\left[#2\right]}
  }
}

\NewDocumentCommand\GammaDist{ m g }{
  \ensuremath{
    \IfNoValueTF{#2}
    {\operatorname{\mathcal{G}amma}\left( {#1} \right)}
    {\operatorname{\mathcal{G}amma}\left( {#1}, {#2} \right)}
  }
}

\NewDocumentCommand\LapNoise{ m g }{
  \ensuremath{
    \IfNoValueTF{#2}
    {\operatorname{\mathbb{L}ap}\left( {#1} \right)}
    {\operatorname{\mathbb{L}ap}\left( {#1}, {#2} \right)}
  }
}

\NewDocumentCommand\GumbelNoise{ m g }{
  \ensuremath{
    \IfNoValueTF{#2}
    {\operatorname{\mathbb{G}umbel}\left( {#1} \right)}
    {\operatorname{\mathbb{G}umbel}\left( {#1}, {#2} \right)}
  }
}

\newcommand{\subsetsampling}{Subset Sampling\xspace}
\newcommand{\subsetSequencesampling}{Sequence Sampling\xspace}
\newcommand{\ourAlgo}{\mathcal{A}}
\newcommand{\ourAlgoName}{\textsc{FastJoint}\xspace}
\newcommand{\joint}{\textsc{Joint}\xspace}
\newcommand{\jointMesm}{\mathcal{M}_{\joint}}
\newcommand{\NameExpoMesm}{\textsc{exp}}
\newcommand{\ExpoMesm}{\mathcal{M}_{\textsc{exp}}}

\newcommand{\SensitivityExpoMesm}{\Delta_{\textsc{exp}}}
\newcommand{\cpeel}{\textsc{CDP-Peel}\xspace}
\newcommand{\ppeel}{\textsc{PNF-Peel}\xspace}
\newcommand{\partialPermutation}[2]{\mathcal{P}_{ {#1}, {#2} } }

\newcommand{\hist}{\vec{h}}

\newcommand{\forwardSum}{\vec{\sigma}}
\newcommand{\cev}[1]{\reflectbox{\ensuremath{\vec{\reflectbox{\ensuremath{#1}}}}}}
\newcommand{\backwardSum}{\cev{\sigma}}

\newcommand{\cntUnified}[2]{\bar{C}_{#1, #2}}
\newcommand{\cnt}[2]{{C}_{#1, #2}}

\newcommand{\topkSoln}{\vec{s}}
\newcommand{\uniformSample}{\overset{\,\,r}{\longleftarrow}}

\newcommand{\DataDomain}{\mathcal{D}}
\newcommand{\IntSet}[2]{[{#1}\,.\,.\,{#2}]}
\newcommand{\indicator}[1]{\mathds{1}_{\left[#1\right]}}

\NewDocumentCommand\loss{g g g}{
  \ensuremath{
    {\IfNoValueTF{#3}
        { 
            \IfNoValueTF{#2} 
                {
                    \IfNoValueTF{#1}
                        {\cE}
                        {\cE_{#1}}
                }
                {\cE\paren{{#1, #2}}}
        }
        {
            \cE_{#1}\paren{{#2, #3}}
        }
    }
  }
}
\NewDocumentCommand\groupSeq{ m g }{
  \ensuremath{
    \IfNoValueTF{#2}
    {{\mathcal{S}}_{#1}}
    {{\mathcal{S}}_{#1, #2}}
  }
}
\NewDocumentCommand\unifiedGroupSeq{ m g }{
  \ensuremath{
    \IfNoValueTF{#2}
    {\bar{\mathcal{S}}_{#1}}
    {\bar{\mathcal{S}}_{#1, #2}}
  }
}
\NewDocumentCommand\DiscreteLapNoise{ m g }{
  \ensuremath{
    \IfNoValueTF{#2}
    {\operatorname{\mathbb{DL}ap}\left( {#1} \right)}
    {\operatorname{\mathbb{DL}ap}\left( {#1}, {#2} \right)}
  }
}
\NewDocumentCommand\TDiscreteLapNoise{ m g }{
  \ensuremath{
    \IfNoValueTF{#2}
    {\operatorname{\mathbb{TDL}ap}\left( {#1} \right)}
    {\operatorname{\mathbb{TDL}ap}\left( {#1}, {#2} \right)}
  }
}
\NewDocumentCommand\ExpNoise{ m g }{
  \ensuremath{
    \IfNoValueTF{#2}
    {\operatorname{\mathbb{E}xp}\left( {#1} \right)}
    {\operatorname{\mathbb{E}xp}\left( {#1}, {#2} \right)}
  }
}

\newif\ifcomment
\commenttrue

\definecolor{DarkGreen}{rgb}{0.1,0.5,0.1}
\ifcomment
\newcommand{\hao}[1]{\textcolor{blue}{[HAO: #1]}}

\else

\makeatletter
\newcommand{\hao}[1]{%
  \@bsphack
  \@esphack
}
\makeatother

\fi

\title{Faster Differentially Private Top-$k$ Selection: \hspace{2cm} A Joint Exponential Mechanism with Pruning}

\author{%
  Hao WU\thanks{This work was conducted while the author was a Postdoctoral Fellow at the University of Copenhagen.
  } \\
  University of Waterloo \\
    Canada\\
  \texttt{hao.wu1@uwaterloo.ca} \\
  \And
  Hanwen Zhang \\
  University of Copenhagen \\
  Denmark \\
  \texttt{hazh@di.ku.dk} 
}

\begin{document}

\maketitle

\begin{abstract}
    We study the differentially private top-$k$ selection problem,
    aiming to identify a sequence of $k$ items with approximately the highest scores from $d$ items. 
    Recent work by \citeauthor{GillenwaterJMD22} (ICML '$22$) employs a direct sampling approach from the vast collection of $d^{\,\Theta(k)}$ possible length-$k$ sequences, showing superior empirical accuracy compared to previous pure or approximate differentially private methods. 
    Their algorithm has a time and space complexity of $\Tilde{O}(dk)$. 

    \vspace{1mm}
    In this paper, we present an improved algorithm with time and space complexity 
    $
        O(d + k^2 / \eps \cdot \ln d)
        \footnote{A simplified bound from Theorem~\ref{theorem: main result} for a wide range of failure probabilities concerning solution quality.}
    $
    , where $\eps$ denotes the privacy parameter.
    Experimental results show that our algorithm runs orders of magnitude faster than their approach, while achieving similar empirical accuracy. 
    
\end{abstract}

\vspace{-2.5mm}
\section{Introduction}
\label{sec: introduction}

\vspace{-0.5mm}
Top-$k$ selection is a fundamental operation with a wide range of applications: \emph{search engines, e-commerce recommendations, data analysis, social media feeds etc}. 
Here, we consider the setting where the dataset consists of $d$ items evaluated by $n$ people.
Each person can cast at most one vote for each item, and vote for unlimited number of items. 
Our goal is to find a sequence of $k$ items which receives the highest number of votes.

Given that data can contain sensitive personal information such as medical conditions, browsing history, or purchase records, we focus on top-$k$ algorithms that are \emph{differentially private}~\citep{DworkMNS06}: 
it is guaranteed that adding/removing an arbitrary single person to/from the dataset does not substantially affect the output.
Research for algorithms under this model centers around how accurate the algorithms can be and how efficient they are.

Significant progress has been made in understanding the theoretical boundaries.
There are approximate differentially private algorithms~\citep{DurfeeR19, QiaoSZ21} that achieve asymptotic accuracy lower bound~\citep{BafnaU17, SteinkeU17}, and have $O(d)$ time and space usage.

There is also a research endeavor aimed at enhancing the empirical performance of the algorithms.
A particularly noteworthy one is the \joint mechanism by~\citet*{GillenwaterJMD22}, which exhibits best empirical accuracy across various parameter settings. 
Diverging from the prevalent \emph{peeling strategy} for top-$k$ selection--wherein items are iteratively selected, removed and repeated $k$ times--the \joint mechanism considers the sequence holistically, directly selecting an output from the space comprising all $d^{\,\Theta(k)}$ possible length-$k$ sequences.

While the algorithm has running time and space $\tilde{O}(dk)$, successfully avoiding an exponential time or space consumption, it notably incurs a higher computational cost than its $O(d)$ counterparts. 
This prompts the interesting question:
\begin{quote}
    \it
    {Research Question:}
    Can we design a mechanism equivalent to the \joint mechanism with running time and space linear in $d$?
\end{quote}

\paragraph{Our Contributions.}
Our paper answers the research question when $k$ is not too large.
Specifically,
\vspace{-1mm}
\begin{itemize}[leftmargin=4.5mm, topsep=2pt, itemsep=2pt, partopsep=2pt, parsep=2pt]
    \item We present an improved algorithm with time and space complexity of
    $
        O(d + k^2 / \eps \cdot \ln d)
    $
\end{itemize}

This is an informal statement of Theorem~\ref{theorem: main result}.
When $k \in O( \sqrt{d} )$ (a common scenario in practical settings), the time and space complexity simplifies to $\Tilde{O}(d)$.
Moreover, the proposed algorithm achieves the same asymptotic accuracy guarantee as the \joint mechanism.

Similar to the \joint mechanism, our algorithm is an instance of the exponential mechanism (detailed in Section~\ref{sec: preliminary}) that directly samples from the output space comprising all length-$k$ sequences. 
We introduce a "group by" sampling framework, which partitions the sequences in the output space into $O(nk)$ subsets, aiming to streamline the sampling process. 
The framework consists of two steps: sampling a subset and then sampling a sequence from that subset.
We provide efficient algorithms for both steps. 
Furthermore, we introduce a pruning technique to handle outputs with low accuracy uniformly. 
This technique effectively reduces the number of subsets to $\tilde{O}(k^2)$, leading to an algorithm in $\Tilde{O}\paren{d + k^2}$ time and space complexity. 

Finally, we perform extensive experiments to
\begin{itemize}[leftmargin=4.5mm, topsep=2pt, itemsep=2pt, partopsep=2pt, parsep=2pt]
    \item Verify the theoretical analysis of our algorithm.
    \item Demonstrate that our algorithm runs 10-100 times faster than \joint on the tested datasets.
    \item Show that our algorithm maintains comparable accuracy to \joint. 
\end{itemize}

\paragraph{Organization.} 
Our paper is structured as follows: Section~\ref{sec: problem definition} formally introduces the problem, while Section~\ref{sec: preliminary} delves into the necessary preliminaries for our algorithm. Section~\ref{sec: algorithm} introduces our novel algorithm, and Section~\ref{sec: experiments} presents our experiment results.

\vspace{-1mm}
\section{Problem Description}
\label{sec: problem definition}

Let $\DataDomain \doteq \set{1, \ldots, d}$ be a set of~$d$ items and~$\cU \doteq \set{1, \ldots, n}$ be a set of~$n$ clients.
Each client~$v \in \cU$ can cast at most one vote for each item, and can vote for an unlimited number of items.
For each item $i \in \DataDomain$, its score $\hist[i]$ is the number of votes it received.
The \emph{histogram} is a vector~$\hist \doteq \paren{ \hist[1], \ldots, \hist[d] } \in \IntSet{0}{n}^d$. 
Define $\partialPermutation{\DataDomain}{k} \doteq \set{ \paren{i_1, \ldots, i_k} \in \DataDomain^k : i_1, \ldots i_k \, \text{ are distinct} }$ be the collection of all possible length-$k$ sequences. 

The differentially private top-$k$ selection problem aims at finding a sequence from $\partialPermutation{\DataDomain}{k}$ with approximately largest scores, while protecting the privacy of each individual vote. 

\textit{Privacy Guarantee.}
Two voting histograms $\hist, \hist' \in \N^d$ are neighboring, denoted by~$\hist \sim \hist'$, if $\hist'$ can be obtained from $\hist$ by adding or removing an arbitrary individual's votes.
Therefore, when $\hist \sim \hist'$, we have $|| \hist - \hist' ||_\infty \le 1,$ and $\hist \le \hist'$ or $\hist \ge \hist'$.
To protect personal privacy, a top-$k$ selection algorithm should have similar output distributions on neighboring inputs. 

\begin{definition}[$\paren{\eps, \delta}$-Private Algorithm~\citep{DR14}] \label{def: Differential Privacy}
    Given~$\eps, \delta > 0$, a randomized algorithm 
    $\cM: \N^d \rightarrow \partialPermutation{\DataDomain}{k}$
    is called~$\paren{\eps, \delta}$-differentially private (DP),
    if for every~$\hist, \hist' \in \N^d$ such that~$\hist \sim \hist'$, 
    and all $Z \subseteq \partialPermutation{\DataDomain}{k}$,
    \vspace{-2mm}
    \begin{equation} \label{ineq: def private algo}
        \begin{array}{c}
            \Pr[ \cM (\hist) \in Z ] \le e^\eps \cdot \Pr [ \cM (\hist') \in Z ] + \delta\,.
        \end{array}    
    \end{equation}
\end{definition}

\textit{Remark:} 
An algorithm~$\cM$ is also called~$\eps$-DP for short, if it is~$\paren{\eps, 0}$-DP. 
If an algorithm is~$\eps$-DP, it is also called \emph{pure DP}, whereas it is called \emph{approximate DP} if it is $\paren{\eps, \delta}$-DP.
Although we present the definition in the context of top-$k$ selection algorithms, it applies more generally to any randomized algorithms $\cM: \cX \rightarrow \cY$, where $\cX$ is the input space, which is associated with a symmetric relation $\sim$ that defines neighboring inputs.

\vspace{-2mm}
\section{Preliminaries}
\label{sec: preliminary}

\subsection{Exponential Mechanism} 

The exponential mechanism~\citep{McSherryT07} is a well-known differentially private algorithm for publishing discrete values. 
Given a general input space $\cX$ (associated with a relation $\sim$ which defines neighboring datasets), a finite output space $\cY$, the exponential mechanism~$\ExpoMesm: \cX \rightarrow \cY$ is a randomized algorithm given by 
\begin{align}
    \P{\ExpoMesm(x) = y} 
        \propto \exp \bigparen{ - \eps \cdot \
        \loss{\NameExpoMesm}{x}{y} \, / \, \paren{ 2 \cdot \SensitivityExpoMesm } }, 
    &\quad  \forall x \in \cX, \, y \in \cY,
\end{align}
where $\loss{\NameExpoMesm}: \cX \times \cY \rightarrow \R$ is called the \emph{loss function} measuring how ``bad'' $y$ is when the input is $x$, and 
$\SensitivityExpoMesm$ is the \emph{sensitivity} of $\loss{\NameExpoMesm}$ 
which is the maximum deviation of $\loss{\NameExpoMesm}$:
\begin{align}
    \SensitivityExpoMesm \doteq \max_{x \sim x', y \in \cY} \card{
        \loss{\NameExpoMesm}{x}{y} - \loss{\NameExpoMesm}{x'}{y}
    }.
\end{align}

\begin{fact}[Privacy~\citep{McSherryT07}]
    The exponential mechanism $\ExpoMesm$ is $\eps$-DP.
\end{fact}

\begin{fact}[Utility Guarantee~\citep{McSherryT07}]
    \label{fact: utility of exponential mechanism}
    For each $\beta \in \paren{0, 1}$, and
    $
        \tau \doteq \frac{2 \cdot \SensitivityExpoMesm}{\eps} \cdot \ln \frac{\card{\cY}}{\beta}, 
    $
    the exponential mechanism $\ExpoMesm$ 
    satisfies
    \begin{equation*}
        \begin{array}{cc} 
            \P{
                \loss{\NameExpoMesm}{x}{\ExpoMesm(x)} \ge { \min_{y \in \cY} \, \loss{\NameExpoMesm}{x}{y} } + \tau 
            } \le \beta,
            &
            \forall x \in \cX. 
        \end{array}
    \end{equation*}
\end{fact}

{\it Implementation.}
Given input $x \in \cX$, a technique for implementing the exponential mechanism is to add i.i.d. Gumbel noises to the terms of $\set{ - \eps \cdot \loss{\NameExpoMesm}{x}{y} \, / \, \paren{ 2 \cdot \SensitivityExpoMesm } : y \in \cY}$, and then select the $y$ corresponding to the noisy maximum.

\begin{definition}
    Given $b > 0$, the Gumbel distribution with parameter $b$, denoted by $\GumbelNoise{b}$, has probability density function
    $
        p(x) = \frac{1}{b} \cdot \exp \PAREN{ - \PAREN{ \frac{x}{b} + \exp \PAREN{ - \frac{x}{b} } }  }, \,
        \forall x \in \R.
    $
\end{definition}

\begin{fact}[\citep{YELLOTT1977109}]
    \label{fact: sampling based on grumbel noisy max}
    Assume that $w_i \ge 0,$ for $i \in [m]$, 
    and $X_i \sim \GumbelNoise{1}, i \in [m]$ are independent random variables.
    Then 
    $       
        \Pr \big[
            i = \argmax_{j \in [m]} \PAREN{
                X_j + \ln w_j
            }
        \big] 
        \propto  
        w_i.
    $
\end{fact}

It follows that, if $X_y \sim \GumbelNoise{1}, y \in \cY$ are independent random variables, then 
\begin{equation*} 
    \begin{array}{c}
        \P{
            y = \argmax_{y' \in \cY} \set{
                X_{y'} - \eps \cdot \loss{\NameExpoMesm}{x}{y'} \, / \, \paren{ 2 \cdot \SensitivityExpoMesm }
            } 
        } 
        \propto
        \exp \bigparen{ -\eps \cdot \loss{\NameExpoMesm}{x}{y} \, / \, \paren{ 2 \cdot \SensitivityExpoMesm } }.
    \end{array}
\end{equation*}

\subsection{\joint Mechanism} 
The \joint mechanism $\jointMesm: \N^d \rightarrow \partialPermutation{\DataDomain}{k}$~\citep{GillenwaterJMD22} is an instance of the exponential mechanism which samples a sequence $\topkSoln = \paren{\topkSoln[1], \ldots, \topkSoln[k]}$ directly from $\partialPermutation{\DataDomain}{k}$, 
with the loss function 
\begin{align}
    \label{eq: definition of error for joint}
    \begin{array}{c}
        \loss{\joint}{\hist}{\topkSoln} 
            \doteq \max_{i \in [k]} \bigparen{ \hist_{(i)} - \hist\BigBracket{\topkSoln[i]} },
    \end{array}
\end{align}
where $\hist_{(i)}$ is the true $i^{(th)}$ largest entry in $\hist$.
It can be seen that $\loss{\joint}(\cdot)$ has sensitivity $\Delta_\joint = 1$.

Observe that a naive implementation of this exponential mechanism needs evaluating and storing the scores of $\card{\partialPermutation{\DataDomain}{k}} = d^{\, \Omega(k)}$ sequences. 
Remarkably, \citet*{GillenwaterJMD22} demonstrate that the exponential time and space requirements can be reduced to polynomial.

\begin{fact}[\joint Mechanism~\citep{GillenwaterJMD22}]
    \label{fact: joint mechanism}
    There is an implementation of exponential mechanism which directly sample a sequence from $\partialPermutation{\DataDomain}{k}$ according to loss function $\loss{\joint}{\hist}{\topkSoln} = \max_{i \in [k]} \big( \hist_{(i)} - \hist\BigBracket{\topkSoln[i]})$
    with time $O(dk \log k + d \log d)$ time and space~$O(dk)$.
\end{fact}

For completeness, we includes a short proof of Fact~\ref{fact: joint mechanism} in Appendix~\ref{appendix: missing proofs}. 
Let $\topkSoln\,^*$ corresponds to the $k$ items with the largest scores. 
Then clearly $\min_{\topkSoln} \, \loss{\joint}{\hist}{\topkSoln} = \loss{\joint}{\hist}{\topkSoln\,^*} = 0.$
Combining $\card{\partialPermutation{\DataDomain}{k}} = \binom{d}{k} \cdot k!$ and $\Delta_\joint = 1$,
and applying Fact~\ref{fact: utility of exponential mechanism}, provide the theoretic utility guarantee of \joint.

\begin{fact}[Utility Guarantee]
    \label{fact: utility of joint mechanism}
    For each $\beta \in \paren{0, 1}$,  
    $
        \tau \doteq 
        \Big\lceil \frac{2}{\eps} \cdot \ln \frac{\binom{d}{k} \cdot k!}{\beta} 
        \Big\rceil
        \in \Theta \bigparen{
            \frac{k}{\eps} \cdot (k\ln d + \ln \frac{1}{\beta}) 
        }, 
    $
    \vspace{-2mm}
    \begin{equation}
        \begin{array}{c}
            \Pr[ \loss{\joint}{\hist}{\jointMesm(\hist)} \ge \tau ] \le \beta.
        \end{array}
    \end{equation}
    \vspace{-3mm}
\end{fact}

\section{Algorithm}
\label{sec: algorithm}

In this section, we present an algorithm which has similar output distribution as the~\joint mechanism~\citep{GillenwaterJMD22}, but reduces the time and space complexity to $O(d + k \cdot \tau)$.
The main result is stated as follows. 

\begin{theorem}
    \label{theorem: main result}
    Let $\beta \in (0, 1)$, 
    $
        \tau \doteq 
        \Big\lceil \frac{2}{\eps} \cdot \ln \frac{\binom{d}{k} \cdot k!}{\beta} 
        \Big\rceil,
    $ 
    and~$\ourAlgo : \N^d \rightarrow \partialPermutation{\DataDomain}{k}$ be the top-$k$ algorithm 
    s.t.
    \begin{align}
        \begin{array}{c}
            \Pr\big[ \ourAlgo(\hist) = \topkSoln \big] 
                \propto \exp \bigparen{ - \eps \cdot \loss{\ourAlgo}{\hist}{\topkSoln} \, / \, \paren{ 2 \cdot \Delta_\cA } },
        \end{array}
    \end{align}
    where 
    $
        \loss{\ourAlgo}{\hist}{\topkSoln} 
            \doteq \min \paren{
                \loss{\joint}{\hist}{\topkSoln}, 
                \tau
            },
    $
    and $\Delta_\ourAlgo$ is the sensitivity of $\loss{\ourAlgo}$.
    Then $\ourAlgo$ is $\eps$-DP and has an implementation with time and space complexity $O(d + k \cdot \tau)$.
    It satisfies the following condition:
    \begin{equation}
        \begin{array}{c}
            \Pr[ \loss{\joint}{\hist}{\ourAlgo(\hist)} \ge \tau ] \le \beta.
        \end{array}
    \end{equation}
\end{theorem}

{\it Simplification.} When ${1} / {\beta} \in O\PAREN{ d^k }$, the error $\tau$ reduces to 
$
    O(d + k^2 / \eps \cdot \ln d).
$

It can be verified that the sensitivity $\Delta_\ourAlgo = 1.$
Since $\ourAlgo$ is also an instance of the exponential mechanism, its privacy guarantee naturally follows from this property. 
As for the utility guarantee, instead of expressing it in terms of $\loss{\ourAlgo}$ (and directly applying Fact~\ref{fact: utility of exponential mechanism}), we express it in terms of the loss function $\loss{\joint}$. 
The error achieved mirrors that of \joint, as stated in Corollary~\ref{fact: utility of joint mechanism}.

\vspace{2mm}
{\it Road map.} 
In Section~\ref{subsec: algorithm overview}, we first propose a novel "group by" framework for sampling a sequence from the space of all possible length-$k$ sequences. 
Then, we introduce a new choice of groups to materialize this framework, leading to a new algorithm with $O(d + nk)$ time complexity. 
In Section~\ref{subsec: pruning}, we propose a new pruning technique that reduces the running time to $O(d + \tau k)$.

A detailed comparison with \joint is deferred to Section~\ref{sec: related-work}, where we will explain \joint in the context of the novel framework and compare it with our new algorithm.

\emph{
To simplify notation, when the input histogram is clear from the context, we use $\loss(\topkSoln)$ as shorthand for $\loss{\joint}{\hist}{\topkSoln}$, and $\loss{\ourAlgo}(\topkSoln)$ as shorthand for $\loss{\ourAlgo}{\hist}{\topkSoln}$.
}

\subsection{Sampling Framework}
\label{subsec: algorithm overview}

\paragraph{Partitioning.}
We begin with a novel framework for designing algorithms that produces the same output distribution as \joint.
Let $P_1, \ldots, P_m$ be an arbitrary partition of $\partialPermutation{\DataDomain}{k}$.
It is called \emph{$\loss$-consistent}, if all sequence belonging the same subset have the same loss (w.r.t loss function $\loss$): 
$
    \forall i \in [m], \, \forall \topkSoln, \topkSoln\,' \in P_i, 
    \, 
    \loss(\topkSoln) = \loss(\topkSoln\,').
$
We regard $\loss(P_i)$ as the loss of the sequences in $P_i$.
Given this partition, we can design an algorithm that reproduces the output distribution of \joint using a two-step approach:

\begin{description}
    \item[\hspace{3mm}\textbf{\textit{\subsetsampling}}:] 
        sample a subset $P_i$ with probability proportional to
          $
            |P_i| \cdot \exp \PAREN{ - \eps \cdot \loss(P_i) / 2}. 
          $
    \item[\hspace{3mm}\textbf{\textit{\subsetSequencesampling}}:] sample an $\topkSoln \in P_i$ uniformly.
\end{description}

There can be more than one choices of partitions of $\partialPermutation{\DataDomain}{k}.$
We would like to find one which enables efficient sampling algorithms for both steps. 
We first consider the partition $\groupSeq{r}{i}, r \in \IntSet{0}{n}, i \in [k]$, given by: 
\vspace{-2mm}
\begin{align}
    \begin{array}{c}
        \groupSeq{r}{i}
            \doteq \set{ 
                \topkSoln = \paren{\topkSoln[1], \ldots, \topkSoln[k]} \in \partialPermutation{\DataDomain}{k} : 
                    \, \loss(\topkSoln) = r  
                    \text{ and } \, 
                    \begin{matrix}
                        \hist\BigBracket{\topkSoln[j]} > \hist_{(j)} - r, & \forall j < i \\
                        \hist\BigBracket{\topkSoln[j]} = \hist_{(j)} - r, & j = i \\            
                        \hist\BigBracket{\topkSoln[j]} \ge \hist_{(j)} - r, & \forall j > i \\
                    \end{matrix}
            }. 
    \end{array}
\end{align}
Based on the definition of~$\loss$ in Equation~\eqref{eq: definition of error for joint}, $\groupSeq{r}{i}$ represents the subset of sequences with an loss equal to $r$, and $i$ denotes the index of the first coordinate reaching this loss.
Hence, $\loss (\groupSeq{r}{i}) = r$.
Via Fact~\ref{fact: sampling based on grumbel noisy max}, to sample an $\groupSeq{r}{i}$ with probability proportional to
$
    |\groupSeq{r}{i}| \cdot \exp \PAREN{ - \eps \cdot r / 2},  
$
we can compute the maximum of $\set{X_{r, i} + \ln \bigparen{|\groupSeq{r}{i}| \cdot \exp \PAREN{ - \eps \cdot r / 2}} : (r, i) \in \IntSet{0}{n}\times[k]}$, where $X_{r, i} \sim \GumbelNoise{1}$.

\emph{The first key advantage} of the partition being discussed is that, each $\ln |\groupSeq{r}{i}|$ can be expressed as a sum of $k$ terms and it can be computed efficiently.

\begin{definition}
    For each $r \in \IntSet{0}{n}, j \in [k]$, define 
    $
        \cnt{r}{j} \doteq 
            | \{ j' \in [d] : \hist\bracket{j'} \ge \hist_{(j)} - r \} |.
    $
\end{definition}

\begin{lemma}
    \label{lemma: formula for the group size}
    For each $r \in \IntSet{0}{n}, i \in [k]$, it holds that  
    \begin{equation} 
        \label{eq: formula for the group size}
        \begin{array}{c}
            \ln |\groupSeq{r}{i}| 
                = \sum_{j = 1}^{i - 1} \ln \PAREN{ \cnt{r - 1}{j} - (j - 1) } 
                + \ln \big( \cnt{r}{i} -  \cnt{r - 1}{i} \big)
                + \sum_{j = i + 1}^{k} \ln \PAREN{ \cnt{r}{j}  - (j - 1) }.
        \end{array}
    \end{equation}
\end{lemma}

\begin{lemma}
    \label{lemma: computing formula for the group size}
    For all $r \in \IntSet{0}{n}, j \in [k]$, $\cnt{r}{j}$ can be computed in $O(d + nk)$ time.
    Furthermore, given the $\cnt{r}{j}$'s, for all $r \in \IntSet{0}{n}$, $\ln |\groupSeq{r}{i}|$ can be computed in $O(nk)$ time.
\end{lemma}

The algorithms for proving Lemma~\ref{lemma: computing formula for the group size} are detailed in Appendix~\ref{appendix: implementation}. 
At a high level, for a fixed \(r\), the \(\cnt{r}{j}, j \in [k]\) constitute a monotone sequence, enabling us to devise a recursion formula to compute them. 
Additionally, the prefix sums (the first term) and the suffix sums (the last term) in Equation~\eqref{eq: formula for the group size} can be pre-computed, simplifying the computation of \( \ln |\groupSeq{r}{i}|\) to adding only three terms.

Here, we present a proof for Lemma~\ref{lemma: formula for the group size}, offering insights into the structure of $\groupSeq{r}{i}$.

\begin{proof}[Proof for Lemma~\ref{lemma: formula for the group size}]
    It suffices to show that 
    \begin{equation}
        \label{eq: formula for the un-log group size}
        \begin{array}{c}
            |\groupSeq{r}{i}| 
                = {
                    \prod_{j = 1}^{i - 1} \PAREN{ \cnt{r - 1}{j} - (j - 1) }  \cdot \big( \cnt{r}{i} -  \cnt{r - 1}{i} \big)
                    \cdot \prod_{j = i + 1}^{k} \PAREN{ \cnt{r}{j}  - (j - 1) }.
                } 
        \end{array}
    \end{equation}
    The proof is via standard counting argument: assume we want to select a sequence $\topkSoln \in \groupSeq{r}{i}$. 
    Since $\topkSoln[1] \in \{ j' \in [d] : \hist\bracket{j'} > \hist_{(1)} - r \}$, 
    the number of possible choices for $\topkSoln[1]$ is 
    $$
        | \{ j' \in [d] : \hist\bracket{j'} > \hist_{(1)} - r \} |
        = | \{ j' \in [d] : \hist\bracket{j'} \ge \hist_{(1)} - (r - 1) \} |
        = \cnt{r - 1}{1}.
    $$
    The first equality holds because the $\hist[j']$ values are integers.
    
    Next, since $\hist_{(1)} \ge \hist_{(2)}$, it also holds that $\topkSoln[1] \in \{ j' \in [d] : \hist\bracket{j'} > \hist_{(2)} - r \}$. 
    After determining $\topkSoln[1]$, the number of choices for $\topkSoln[2]$ is $|\{ j' \in [d] : \hist\bracket{j'} > \hist_{(2)} - r \}| - 1 = \cnt{r - 1}{2} - 1$.
    Continuing this argument, for each $j \in \IntSet{1}{i - 1}$, the number of choices for $\topkSoln[j]$, after determining $\topkSoln\IntSet{1}{j-1}$, is $|\{ j' \in [d] : \hist\bracket{j'} > \hist_{(j)} - r \}| - (j - 1) = { \cnt{r - 1}{j} - (j - 1) }$.
    
    Now we consider the number of choices for $\topkSoln[i]$.
    Since $\topkSoln[1], \ldots, \topkSoln[i - 1] \in \{ j' \in [d] : \hist\bracket{j'} > \hist_{(i)} - r \}$, they do not appear in $\{ j' \in [d] : \hist\bracket{j'} = \hist_{(i)} - r \}$. 
    The number of choices for $\topkSoln[i]$ is exactly 
    $
        |\{ j' \in [d] : \hist\bracket{j'} = \hist_{(i)} - r \}|
            = \cnt{r}{i} -  \cnt{r - 1}{i}.
    $
    
    The cases for $j \in \IntSet{i + 1}{k}$ are similar to the cases of $j \in \IntSet{1}{i - 1}$. 
    Since $\hist_{(1)} \ge \hist_{(2)} \ge \cdots \ge \hist_{(j - 1)}$, 
    it holds that $\topkSoln[1], \ldots, \topkSoln[j - 1] \in \{ j' \in [d] : \hist\bracket{j'} \ge \hist_{(j)} - r \}$. 
    As a result, for $j \in \IntSet{i + 1}{k}$, the number of choices for $\topkSoln[j]$, after determining $\topkSoln\IntSet{1}{j-1}$, is ${ \cnt{r}{j} - (j - 1) }$. 

    Multiplying the number of choices for each element in $\topkSoln \in \groupSeq{r}{i}$, we obtain Equation~\eqref{eq: formula for the un-log group size}.
\end{proof}

\emph{The second key advantage} of the partition being considered is that, there is an algorithm for sampling a uniform random sequence from $\groupSeq{r}{i}$ in $O(d)$ time, as implicitly suggested by the proof for Lemma~\ref{lemma: formula for the group size}. 
Further details of this implementation are provided in Appendix~\ref{appendix: implementation}.

\subsection{Pruning}
\label{subsec: pruning}
The previous discussion suggests an new algorithm with $O(d + nk)$ running time. 
Based on Lemma~\ref{lemma: computing formula for the group size}, computing the $\cnt{r}{j}$ values for $r \in \IntSet{0}{n}$ and $j \in [k]$ takes $O(d + nk)$ time. 
The total time to compute $\ln |\groupSeq{r}{i}|$ for $r \in \IntSet{0}{n}$ and $i \in [k]$ is $O(nk)$. 
Finally, sampling a sequence from the chosen $\groupSeq{r}{i}$ takes $O(d)$ time.
The bottleneck here lies in the $nk$ term, which arises from the need to compute the $\cnt{r}{j}$'s and $\ln |\groupSeq{r}{i}|$'s for all $r \in \IntSet{0}{n}$.

However, this is unnecessary. 
We need only to consider the cases for $r \in \IntSet{0}{\tau}$.
The key observation is that, the probability of sampling an $\topkSoln \in \partialPermutation{\DataDomain}{k}$ %
decreases exponentially with increasing $\loss(\topkSoln)$. 
\begin{claim}[Restatement of Fact~\ref{fact: utility of joint mechanism}]
    The probability of sampling an $\topkSoln$ with $\loss(\topkSoln) \ge \tau$ is at most $\beta$.
\end{claim}
To provide further insight, we present a short proof here.
Let $\mathcal{S}_{\ge \tau} \doteq \{\topkSoln \in \partialPermutation{\DataDomain}{k} : \loss(\topkSoln) \ge \tau\}$, then 
\begin{equation}
    \label{ineq: tail probability}
    \begin{array}{c}
        \P{ \text{sampling an } \topkSoln \in \mathcal{S}_{\ge \tau}}
            \le \frac{ \sum_{\topkSoln \in \mathcal{S}_{\ge \tau}}
                e^{-\eps \cdot \loss(\topkSoln) / 2}
            }{
                e^{-\eps \cdot \loss(\topkSoln\,^*) / 2}
            }
            \le |\partialPermutation{\DataDomain}{k}| \cdot e^{-\eps \cdot \tau / 2}
            = \beta,
    \end{array}
\end{equation}
where $\topkSoln\,^*$ is the $k$ items with the largest scores and $\loss(\topkSoln\,^*) = 0$.

Given this, if we slightly adjust the loss function of sequences in $\groupSeq{\ge \tau}$, their probabilities of being outputted will not be significantly affected.
It motivates to consider the truncated loss function: 
$\loss{\ourAlgo}(\topkSoln)
    \doteq \min \paren{
        \loss(\topkSoln), 
        \tau
    },
$
and an algorithm $\ourAlgo$ which samples an $\topkSoln$ with probability proportional to $e^{-\eps \cdot \loss{\ourAlgo}(\topkSoln) / 2}$.
As inequality~\eqref{ineq: tail probability} still holds if we the $\loss(\cdot)$ with $\loss{\ourAlgo}(\cdot)$, we immediately obtain the following lemma.

\begin{lemma}
    \label{lemma: utility guarantee of our algorithm}
    The probability of $\ourAlgo$ sampling an $\topkSoln$ with $\loss(\topkSoln) \ge \tau$ is at most $\beta$. 
\end{lemma}

{\it Subset Merging.}
The most important benefit of truncated loss is that, it allows us to reduce to the number of subsets in the partition $\groupSeq{r}{i}, r \in \IntSet{0}{n}, i \in [k]$ from $O(nk)$ to $O(\tau k)$.
In particular, for each $i \in [k]$, as the sequences in the subsets $\groupSeq{r}{i}, r \in \IntSet{\tau}{n}$ has the same truncated loss, it suffices to merge them into one 
\vspace{-2mm}
\begin{equation}
    \label{eq: def of tail groups}
    \begin{array}{c}
        \groupSeq{\ge \tau}{i} \doteq \cup_{r \in \IntSet{\tau}{n}} \groupSeq{r}{i}
        = \set{ 
            \topkSoln = \paren{\topkSoln[1], \ldots, \topkSoln[k]} \in \partialPermutation{\DataDomain}{k} : 
                \begin{matrix}
                    \hist\BigBracket{\topkSoln[j]} > \hist_{(j)} - \tau, & \forall j < i \\
                    \hist\BigBracket{\topkSoln[j]} \le \hist_{(j)} - \tau, & j = i
                \end{matrix}
        }.
    \end{array}
\end{equation}
$\groupSeq{\ge \tau}{i}$ shares a similar formula on its size as Equation~\eqref{eq: formula for the group size} and can be uniformly sampled efficiently. 
\begin{lemma}
    \label{lemma: formula for the merged group size}
    For each $i \in [k]$, we have 
    \begin{equation}
        \label{eq: formula for the merged group size}
        \begin{array}{c}
            \ln |\groupSeq{\ge \tau}{i}| 
                = \sum_{j = 1}^{i - 1} \ln \PAREN{ \cnt{\tau - 1}{j} - (j - 1) } 
                + \ln \big( d -  \cnt{r - 1}{i} \big)
                + \sum_{j = i + 1}^{k} \ln \PAREN{ d  - (j - 1) }.
        \end{array}
    \end{equation} 
\end{lemma}
\begin{lemma}
    \label{lemma: time for the merged group size}
    Given the $\cnt{\tau - 1}{j}$'s, each $\ln |\groupSeq{\ge \tau}{i}|$ can be computed in $O(1)$ amortized time.
\end{lemma}

The proof for Lemma~\ref{lemma: formula for the merged group size} is provided in Appendix~\ref{appendix: missing proofs}, while an algorithmic proof for Lemma~\ref{lemma: time for the merged group size} can be found in Appendix~\ref{appendix: implementation}.

\vspace{-2mm}
\section{Experiment}
\label{sec: experiments}
\vspace{-2mm}

In this section, we compare our algorithm, referred to as \ourAlgoName, with existing state-of-the-art methods on real-world datasets. 
Our Python implementation is available publicly.\footnote{\url{https://github.com/wuhao-wu-jiang/Differentially-Private-Top-k-Selection}}

\vspace{-2mm}
\paragraph{Datasets.}
We utilize six publicly available datasets: 
Games (Steam video games with purchase counts)~\citep{T16}, 
Books (Goodreads books with review counts)~\citep{S19}, 
News (Mashable articles with share counts)~\citep{misc_online_news_popularity_332},
Tweets (Tweets with like counts)~\citep{DVN/JBXKFD_2017},
Movies (Movies with rating counts)~\citep{HarperK16} and
Foods (Amazon grocery and gourmet foods with review counts)~\citep{McAuleyTSH15}.
Table~\ref{tab:datasets} summarizes their sizes.

\vspace{-2mm}
\newcolumntype{C}{>{\centering\arraybackslash}X}
\begin{table}[!ht]
    \centering
    \begin{tabularx}{0.93\textwidth}{|C|C|C|C|C|C|C|}
    \hline
    Dataset      & Games & Books & News  & Tweets & Movies & Food   \\ \hline
    \#items      & 5,155  & 11,126 & 39,644 & 52,542  & 59,047  & 166,049 \\ \hline
    \end{tabularx}
    \vspace{1mm}
    \caption{Dataset Size Summary}
    \label{tab:datasets}
    \vspace{-7mm}
\end{table}

\paragraph{Baselines.}
Apart from the \joint mechanism~\citep{GillenwaterJMD22}, we consider the following two candidates: the peeling variant of permute-and-flip mechanism~\citep{McKennaS20}, denoted \ppeel; and the peeling exponential mechanism~\citep{DurfeeR19}, denoted \cpeel.
We don't compare with other mechanisms, e.g. the Gamma mechanism~\citep{SteinkeU16} and the Laplace mechanism~\citep{BLST10, QiaoSZ21}, which are empirically dominated by \ppeel and \cpeel respectively~\citep{GillenwaterJMD22}.

{\it \ppeel:} The permute-and-flip is an \(\eps\)-DP mechanism for top-1 selection.
It can be implemented equivalently by adding exponential noise (with privacy budget \(\eps / k\)) to each entry of \(\hist\) and reporting the item with the highest noisy value \citep{Ding20}.
To report \(k\) items, we use the \emph{peeling} strategy: select one item using the mechanism, remove it from the dataset, and repeat this process \(k\) times, resulting in a running time of \(O(dk)\).

{\it \cpeel:} The \((\eps, \delta)\)-DP peeling exponential mechanism samples \(k\) items without replacement, selecting one item at a time using a privacy budget of \(\tilde{O}(\eps / \sqrt{k})\) according to the exponential mechanism \citep{McSherryT07}. \citet{DurfeeR19} demonstrate that \cpeel has an equivalent \(O(d)\)-time implementation.

The code for all competing algorithms was obtained from publicly accessible GitHub repository by Google Research\footnote{
    \url{https://github.com/google-research/google-research/tree/master/dp_topk}
},
written in Python.

\vspace{-2mm}
\paragraph{Experiment Setups.}
The experiments are conducted on macOS system with M2 CPU and 24GB memory.
We compare the algorithms in terms of running time and error for different values of $k$, $\eps$ and $\beta$.
Note that the parameter $\beta$ (see Theorem~\ref{theorem: main result}) only affects our algorithm. 
The $(\eps, \delta)$-DP mechanism \cpeel is configured with a $\delta$ parameter of $10^{-6}$, consistent with prior research~\citep{GillenwaterJMD22}.

{\it Error Metrics.} 
We evaluate the quality of a solution~$\topkSoln$ using both $\ell_\infty$ and $\ell_1$ errors. The $\ell_\infty$ error is defined as $\max_{i \in [k]} | \hist_{(i)} - \topkSoln[i] |$, while the $\ell_1$ error is given by $\sum_{i \in [k]} | \hist_{(i)} - \topkSoln[i] |$.

{\it Parameter Ranges.}
The parameter ranges tested are as follows:
\begin{equation*}
    \begin{array}{ccc}
        k = 10, 20, \ldots, \underline{100}, \ldots, 200, & \quad
        \eps = 1 / 4, 1 / 2, \underline{1}, 2, 4, & \quad
        \beta = 2^{-6}, 2^{-8}, \underline{2^{-10}}, 2^{-12}, 2^{-14} .
    \end{array}
\end{equation*}
The values indicated by underlining represent the default settings. 
During experiments where one parameter is varied, the other two parameters are kept at their default values.

\vspace{-2mm}
\subsection{Results}

All experiments are repeated 200 times. 
Each figure displays the median running time or $\ell_\infty$ or $\ell_1$ error as the center line, with the shaded region spanning the 25th to the 75th percentiles.

\vspace{-2mm}
\paragraph{Varying $k$.}

Figure~\ref{fig: complete k results} presents the results for different values $k$.
\ourAlgoName consistently outperforms \joint in terms of execution speed, running 10 to 100 times faster across various datasets.
\ourAlgoName is slower than \cpeel; the later has theoretical time complexity $O(d)$ and therefore this is expected. 

We observe ``jumps'' in running time of \joint on the \emph{games} and \emph{food} datasets.
This phenomenon can also be found in the original work by \citet{GillenwaterJMD22} in the only running time plot for the \emph{food} dataset.
Upon investigation, we found that, as noted in their code comments, the current Python implementation of the \emph{\subsetSequencesampling} step of \joint has a worst-case time complexity of \(O(dk^2)\) instead of \(O(dk)\).
Although this step does not constitute a bottleneck in their code and accounts for a constant fraction of the total running time,  it still introduces instability in the running time. 
To delve deeper into this issue, we provide a comparison in the appendix where we plot the running time of \joint (excluding the \emph{\subsetSequencesampling} step) against the running time of \ourAlgoName (including the \emph{\subsetSequencesampling} step), resulting in smoother time plots. 
Even with this adjustment, \joint remains order of magnitude slower.

Interestingly, for small datasets, \ourAlgoName can be slower than \ppeel, which has an $O(dk)$ time complexity. 
This is because \ppeel has a simple algorithmic structure that can be implemented as $k$ rounds of vector operations: each round involves adding a noisy vector to the input histogram and then selecting an item (not previously selected) with the highest score. 
It is well-known that vectorized implementations\footnote{In Appendix~\ref{appendix: implementation}, we discuss in more detail the possibility of implementing \ourAlgoName with vertoziation.
} gain significant speed boosts by utilizing dedicated Python numerical libraries such as NumPy~\citep{harris2020array}. 
However, as the dataset size increases (e.g., the food dataset), \ourAlgoName outperforms \ppeel in terms of speed.

In terms of solution quality, even with the pruning strategy, \ourAlgoName does not experience quality degradation compared to \joint. 
It delivers similar performance to \joint across all datasets and performs particularly well on the \emph{books}, \emph{news}, \emph{tweets}, and \emph{movies} datasets, where there are large gaps between the top-$k$ scores. (Due to space limitations, the complete plots of these score gaps are included in the appendix, with partial plots provided in Figure~\ref{fig: partial eps results}). 
\ourAlgoName consistently outperforms the pure differentially private \ppeel for all values of $k$ and the approximate differentially private \cpeel for at least moderately large $k$. 
These results align with the findings of \citet{GillenwaterJMD22}, who compared \joint with \ppeel and \cpeel.

\begin{figure}[p]
    \centering
    \makebox[\textwidth]{
        \includegraphics[scale=0.32]{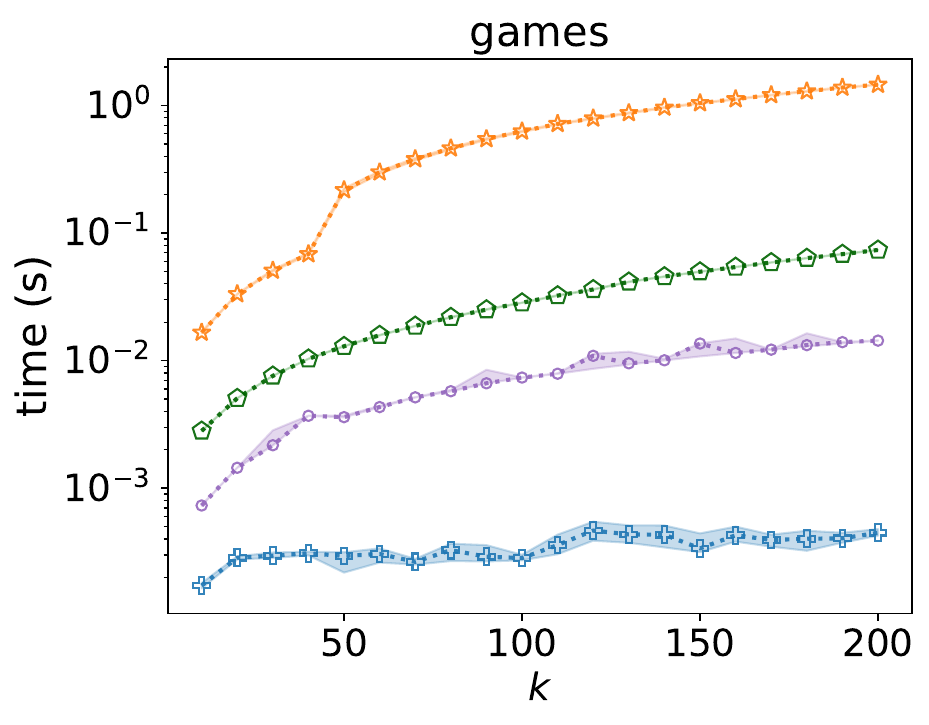}
        \hspace{-2mm}
        \includegraphics[scale=0.32]{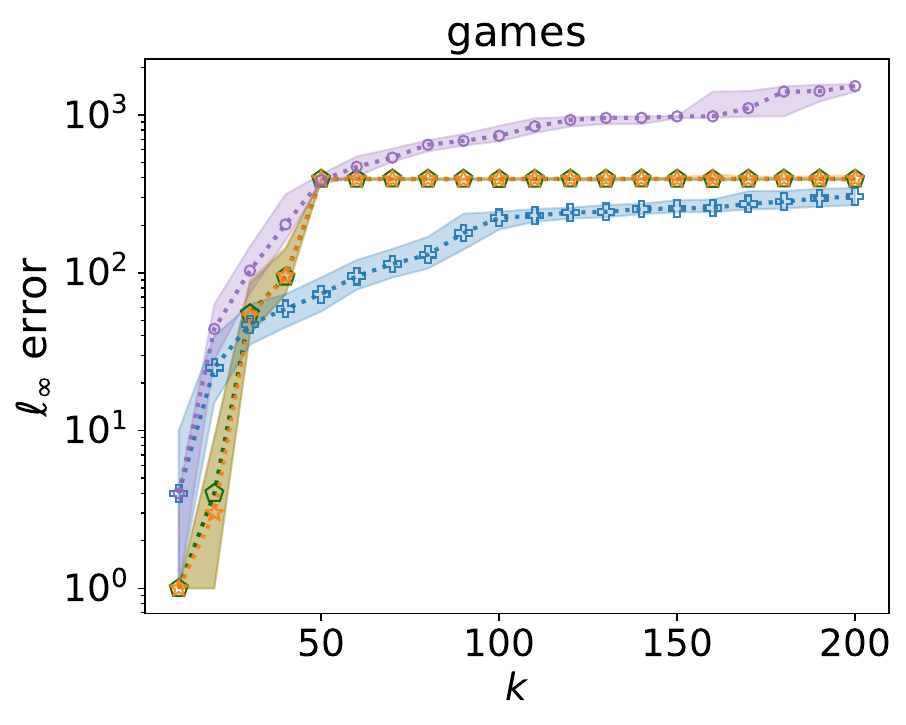}
        \hspace{-2mm}
        \includegraphics[scale=0.32]{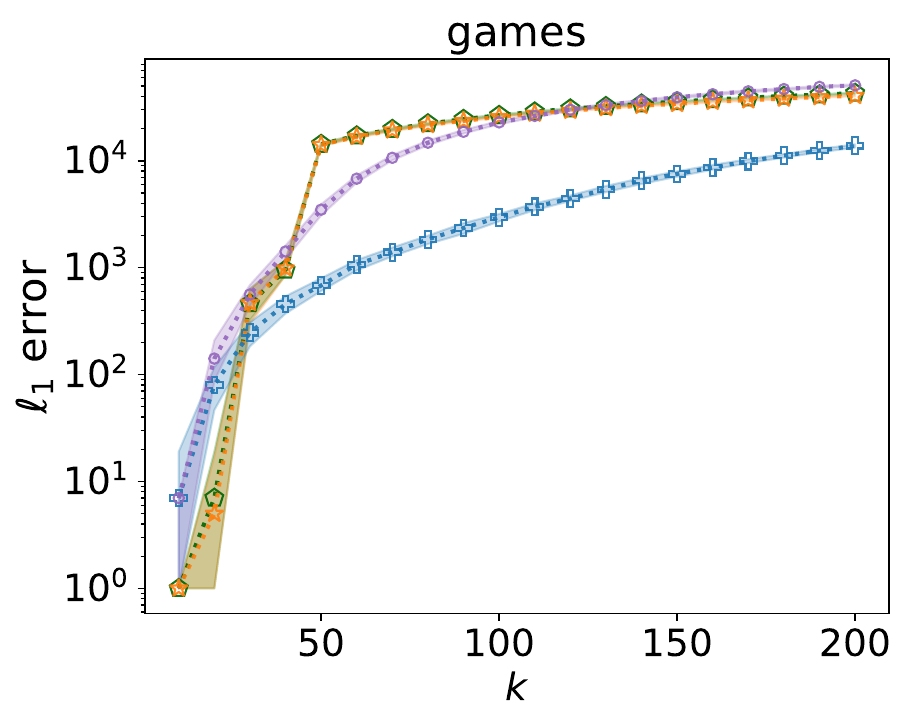} 
    } \\ \vspace{-10pt}
    \makebox[\textwidth]{
        \includegraphics[scale=0.32]{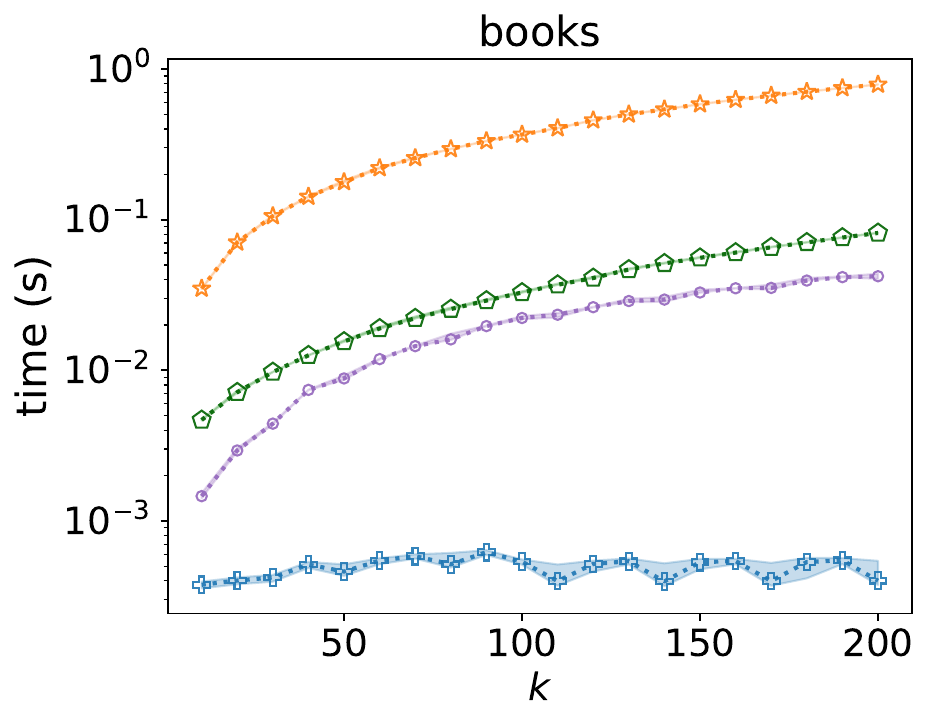}
        \hspace{-2mm}
        \includegraphics[scale=0.32]{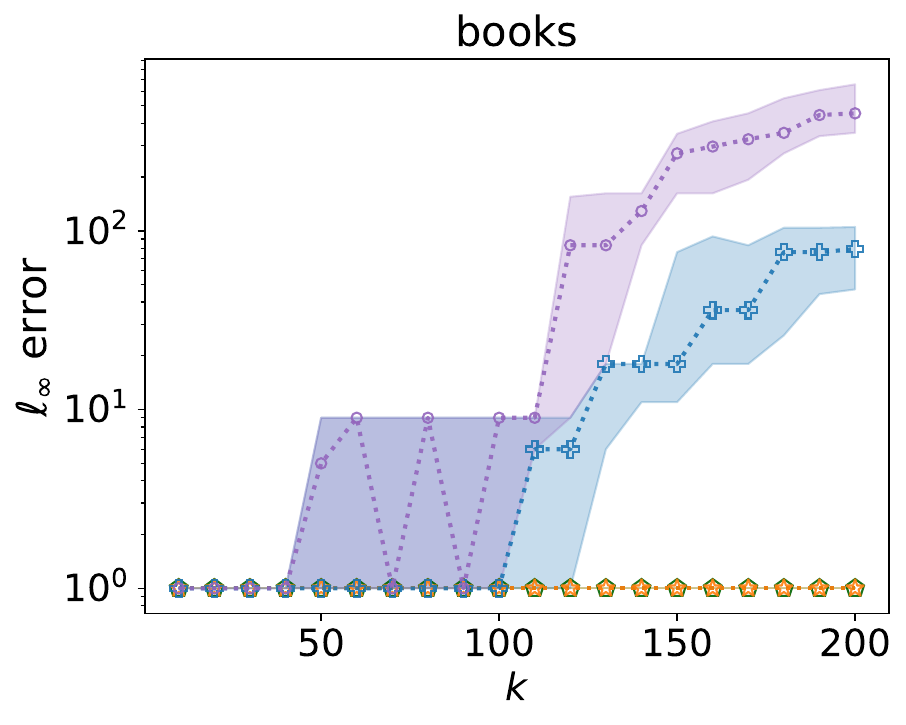}
        \hspace{-2mm}
        \includegraphics[scale=0.32]{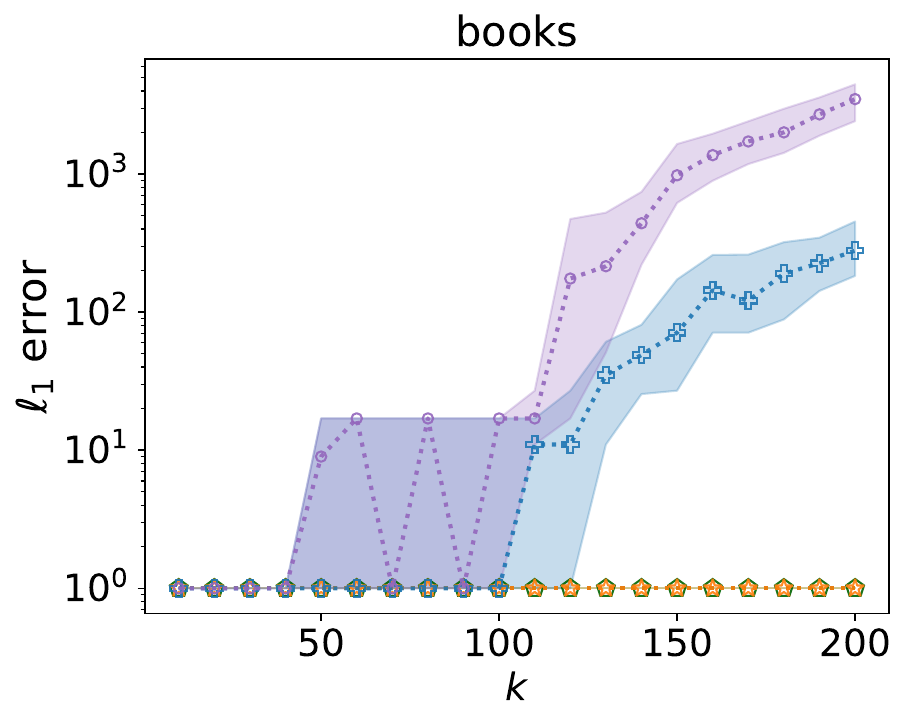} 
    }\\ \vspace{-10pt}
    \makebox[\textwidth]{
        \includegraphics[scale=0.32]{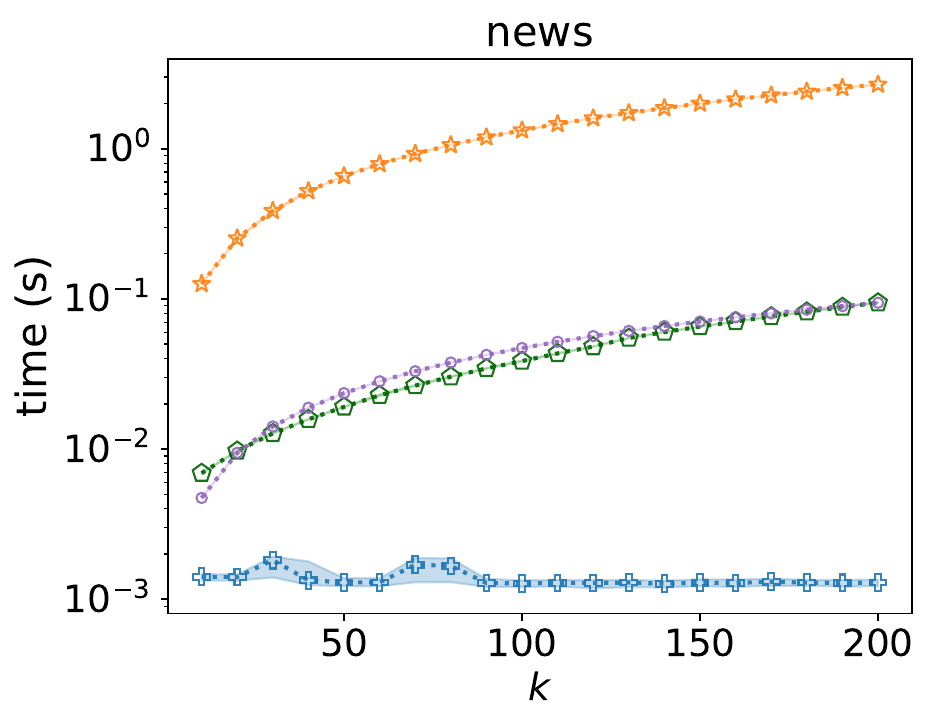}
        \hspace{-2mm}
        \includegraphics[scale=0.32]{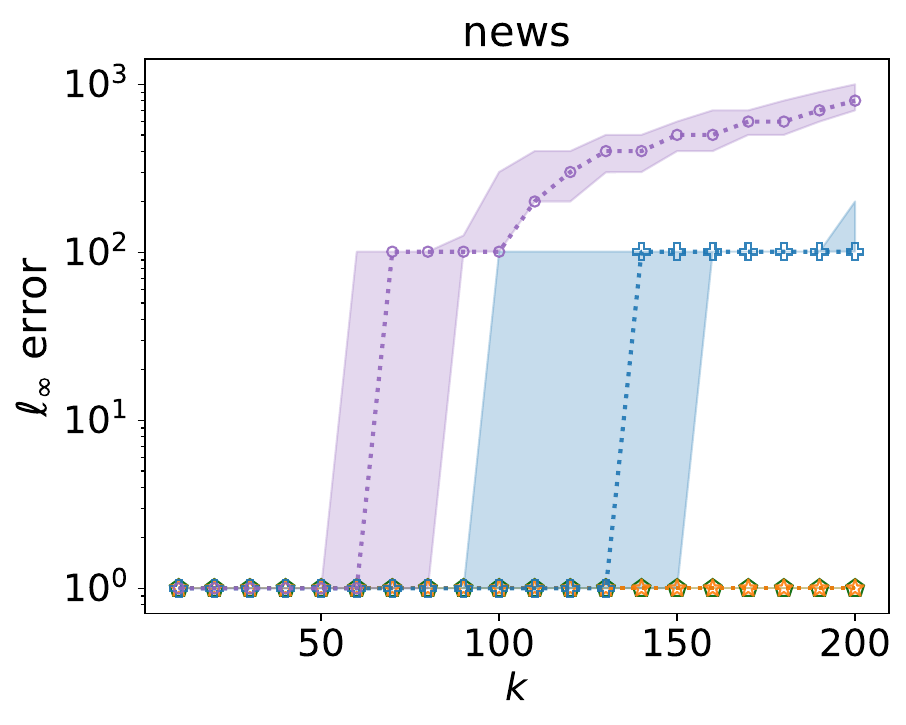}
        \hspace{-2mm}
        \includegraphics[scale=0.32]{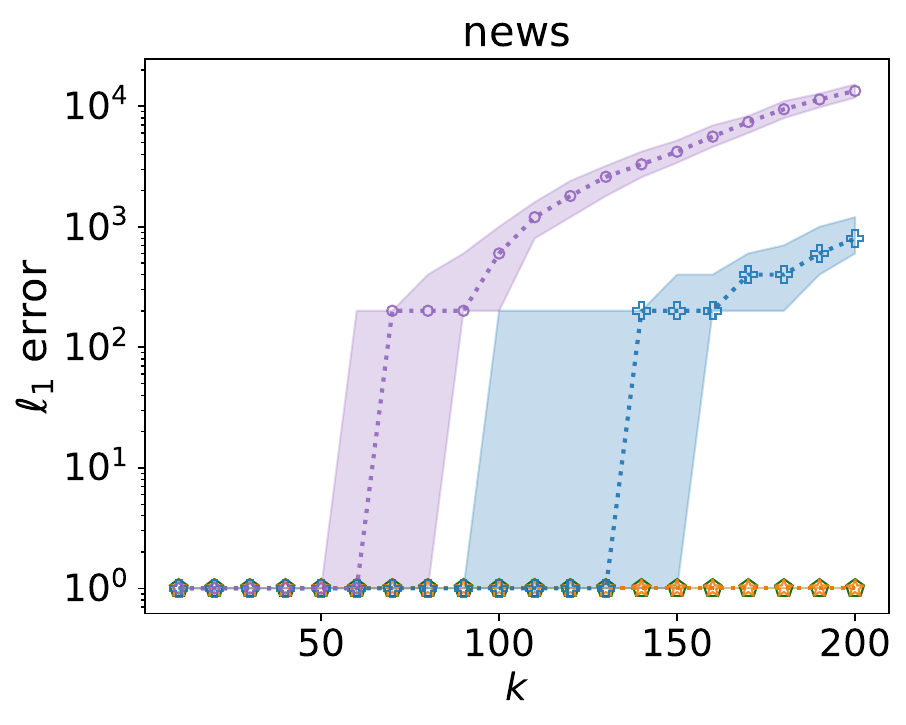} 
    } \\ \vspace{-10pt}
    \makebox[\textwidth]{
        \includegraphics[scale=0.32]{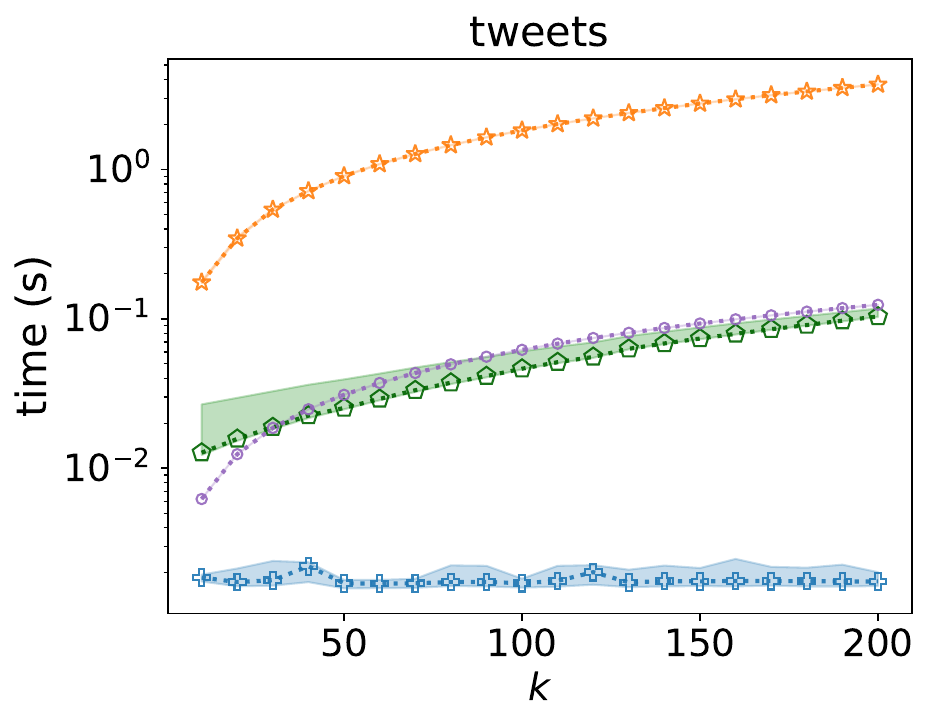}
        \hspace{-2mm}
        \includegraphics[scale=0.32]{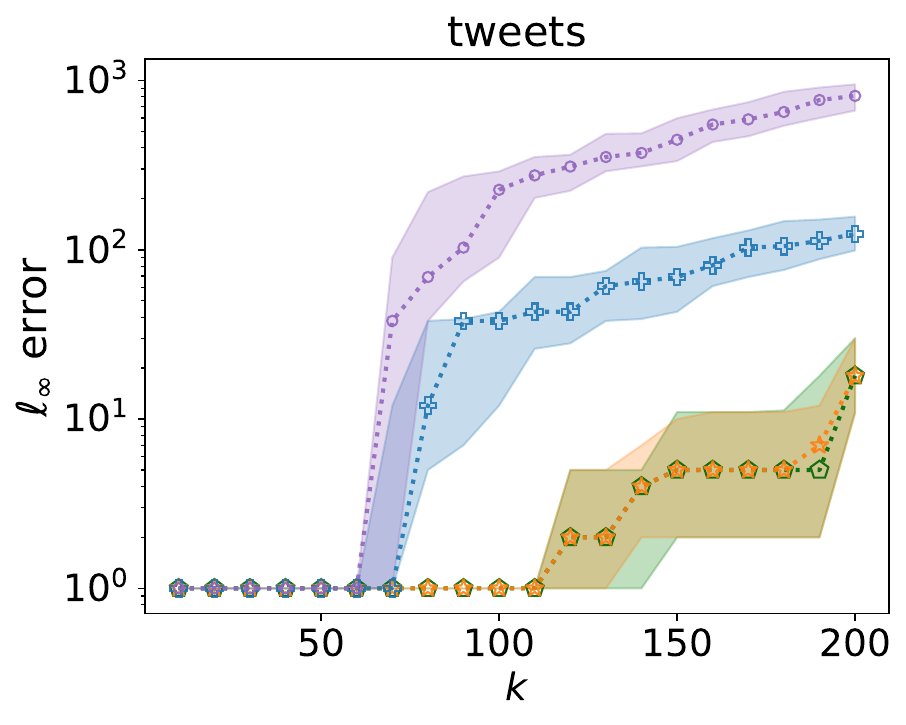}
        \hspace{-2mm}
        \includegraphics[scale=0.32]{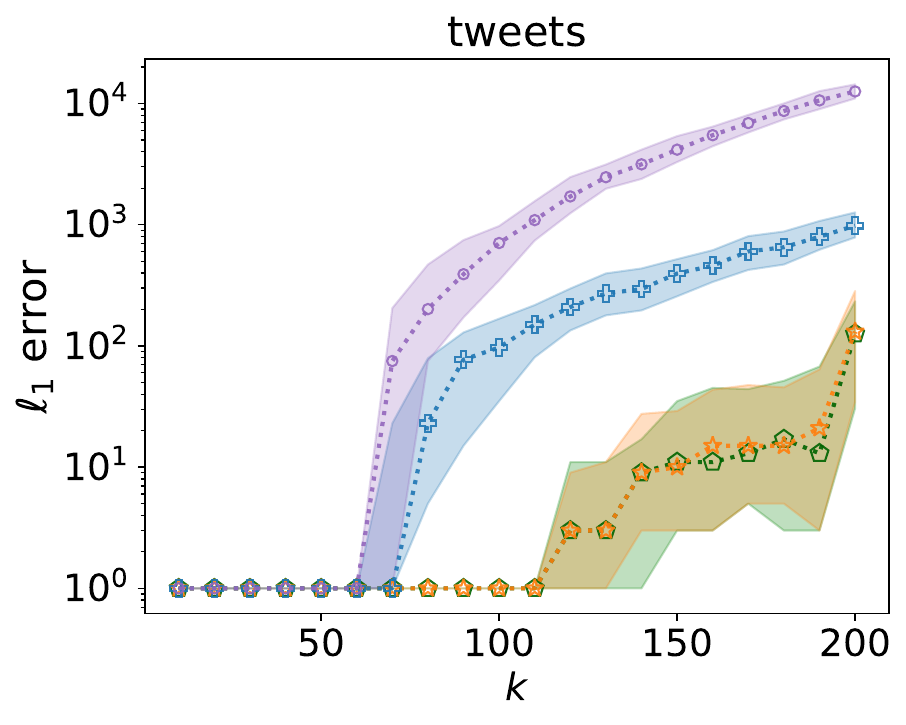} 
    } \\ \vspace{-10pt}
    \makebox[\textwidth]{
        \includegraphics[scale=0.32]{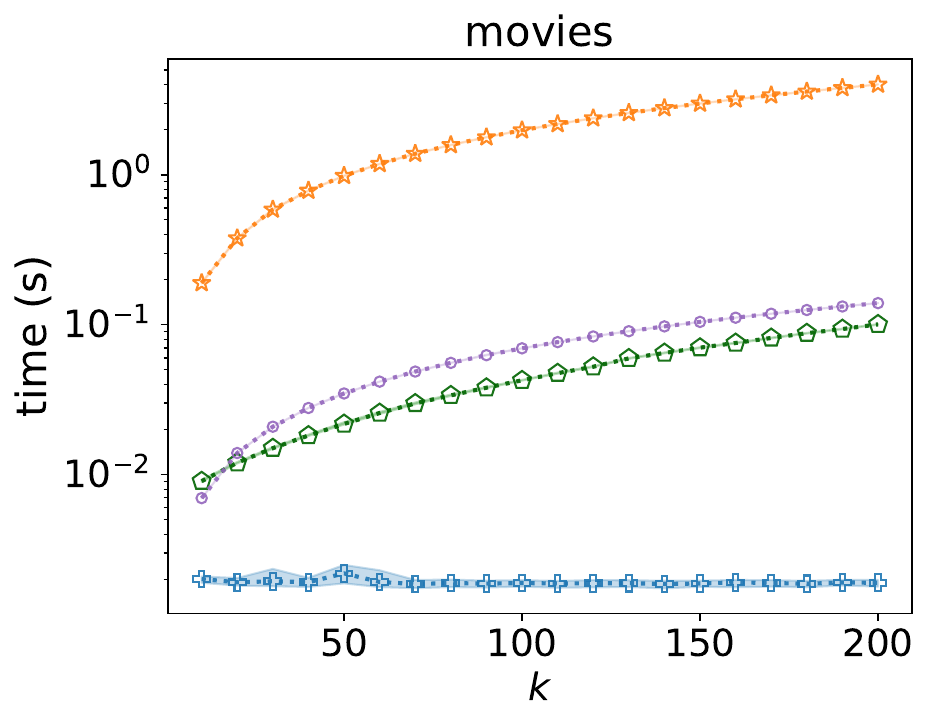}
        \hspace{-2mm}
        \includegraphics[scale=0.32]{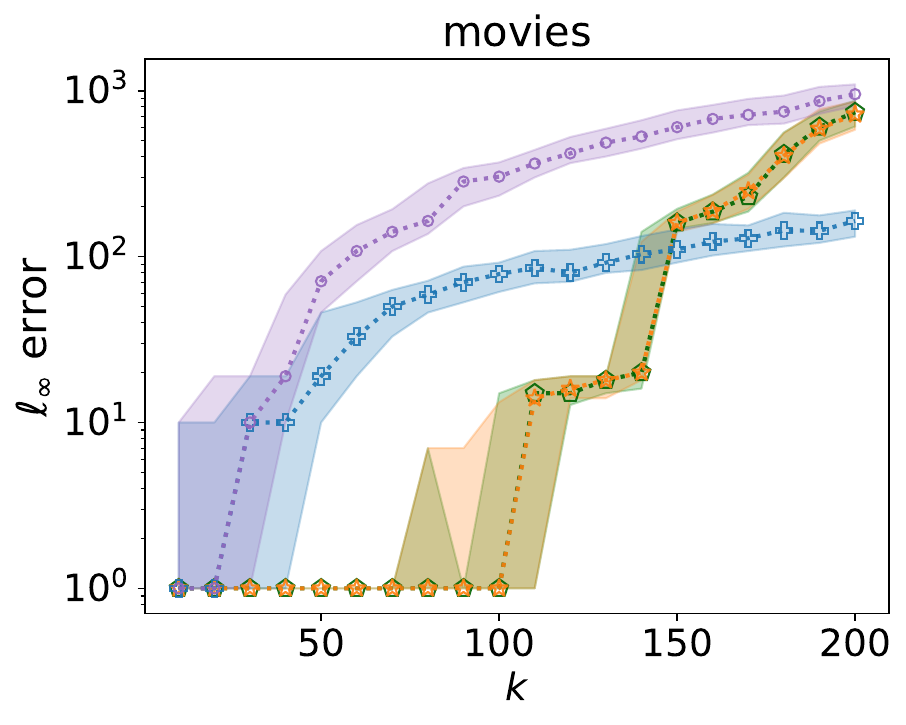}
        \hspace{-2mm}
        \includegraphics[scale=0.32]{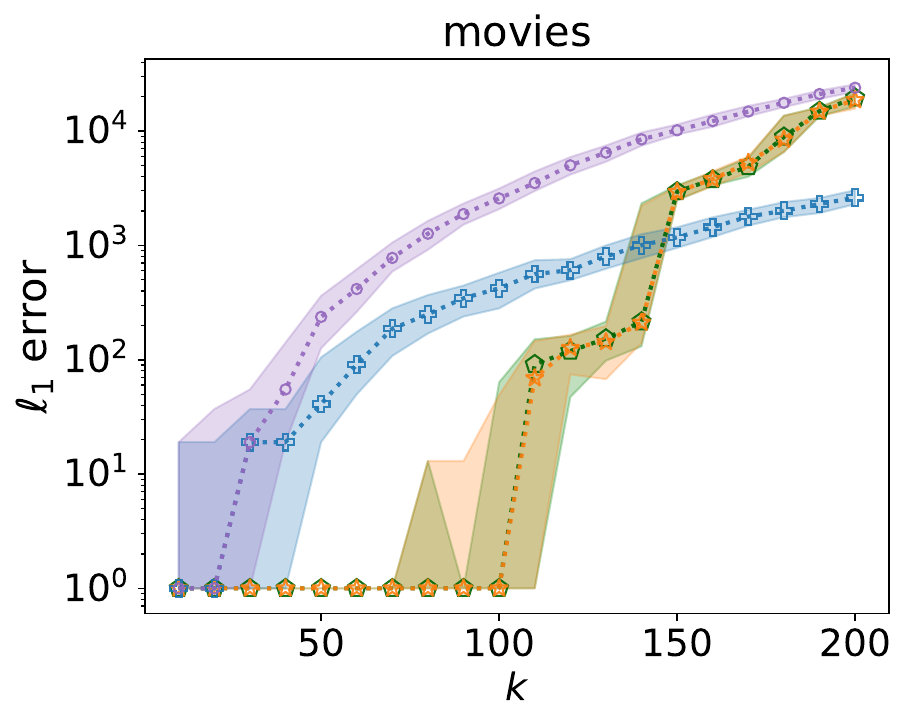} 
    } \\ \vspace{-10pt}
    \makebox[\textwidth]{
        \includegraphics[scale=0.32]{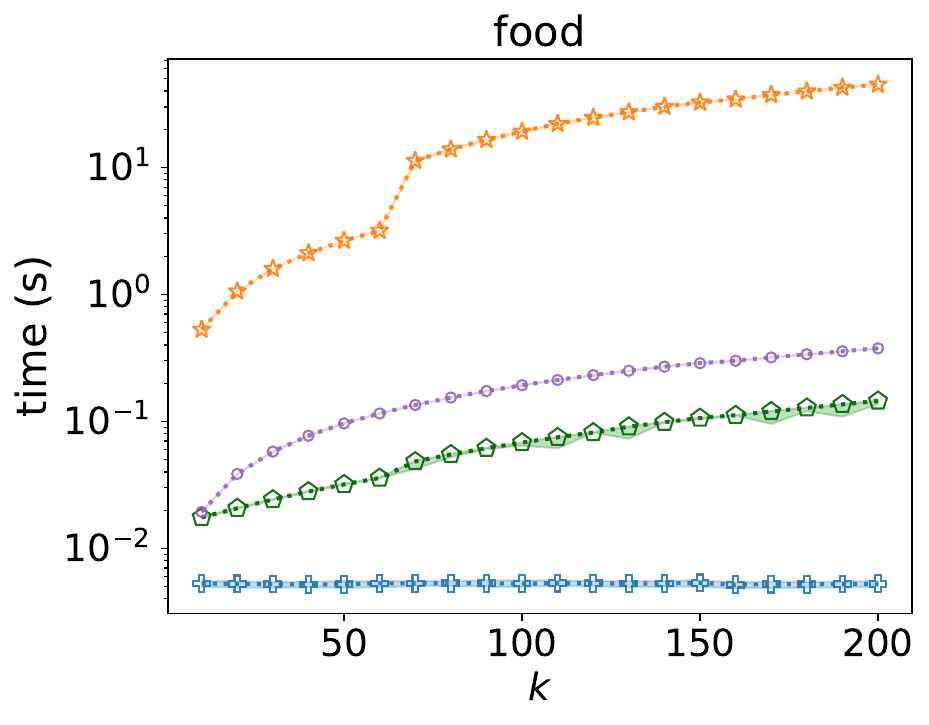}
        \hspace{-2mm}
        \includegraphics[scale=0.32]{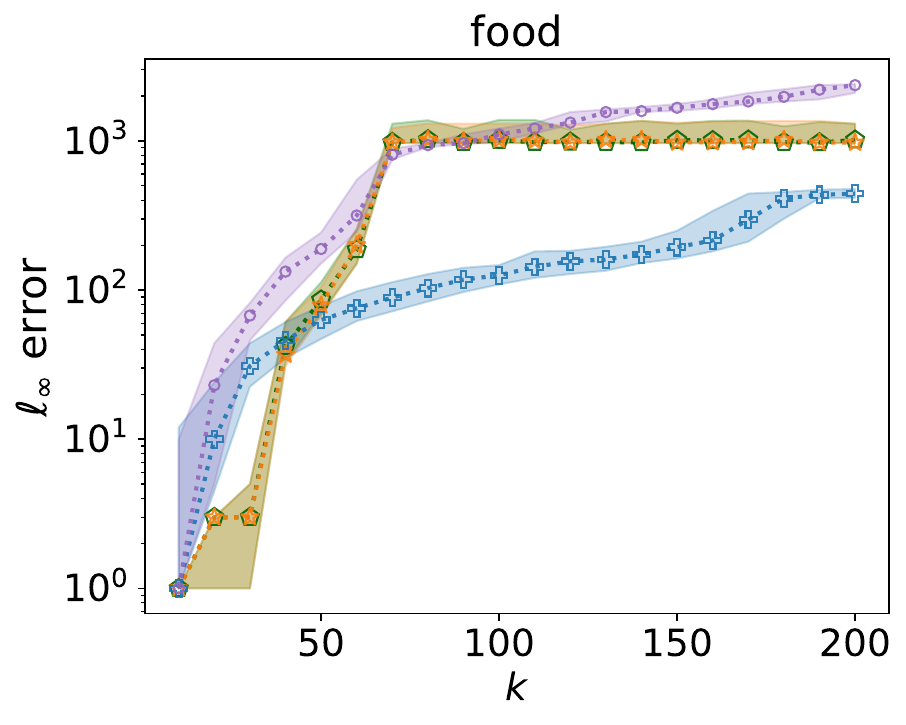}
        \hspace{-2mm}
        \includegraphics[scale=0.32]{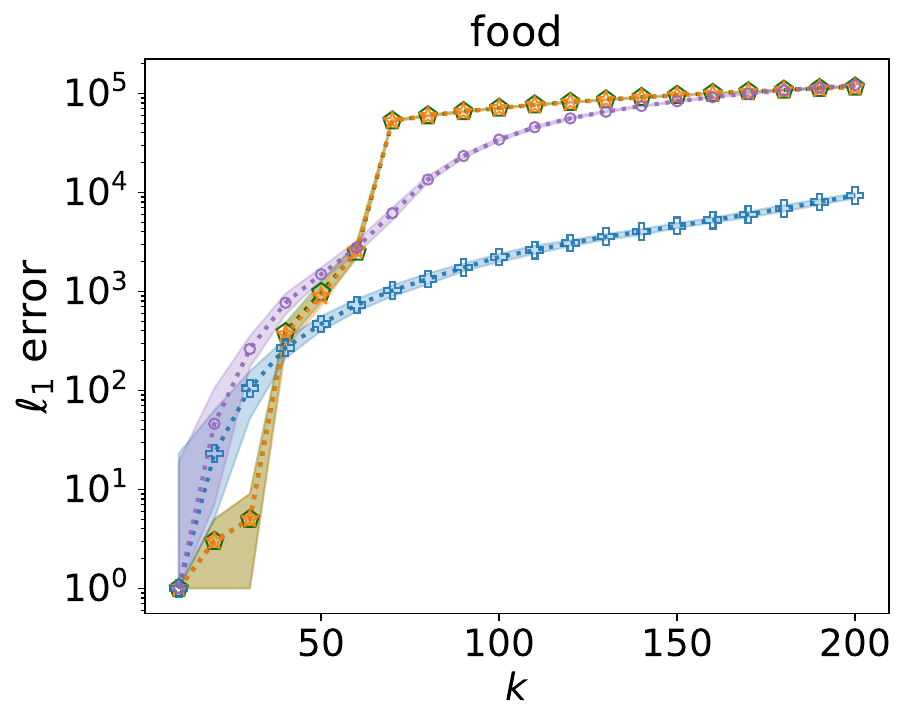} 
    } \\ \vspace{-4pt}
    \includegraphics[scale=0.5, trim={0 9mm 0 9mm}, clip]{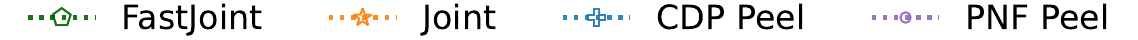}
    \vspace{-5pt}
    \caption{\centering 
        \textbf{Left}: Running time vs $k$. \,
        \textbf{Center}: $\ell_\infty$ error vs $k$. \,
        \textbf{Right}: $\ell_1$ error vs $k$. \hspace{3cm}
        The $\ell_1/\ell_\infty$ plots are padded by $1$ to avoid $\log 0$ on the $y$-axis. 
    }
    \label{fig: complete k results}
\end{figure}

\vspace{-2mm}
\paragraph{Varying $\eps$.}

Due to space constraints, Figure~\ref{fig: partial eps results} presents results for different values of \(\eps\) on two typical datasets: one where \ourAlgoName performs well and one where it does not. 
We replaced the \(\ell_1\) error plot (as it exhibits similar trends to the \(\ell_\infty\) plots) with a plot showing the gap between the top-300 scores\footnote{The \(k\) used for the varying \(\eps\) experiments is 100; here we plot the gap for the top-300 scores}.
The complete plots across all datasets are included in Appendix~\ref{sec: supplementary plots}.
The running time comparison resembles that of the varying-\(k\) plots in Figure~\ref{fig: complete k results}, with one notable difference: the running time of \ourAlgoName exhibits a clear decrease as \(\eps\) increases.
This observation aligns with our theoretical statement about the running time of \ourAlgoName, as detailed in Theorem~\ref{theorem: main result}. 

In terms of solution quality, \ourAlgoName and \joint perform particularly well on the \emph{news} dataset and are only slightly better than \ppeel and inferior to \cpeel on the \emph{games} dataset, where the gaps between the large scores in the former dataset are significantly larger than in the latter (note the values on the log-scale y-axis).
We provide an informal but informative explanation for this phenomenon: based on Lemma~\ref{lemma: utility guarantee of our algorithm}, \ourAlgoName is unlikely to sample sequences with loss greater than \(\tau\). 
Furthermore, when the distribution of top-\(k\) score gaps is highly skewed, there are very few sequences with errors between \((0, \tau]\), and the likelihood of sampling these sequences scales with $e^{-O(\eps)}$ as \(\eps\) varies. 
In contrast, the peeling-based mechanism needs to divide its privacy budget by \(k\) or \(\tilde{O}(\sqrt{k})\) for each round, causing the sampling probability of an error item to scale only with \(e^{- O(\eps / k) }\) or \(e^{-\tilde{O}(\eps / \sqrt{k})}\), which is higher than $e^{-O(\eps)}$.

\begin{figure}[!t]
    \centering
    \makebox[\textwidth]{
        \includegraphics[scale=0.32]{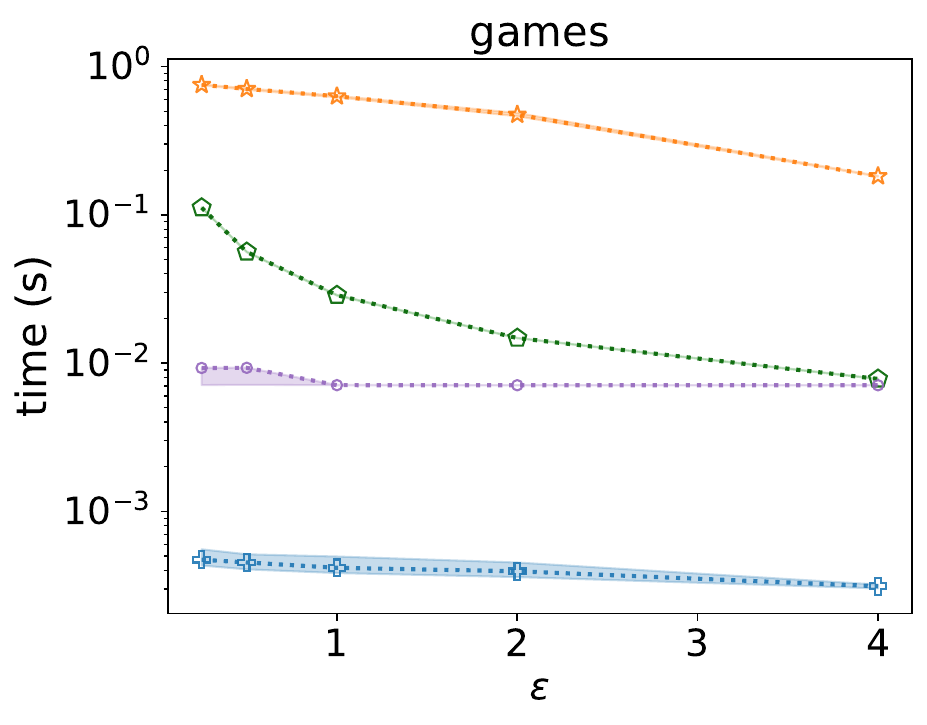}
        \hspace{-2mm}
        \includegraphics[scale=0.32]{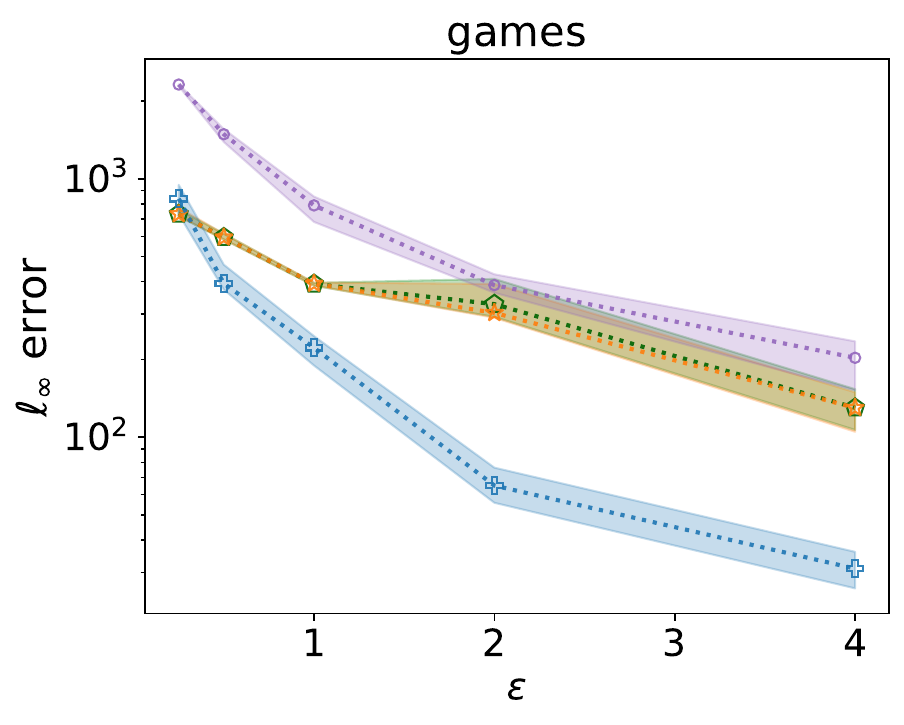}
        \hspace{-2mm}
        \includegraphics[scale=0.32]{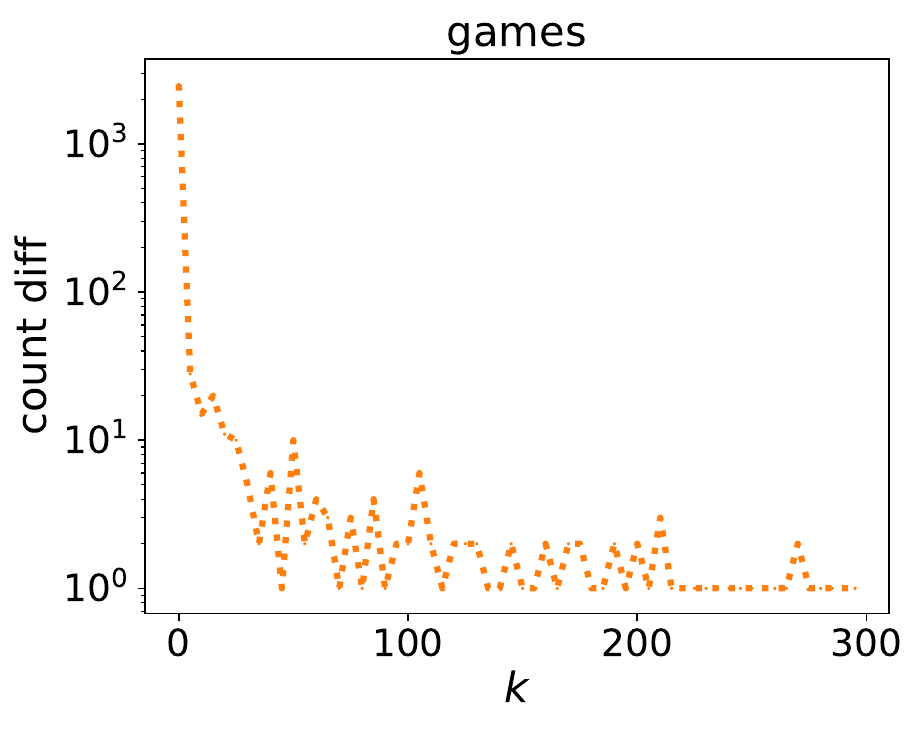} 
    } \\ \vspace{-10pt}
    \makebox[\textwidth]{
        \includegraphics[scale=0.32]{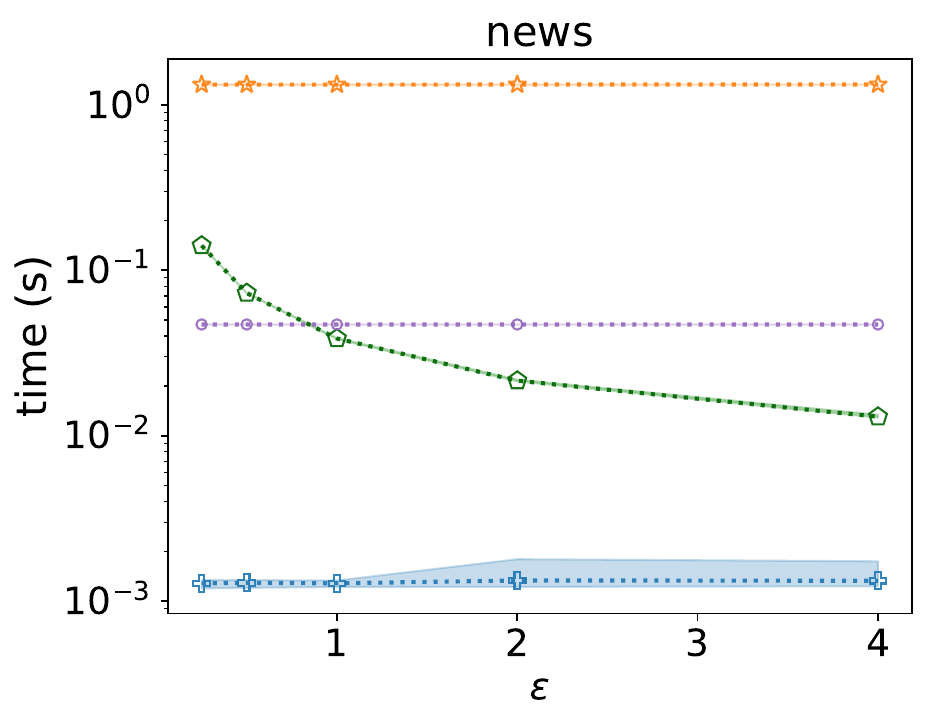}
        \hspace{-2mm}
        \includegraphics[scale=0.32]{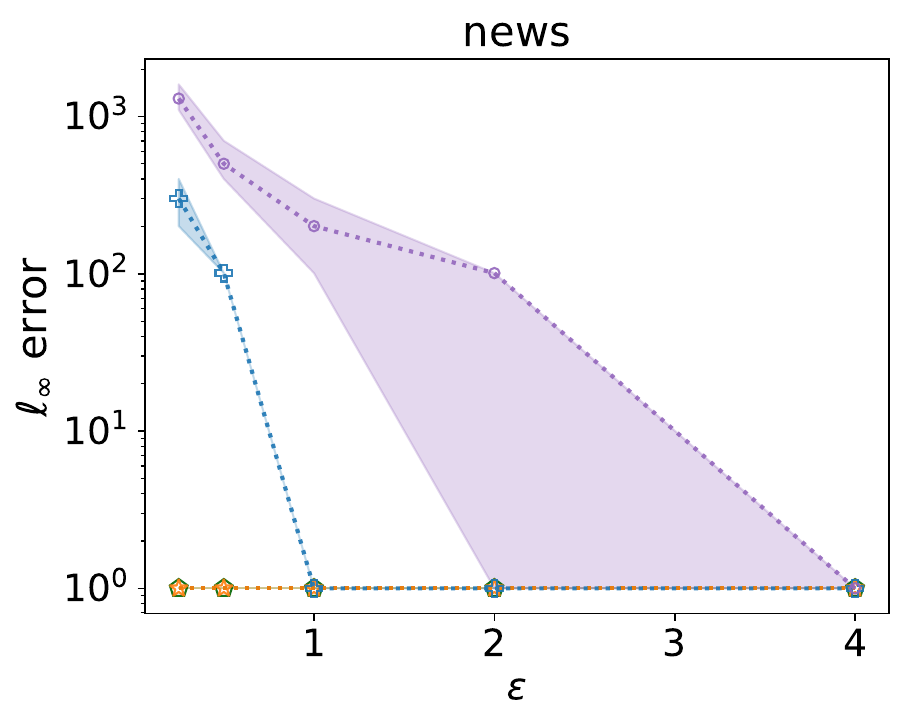}
        \hspace{-2mm}
        \includegraphics[scale=0.32]{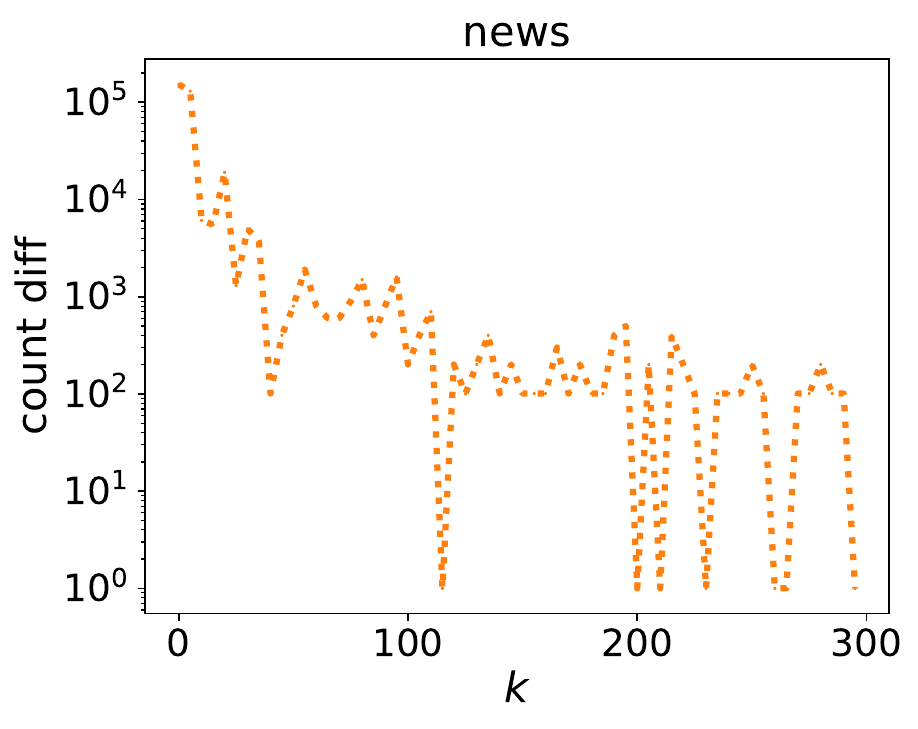} 
    } \\ \vspace{-3pt}
    \includegraphics[scale=0.5, trim={0 9mm 0 9mm}, clip]{figures/legend.pdf}
    \vspace{-5pt}
    \caption{
        \centering 
        \textbf{Left}: Running time Vs $\eps$. \,
        \textbf{Center}: $\ell_\infty$ error vs $\eps$. \,
        \textbf{Right}: Top-$300$ scores gaps. \hspace{3cm}
        The $\ell_1/\ell_\infty$ plots are padded by $1$ to avoid $\log 0$ on the $y$-axis. 
    }
    \label{fig: partial eps results}
    \vspace{-4mm}
\end{figure}

\vspace{-1mm}
\paragraph{Varying $\beta$.}
Due to space constraints, Figure~\ref{fig: partial delta results} presents results only for different values of $\beta$ on a medium-sized dataset. Similar plots for other datasets can be found in Appendix~\ref{sec: supplementary plots}.
It is anticipated that \joint, \ppeel and \cpeel do not exhibit significant performance variation concerning $\beta$.
However, it is somewhat surprising that \ourAlgoName does not neither.
This stability can be attributed to the threshold used for pruning, given by 
$
    \tau = 
    \lceil \frac{1}{\eps} \cdot \ln \paren{\binom{d}{k} \cdot k! / \beta} \rceil.
$
The numerator inside logarithm term, $\binom{d}{k} \cdot k!$, grows as $d^{\Theta(k)}$, 
significantly overshadowing $1 / \beta$. 
Consequently, $\tau$ changes only slightly as $\beta$ varies. 
This experiment demonstrates the robustness of \ourAlgoName's pruning strategy concerning the choice of $\beta$.

\begin{figure}[!ht]
    \vspace{-2mm}
    \centering
    \makebox[\textwidth]{
        \includegraphics[scale=0.32]{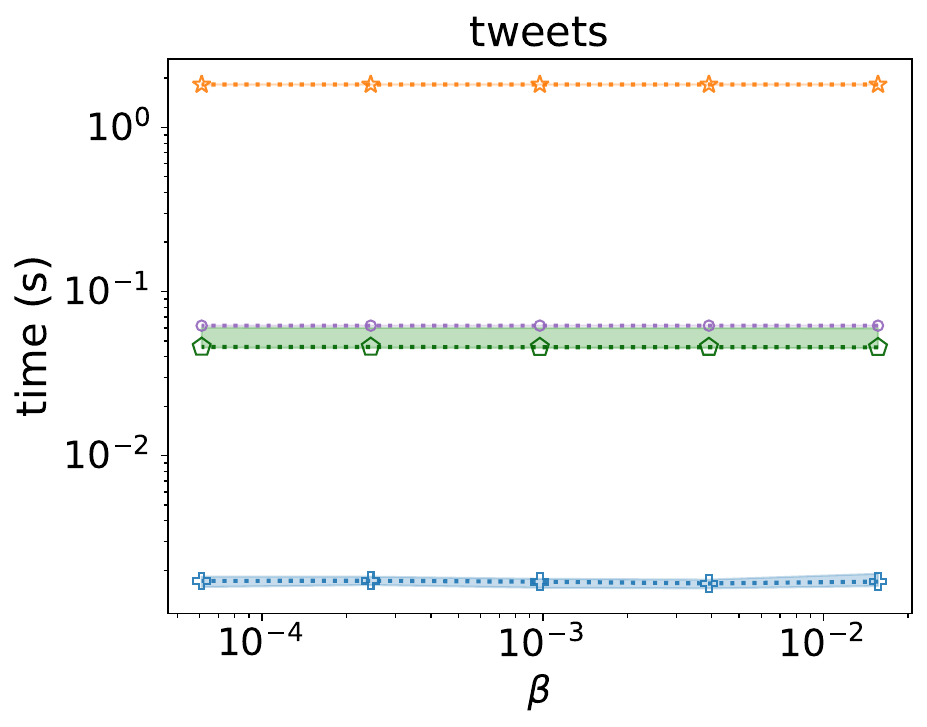}
        \hspace{-2mm}
        \includegraphics[scale=0.32]{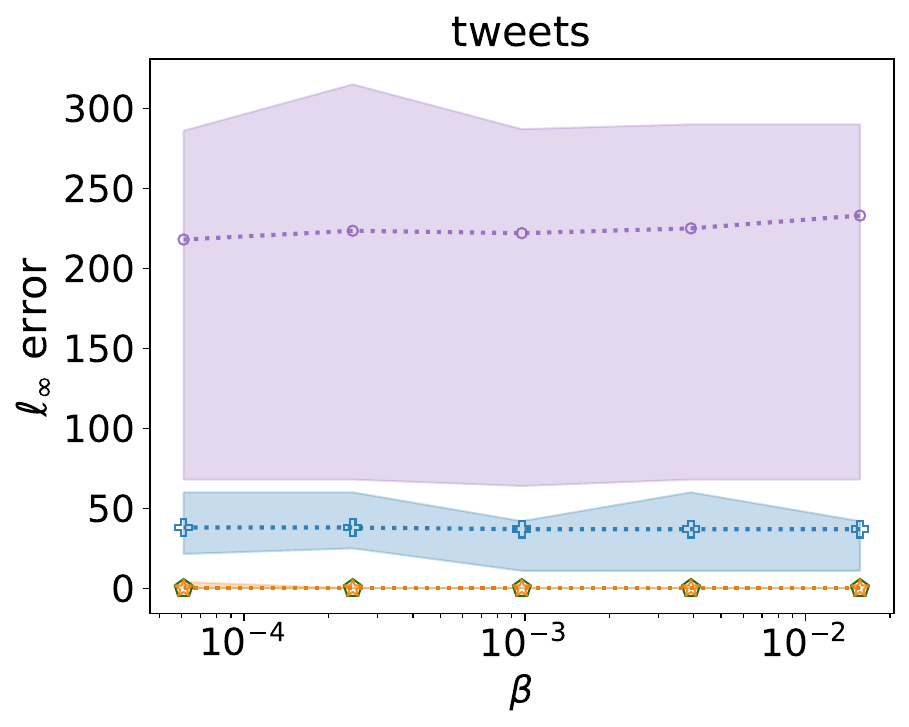}
        \hspace{-2mm}
        \includegraphics[scale=0.32]{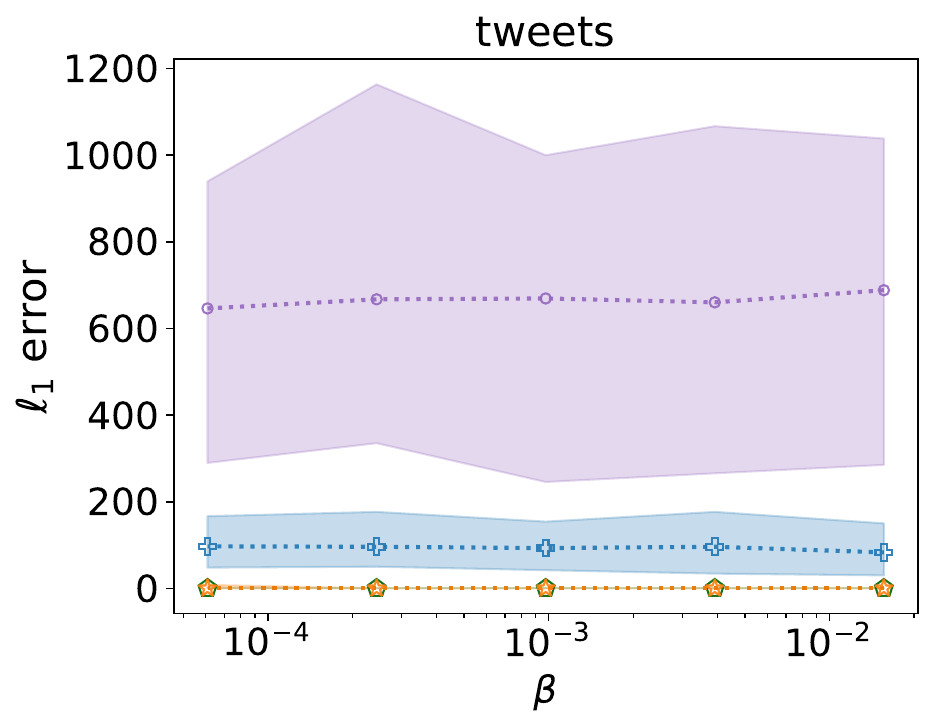} 
    } \\ \vspace{-10pt}
    \caption{
        \centering
        \textbf{Left}: Running time vs $\beta$. \,
        \textbf{Center}: $\ell_\infty$ error vs $\beta$. \,
        \textbf{Right}: $\ell_1$ error vs $\beta$. \hspace{3cm}
        The $\ell_1/\ell_\infty$ plots are padded by $1$ to avoid $\log 0$ on the $y$-axis. 
    }
    \label{fig: partial delta results}
    \vspace{-2mm}
\end{figure}

\section{Related Work}
\label{sec: related-work}

\paragraph{Comparison with \joint.}
We can also explain \joint~\citep{GillenwaterJMD22} within the novel framework proposed in Section~\ref{subsec: algorithm overview}, which consists of the \emph{\subsetsampling} and \emph{\subsetSequencesampling} steps. 
Their approach assumes that $\hist[1] \ge \dots \ge \hist[d]$, which can be achieved by sorting $\hist$. 
For each $i \in [k]$ and $j \in [d]$, define $\loss_{i, j} \doteq \hist[i] - \hist[j]$.
The partition they consider is equivalent to the one defined as follows:
\begin{equation*}
    \begin{array}{c}
        U_{i, j} \doteq \set{ 
            \topkSoln = \paren{\topkSoln[1], \ldots, \topkSoln[k]} \in \partialPermutation{\DataDomain}{k} : 
                \arraycolsep=0.8pt
                \begin{array}{rll}
                    \hist\BigBracket{\topkSoln[\ell]} &> \hist[\ell] - \loss_{i, j}, & \quad \forall \ell < i \\
                    \topkSoln[\ell] &= j,  & \quad \, \ell = i \\
                    \hist\BigBracket{\topkSoln[\ell]} &\ge \hist[\ell] - \loss_{i, j}, & \quad \forall \ell > i
                \end{array}
        }
    \end{array}, 
    \forall i \in [k], j \in [d].
\end{equation*}
Intuitively, $U_{i, j}$ consists of the length-$k$ sequences, $\topkSoln$, that satisfy: 1) the $i^{\text{th}}$ element in $\topkSoln$ is exactly element $j$; 2) the first $i - 1$ elements in $\topkSoln$ have losses less than $\loss_{i, j}$; and 3) the last $k - i$ elements in $\topkSoln
$ have losses at most $\loss_{i, j}$.

Therefore, the sequences $\topkSoln$ in $U_{i, j}$ share the same loss, $\loss{\joint}{\hist}{\topkSoln} = \loss_{i, j}$ (as defined in Equation~\eqref{eq: definition of error for joint}), with the first position reaching this loss being the $i^{\text{th}}$ position, where element $j$ appears. 
Furthermore, sequences in different subsets, $U_{i, j}$ and $U_{i', j'}$, can share the same loss, as it is possible that $\loss_{i, j} = \loss_{i', j'}$. 
There are $dk$ subsets in this partition, resulting in a running time of $\tilde{O}(dk)$ for their implementation~\citep{GillenwaterJMD22}.

Applying the pruning technique to this partition can not reduce the number of subsets to $o(dk)$. Let $\tau$ be as defined in Theorem~\ref{theorem: main result}. 
For each $i \in [k]$, we aim to find the first $j$ such that $\hist[j] \le \hist_{(i)} - \tau$. Denote this value by $\sigma(i) \doteq \min \{j \in [d] : \hist[j] \le \hist_{(i)} - \tau \}$. 
Following the spirit of our pruning technique, we would like to merge the trailing subsets $U_{i, \sigma(i)}, \dots, U_{i, d}$ into a single subset $U_{i, \sigma(i)} \doteq \bigcup_{j \ge \sigma(i)} U_{i, j}$ to reduce the number of subsets. 
However, it is easy to find counterexamples where $\Omega(d)$ items are equal to $\hist_{(i)}$ for all $i \in [k]$. 
For example, consider the case where $\hist[1] = \hist[2] = \cdots = \hist[d] = c$, for some constant $c$. 
In such scenarios, we still have $\sigma(i) \in \Omega(d)$, and therefore, in the worst case, the number of subsets remains $\sum_{i \in [k]} \sigma(i) \in \Omega(dk)$.

\paragraph{Truncated Loss.} 
The technique of applying the exponential mechanism with truncated scores was considered by \citet{BLST10}. 
Their top-$k$ selection algorithm employs a peeling-based approach: it samples $k$ items without replacement, selecting one item at a time using the exponential mechanism with truncated scores.
This method is employed because, in their setting, obtaining the scores of the input histogram $\hist$ is expensive, leading them to treat lower-scoring items uniformly by assigning them a small, identical score.
In contrast, when all scores of \(\hist\) are known, iteratively applying the exponential mechanism to select the top-$k$ items has an equivalent linear time implementation~\citep{DurfeeR19} (the \cpeel algorithm in Section~\ref{sec: experiments}). 
Therefore, truncating the scores of $\hist$ is unnecessary in this case.

\paragraph{Adaptive Private $k$ Selection.}
As the experiments show, the performance of \joint and \ourAlgoName depends on gap size—they perform well when there are large gaps between the top-$k$ items. 
An orthogonal line of research~\citep{WZ22} leverages large gaps to privately identify the index $i$ that approximately maximizes the gap between the $i^{\text{th}}$ and $(i + 1)^{\text{th}}$ largest elements. This is followed by testing whether the gap (using techniques like propose-test-release) between the $i^{\text{th}}$ and $(i + 1)^{\text{th}}$ largest elements is sufficiently large, allowing the top-$i$ items to be returned without additional noise. 
This approach benefits by adding no noise in the final step. 
However, there are two key differences: 1) it does not guarantee returning at least $k$ items, as $i$ can be less than $k$; 2) more crucially, the top-$i$ items must be returned as an \emph{unordered set}. The first issue can be addressed by iteratively applying the above mechanism. 
For the second issue, algorithms introduced in this paper, such as \cpeel and \ourAlgoName, can serve as subroutines. 
Notably, \ourAlgoName may provide better empirical performance when large gaps exist among the top $i$ items.

\paragraph{Utility Lower Bound.}
\citet*{BafnaU17} and~\citet*{SteinkeU17} demonstrate that, for approximate private algorithms, existing methods~\citep{DurfeeR19, QiaoSZ21}, including \cpeel~\citep{DurfeeR19} as compared in Section~\ref{sec: experiments}, achieve theoretically asymptotically optimal privacy-utility trade-offs.

\begin{ack}
    We thank the anonymous reviewers for their feedback which helped improve the paper. 
    Hao WU was a Postdoctoral Fellow at the University of Copenhagen, supported by Providentia, a Data Science Distinguished Investigator grant from Novo Nordisk Fonden, and affiliated with Basic Algorithms Research Copenhagen (BARC), supported by the VILLUM Foundation grant 16582.
    Hanwen Zhang is affiliated with Basic Algorithms Research Copenhagen (BARC), supported by the VILLUM Foundation grant 16582. 
    Hanwen Zhang is partially supported by Starting Grant 1054-00032B from the Independent Research Fund Denmark under the Sapere Aude research career programme. 
\end{ack}

\bibliographystyle{IEEEtranSN}
\bibliography{reference}

\newpage
\appendix
\section{Missing proofs}
\label{appendix: missing proofs}

\begin{proof}[Proof for Fact~\ref{fact: joint mechanism}]
    The proof has been implicitly suggested in our comparison with \joint in Section~\ref{sec: related-work}. 
    We reiterate it here.

    We explain \joint~\citep{GillenwaterJMD22} within the novel framework proposed in Section~\ref{subsec: algorithm overview}, which consists of the \emph{\subsetsampling} and \emph{\subsetSequencesampling} steps. 
    Their approach assumes that $\hist[1] \ge \dots \ge \hist[d]$, which can be achieved by sorting $\hist$. 
    For each $i \in [k]$ and $j \in [d]$, define $\loss_{i, j} \doteq \hist[i] - \hist[j]$.
    The partition they consider is equivalent to the one defined as follows:
    \begin{equation*}
        \begin{array}{c}
            U_{i, j} \doteq \set{ 
                \topkSoln = \paren{\topkSoln[1], \ldots, \topkSoln[k]} \in \partialPermutation{\DataDomain}{k} : 
                    \arraycolsep=0.8pt
                    \begin{array}{rll}
                        \hist\BigBracket{\topkSoln[\ell]} &> \hist[\ell] - \loss_{i, j}, & \quad \forall \ell < i \\
                        \topkSoln[\ell] &= j,  & \quad \, \ell = i \\
                        \hist\BigBracket{\topkSoln[\ell]} &\ge \hist[\ell] - \loss_{i, j}, & \quad \forall \ell > i
                    \end{array}
            }
        \end{array}, 
        \forall i \in [k], j \in [d].
    \end{equation*}
    Intuitively, $U_{i, j}$ consists of the length-$k$ sequences, $\topkSoln$, that satisfy: 1) the $i^{\text{th}}$ element in $\topkSoln$ is exactly element $j$; 2) the first $i - 1$ elements in $\topkSoln$ have losses less than $\loss_{i, j}$; and 3) the last $k - i$ elements in $\topkSoln
    $ have losses at most $\loss_{i, j}$.
    
    Therefore, the sequences $\topkSoln$ in $U_{i, j}$ share the same loss, $\loss{\joint}{\hist}{\topkSoln} = \loss_{i, j}$ (as defined in Equation~\eqref{eq: definition of error for joint}), with the first position reaching this loss being the $i^{\text{th}}$ position, where element $j$ appears. 
    Furthermore, sequences in different subsets, $U_{i, j}$ and $U_{i', j'}$, can share the same loss, as it is possible that $\loss_{i, j} = \loss_{i', j'}$. 

    Next, we briefly explain how to fulfill the \emph{\subsetsampling} and \emph{\subsetSequencesampling} steps for the chosen partition.

    The \emph{\subsetSequencesampling} step is straightforward: sampling a sequence uniformly at random from $U_{i, j}$ can be achieved similarly (though not identically) to the approach used in Algorithm~\ref{algo: sampling step II} for sampling a sequence from $\groupSeq{r}{i}$. This process can be completed in $O(d)$ time.

    For the \emph{\subsetsampling} step, similar to the algorithm in Section~\ref{subsec: algorithm overview} and using Fact~\ref{fact: sampling based on grumbel noisy max}, sampling a subset $U_{i, j}$ with probability proportional to 
    $
        |U_{i, j}| \cdot \exp \PAREN{ - \eps \cdot \loss_{i,j} / 2},  
    $
    can be achieved by computing the maximum of
    $$
        \set{X_{i, j} + \ln \bigparen{|U_{i,j}| \cdot \exp \PAREN{ - \eps \cdot \loss_{i,j} / 2}} : (i, j) \in [k] \times [d] },
    $$ 
    where \( X_{i, j} \sim \GumbelNoise{1} \). 
    The main task is to efficiently compute \( \ln |U_{i,j}| \). 
    For each \( i \in [k], j \in [d], \ell \in [k] \), we define
    $$
        t_{i, j, >}[\ell] \doteq \card{ \set{ j' \in [d] : \hist[j'] > \hist[\ell] - \loss_{i, j} }},
        \text{ and }
        t_{i, j, \ge}[\ell] \doteq \card{\set{ j' \in [d] : \hist[j'] \ge \hist[\ell] - \loss_{i, j} }}.
    $$

    It can be shown that
    $$
        \ln |U_{i, j}| 
            = \sum_{\ell \in [i - 1]} \ln \bigparen{ t_{i, j, >}[\ell] - (\ell - 1) } + \sum_{\ell \in \IntSet{i + 1}{k}} \ln \bigparen{ t_{i, j, \ge}[\ell] - (\ell - 1) }.
    $$
    Furthermore, if we define
    $$
        \tilde{t}_{i, j}[\ell] \doteq \card{
            \set{ j' \in [d] : \hist[j'] \ge \hist[\ell] - \loss_{i, j} - \PAREN{
                    \frac{i - \ell}{2k} - \frac{j - j'}{2 d k}
                }
            }
        },
    $$
    we observe that $\card{\frac{i - \ell}{2k} - \frac{j - j'}{2 d k}} < 1.$ 
    When $\ell < i$, $\frac{i - \ell}{2k} - \frac{j - j'}{2 d k} \ge \frac{1}{2k} - \frac{d - 1}{2 d k} > 0$, so $\tilde{t}_{i, j}[\ell] = t_{i, j, >}[\ell]$.
    On the other hand, when $\ell > i$, $\frac{i - \ell}{2k} - \frac{j - j'}{2 d k} \le - \frac{1}{2k} - \frac{d - 1}{2 d k} < 0$, so $\tilde{t}_{i, j}[\ell] = t_{i, j, \ge}[\ell]$.
    Therefore, it follows that
    $$
        \ln |U_{i, j}| 
            = \sum_{\ell \neq i} \ln \bigparen{ \tilde{t}_{i, j}[\ell] - (\ell - 1) } .
    $$

    It remains to demonstrate an efficient method for computing $\tilde{t}_{i, j}$. 
    First, we sort the $(i, j) \in [k] \times [d]$ pairs in increasing order based on $\loss_{i, j} + \frac{i}{2k} - \frac{j}{2dk}$. 
    For a fixed $i$, the values $\loss_{i, j}$ for $j \in [d]$ are already in increasing order, so a sorted sequence of $\loss_{i, j} + \frac{i}{2k} - \frac{j}{2dk}$ for $j \in [d]$ can be obtained in $O(d)$ time. 
    We then merge $k$ sorted sequences in $O(dk \log k)$ time using $k$-way merging.
    
    If $(\hat{i}, \hat{j})$ appears immediately after $(i, j)$ in the sorted order, then $\ln |U_{\hat{i}, \hat{j}}|$ can be derived from $\ln |U_{i, j}|$ in $O(1)$ time. 
    Specifically, we claim that
    $$
        \tilde{t}_{\hat{i}, \hat{j}}[\ell] = \begin{cases}
            \tilde{t}_{i, j}[\ell],    &\text{ if } \ell \neq \hat{i},  \\
            \tilde{t}_{i, j}[\ell] + 1,    &\text{ if } \ell = \hat{i}.
        \end{cases}
    $$

    For the first case, assume for contradiction that there exists some $\ell \neq \hat{i}$ such that $\tilde{t}_{\hat{i}, \hat{j}}[\ell] > \tilde{t}_{i, j}[\ell]$. 
    Consequently, there exists $j'$ such that
    \begin{equation}
        \label{ineq:joint1}
        \hist[\ell] - \loss_{i, j} - \PAREN{
                    \frac{i - \ell}{2k} - \frac{j - j'}{2 d k}
        }
        > \hist[j'] \ge 
        \hist[\ell] - \loss_{\hat{i}, \hat{j}} - \PAREN{
                    \frac{\hat{i} - \ell}{2k} - \frac{\hat{j} - j'}{2 d k}
        }
    \end{equation}
    which implies that
    \begin{equation}
        \label{ineq:joint2}
        \loss_{\hat{i}, \hat{j}} + \frac{\hat{i}}{2k} - \frac{\hat{j}}{2dk}
        > 
        \hist[\ell] - \hist[j'] + \PAREN{
            \frac{\ell}{2k} - \frac{j'}{2 d k}
        }
        = \loss_{\ell, j'} + \frac{\ell}{2k} - \frac{j'}{2dk}
        \ge 
        \loss_{i, j} + \frac{i}{2k} - \frac{j}{2dk}.
    \end{equation}
    Since the values of $\loss_{i, j} + \frac{i}{2k} - \frac{j}{2dk}$ are distinct, it follows that 
    \begin{equation}
        \label{ineq:joint3}
        \loss_{\hat{i}, \hat{j}} + \frac{\hat{i}}{2k} - \frac{\hat{j}}{2dk}
        > 
        \loss_{\ell, j'} + \frac{\ell}{2k} - \frac{j'}{2dk} 
        >
        \loss_{i, j} + \frac{i}{2k} - \frac{j}{2dk},
    \end{equation}
    which contradicts the assumption that $(\hat{i}, \hat{j})$ appears immediately after $(i, j)$ in the sorted order.
    
    For the second case ($\ell = \hat{i}$), it is clear that $\tilde{t}_{\hat{i}, \hat{j}}[\ell] \ge \tilde{t}_{i, j}[\ell]$. 
    Additionally, observe that
    $$
        \hist[\hat{i}] - \loss_{i, j} - \PAREN{
                    \frac{i - \hat{i}}{2k} - \frac{j - \hat{j}}{2 d k}
        }
        > \hist[\hat{j}] \ge 
        \hist[\hat{i}] - \loss_{\hat{i}, \hat{j}} - \PAREN{
                    \frac{\hat{i} - \hat{i}}{2k} - \frac{\hat{j} - \hat{j}}{2 d k}
        }
    $$
    implying that $\tilde{t}_{\hat{i}, \hat{j}}[\ell] \ge \tilde{t}_{i, j}[\ell] + 1$. 
    Finally, we show that it is impossible for $\tilde{t}_{\hat{i}, \hat{j}}[\ell]$ to exceed $\tilde{t}_{i, j}[\ell] + 1$. If it did, then, using reasoning similar to Inequalities~\eqref{ineq:joint1}, \eqref{ineq:joint2}, and \eqref{ineq:joint3}, we could conclude that there exists some $(\hat{i}, j')$ that appears between $(i, j)$ and $(\hat{i}, \hat{j})$, which leads to a contradiction.

\end{proof}

\begin{proof}[Proof for Lemma~\ref{lemma: formula for the merged group size}]
    The proof is under the same spirit as the proof of Lemma~\ref{lemma: formula for the group size}. 
    
    It suffices to show that 
    \begin{equation}
        \begin{array}{c}
            |\groupSeq{\ge \tau}{i}| 
                = {
                    \prod_{j = 1}^{i - 1} \PAREN{ \cnt{\tau - 1}{j} - (j - 1) }  \cdot \big( d -  \cnt{\tau - 1}{i} \big)
                    \cdot \prod_{j = i + 1}^{k} \PAREN{ d  - (j - 1) }.
                } 
        \end{array}
    \end{equation}
    Recall Equation~\eqref{eq: def of tail groups} that
    \begin{equation*}
        \groupSeq{\ge \tau}{i} \doteq \cup_{r \in \IntSet{\tau}{n}} \groupSeq{r}{i}
        = \set{ 
            \topkSoln = \paren{\topkSoln[1], \ldots, \topkSoln[k]} \in \partialPermutation{\DataDomain}{k} : 
                \begin{matrix}
                    \hist\BigBracket{\topkSoln[j]} > \hist_{(j)} - \tau, & \forall j < i \\
                    \hist\BigBracket{\topkSoln[j]} \le \hist_{(j)} - \tau, & j = i
                \end{matrix}
        }.
    \end{equation*}
    
    Assume we want to select a sequence $\topkSoln \in \groupSeq{\ge \tau}{i}$. 
    Since $\topkSoln[1] \in \{ j' \in [d] : \hist\bracket{j'} > \hist_{(1)} - \tau \}$, 
    the number of possible choices for $\topkSoln[1]$ is 
    $
        | \{ j' \in [d] : \hist\bracket{j'} > \hist_{(1)} - \tau \} |
        = | \{ j' \in [d] : \hist\bracket{j'} \ge \hist_{(1)} - (\tau - 1) \} |
        = \cnt{\tau - 1}{1}.
    $
    The first equality holds because the $\hist[j']$ values are integers.
    
    Since $\hist_{(\ell)}$ is non-decreasing, for each $\ell < j < i$, $\hist[\topkSoln[\ell]] > \hist_{(\ell)} - \tau \geq \hist_{(j)} - \tau$, therefore $\hist_{(\ell)} \in \{ j' \in [d] : \hist\bracket{j'} > \hist_{(j)} - \tau \}$. 
    After determining $\topkSoln\IntSet{1}{j-1}$, $\topkSoln[j]$ must be chosen from $\{ j' \in [d] : \hist\bracket{j'} > \hist_{(j)} - \tau \} \setminus \{\topkSoln[\ell]: \ell < j\}$, so it has $|\{ j' \in [d] : \hist\bracket{j'} > \hist_{(j)} - \tau \}| - (j - 1) = \cnt{\tau - 1}{j} - (j - 1)$ choices. 
    
    Now we consider the number of choices for $\topkSoln[i]$.
    Since $\topkSoln[1], \ldots, \topkSoln[i - 1] \in \{ j' \in [d] : \hist\bracket{j'} > \hist_{(i)} - \tau \}$, they do not appear in $\{ j' \in [d] : \hist\bracket{j'} \le \hist_{(i)} - \tau \}$. 
    The number of choices for $\topkSoln[i]$ is exactly 
    $
        |\{ j' \in [d] : \hist\bracket{j'} \le \hist_{(i)} - \tau \}|
            = d -  \cnt{\tau - 1}{i}.
    $
    
    For $j \in \IntSet{i + 1}{k}$, the number of choices for $\topkSoln[j]$, after determining $\topkSoln\IntSet{1}{j-1}$, is ${ d - (j - 1) }$. 

    Multiplying the number of choices for each element in $\topkSoln \in \groupSeq{\ge \tau}{i}$, we get
    \begin{equation*}
        \begin{array}{c}
            |\groupSeq{r}{i}| 
                = {
                    \prod_{j = 1}^{i - 1} \PAREN{ \cnt{\tau - 1}{j} - (j - 1) }  \cdot \big( d -  \cnt{\tau - 1}{i} \big)
                    \cdot \prod_{j = i + 1}^{k} \PAREN{ d  - (j - 1) }.
                } 
        \end{array}
    \end{equation*}
\end{proof}

\section{Implementation Details}
\label{appendix: implementation}

In this section, we discuss how to implement the \textbf{\textit{\subsetsampling}} and \textbf{\textit{\subsetSequencesampling}} steps, according to the partition $\set{ \groupSeq{r}{i} : r \in \IntSet{0}{\tau - 1}, i \in [k]} \cup \set{ \groupSeq{\ge \tau}{i} :  i \in [k]}$ induced by the loss $\loss{\ourAlgo}$.

\subsection{\subsetsampling}

The algorithm is in Algorithm~\ref{algo: sampling step I}.

\paragraph{Computing the $\cnt{r}{j}$.}
Let $f_{\hist}: \N \rightarrow 2^\N$ be the function given by $f_{\hist}[t] \doteq \{ i \in \cD : \hist[i] = t \}, \forall t \in \N$.
By using standard hash map, $f_{\hist}$ can be computed with $O(d)$ time and space.
Based on the definition of the $\cnt{r}{j}$'s and that $\hist$ consists of only integer scores,  the following recursion holds, 
\begin{equation}
    \label{eq: recursion for the cnt}
    \begin{array}{rll}
        \cnt{0}{1} 
                &= \card{f_{\hist} \bigparen{ \hist_{(1)} }}, \\
        \cnt{0}{j} - \cnt{0}{j - 1} 
                &= \card{f_{\hist} \bigparen{ \hist_{(j)} }} \cdot \indicator{\hist_{(j - 1)} \neq \hist_{(j)}}, 
                & \forall \, 1 < j \le k, \\
        \cnt{r}{j} - \cnt{r - 1}{j}
                &= |f_{\hist} \bigparen{ \hist_{(j)} - r }|
                & \forall \, 1 < r < \tau.
    \end{array}
\end{equation}
Therefore, $\cnt{r}{j}$'s can be computed in $O(d + \tau k)$ time.

\paragraph{Computing $\ln |\groupSeq{r}{i}|$ and $\ln |\groupSeq{\ge\tau}{i}|$.}
To simplify the notation, we apply the following definitions.
\begin{definition}
    \begin{equation}
        \label{eq: formula for unified group sequence}
        \unifiedGroupSeq{r,i} \doteq 
        \begin{cases}
            \groupSeq{r}{i}, &\text{if } r < \tau \\
            \groupSeq{\geq\tau}{i}, &\text{if } r = \tau
        \end{cases}
    \end{equation}
    \begin{equation}
        \label{eq: formula for unified cnt}
        \cntUnified{r,i} \doteq 
        \begin{cases}
            \cnt{r}{i}, &\text{if } r < \tau \\
            d, &\text{if } r = \tau
        \end{cases}
    \end{equation}
\end{definition}

For each $r \in \IntSet{0}{\tau}$ and each $i \in \IntSet{1}{k}$, define the prefix and the suffix sums by 
\begin{equation}
    \label{eq: prefix and suffix sums}
    \begin{array}{c}
        \forwardSum\bracket{r, i} = \sum_{j = 1}^i \ln \PAREN{ \cntUnified{r - 1}{j} - (j - 1) },
        \quad
        \backwardSum\bracket{r, i} = \sum_{j = i}^{k} \ln \PAREN{ \cntUnified{r}{j}  - (j - 1) },
        \quad \forall i \in [k].
    \end{array}
\end{equation}
For convenience, we assume $\forwardSum\bracket{\cdot, 0} = \backwardSum\bracket{\cdot, k + 1} = 0$. 
Combining Equation~\eqref{eq: formula for the group size} and~\eqref{eq: formula for the merged group size}, we have
\begin{equation}
    \label{eq: efficient computation for ln counters}
    \begin{array}{c}
        \ln \card{\unifiedGroupSeq{r}{i}} 
        = \forwardSum\bracket{r, i - 1}
        + \ln \big( \cntUnified{r}{i} -  \cntUnified{r - 1}{i} \big)
        + \backwardSum\bracket{r, i + 1}
    \end{array}
\end{equation}
A corner case is $r = 0$. By definition, $|\unifiedGroupSeq{0}{i}| = 0$ unless $i = 0$. We can set $\cntUnified{-1}{0} = 0$ and the equation still holds. 

\begin{algorithm}[!ht]
    \caption{\subsetsampling}
    \label{algo: sampling step I}
    \begin{algorithmic}[1]
        \Require Histogram $\hist$; Privacy Parameter $\eps$
        \State Compute $f_{\hist}: \N \rightarrow 2^\N$ s.t. $f_{\hist}[t] \doteq \{ i \in \cD : \hist[i] = t \}, \forall t \in \N$.
        \vspace{1mm}
        \State Compute the $\cntUnified{r}{j}$'s according to Equation~\eqref{eq: recursion for the cnt} and Equation~\eqref{eq: formula for unified cnt}
        \vspace{1mm}
        \State Compute the $\forwardSum$'s and $\backwardSum$'s according to Equation~\eqref{eq: prefix and suffix sums}
        \vspace{1mm}
        \State Compute the $\ln |\unifiedGroupSeq{r}{i}|$'s according to Equation~\eqref{eq: efficient computation for ln counters}
        \vspace{2mm}
        \State Sample $(r, i) \leftarrow \argmax \set{X_{r, i} + \ln \bigparen{|\groupSeq{r}{i}| \cdot \exp \PAREN{ - \eps \cdot r / 2}} }$, where $X_{r, i} \sim \GumbelNoise{1}$
        \vspace{2mm}
        \State \Return $(r, i)$
    \end{algorithmic}
\end{algorithm}

\subsection{\subsetSequencesampling}

\begin{algorithm}[!t]
    \caption{\subsetSequencesampling}
    \label{algo: sampling step II}
    \begin{algorithmic}[1]
        \Require{$(r, i)$}
        \Ensure{$\topkSoln \uniformSample \unifiedGroupSeq{r, i}$}
        \State Let $\topkSoln \leftarrow \varnothing$ be an empty length-$k$ array %
        \For{$j \leftarrow 1$ to $i - 1$} 
            \label{line: sampling the head iter}
            \vspace{1mm}
            \State Sample $\topkSoln[j] \uniformSample \{\ell \in \DataDomain : \hist[\ell] > \hist_{(j)} - r \} \setminus \set{\topkSoln[1], \ldots, \topkSoln[j - 1]}$
            \label{line: sampling the head}
        \EndFor
        \vspace{2mm}
        \State Sample an 
        $\topkSoln[j] \uniformSample 
            \begin{cases}
                \{\ell \in \DataDomain : \hist[\ell] = \hist_{(i)} - r \}, &\text{if } r < \tau \\
                \{\ell \in \DataDomain : \hist[\ell] \le \hist_{(i)} - \tau \}, &\text{if } r = \tau
            \end{cases}
        $
        \label{line: sampling error equal}
        \vspace{2mm}
        \For{$j \leftarrow i + 1$ to $k$} 
            \label{line: sampling the tail iter}
            \State Sample an 
            $\topkSoln[j] \uniformSample 
                \begin{cases}
                    \{\ell \in \DataDomain : \hist[\ell] \ge \hist_{(j)} - r \} \setminus \set{\topkSoln[1], \ldots, \topkSoln[j - 1]},  &\text{if } r < \tau \\
                    \DataDomain \setminus \set{\topkSoln[1], \ldots, \topkSoln[j - 1]}, &\text{if } r = \tau
                \end{cases}
            $
            \label{line: sampling the tail}
        \EndFor
        \State \Return $\topkSoln$
    \end{algorithmic}
\end{algorithm}

The algorithm for sequence sampling is Algorithm~\ref{algo: sampling step II}. It follows from the counting argument in Lemma~\ref{eq: formula for the group size} when $r < \tau$ and Lemma~\ref{eq: formula for the merged group size} when $r = \tau$.

\begin{lemma}
    \label{lemma: running time of algo sampling step II}
    Algorithm~\ref{algo: sampling step II} can be implemented in $O(d)$ time.
\end{lemma}

\begin{proof}[Proof of Lemma~\ref{lemma: running time of algo sampling step II}]
    We discuss different sections of pseudo-codes of Algorithm~\ref{algo: sampling step II}, start by the easy ones.

    {\it Case I : Algorithm~\ref{algo: sampling step II}, line~\ref{line: sampling error equal}.} 
    Clearly, this can be implemented in $O(d)$ time. 

    {\it Case II: Algorithm~\ref{algo: sampling step II}, line~\ref{line: sampling the tail iter}-\ref{line: sampling the tail}, when $r = \tau$.} 
    After the first $i$ entries of $\topkSoln$ are determined, we can create an dynamic array, denoted $\vec{a}$, consisting of elements $\DataDomain \setminus \set{\topkSoln[1], \ldots, \topkSoln[i]}$.
    This takes $O(d)$ time.
    Sampling and removing an item from $\vec{a}$ can be done in $O(1)$ time via standard technique, as described in Algorithm~\ref{algo: sampling step II explained} (the $\cA\cS(\cdot)$ procedure).

    {\it Case III: Algorithm~\ref{algo: sampling step II}, line~\ref{line: sampling the head iter}-\ref{line: sampling the head}.} 
    This section can be implemented in $O(d + (k + \tau) \log (k + \tau))$ time, 
    as described in Algorithm~\ref{algo: sampling step II explained} (the efficient sequence sampler procedure).
    It first computes a function $f_{\hist}: \N \rightarrow 2^\N$, s.t., $f_{\hist}(t) \doteq \{ i \in \cD : \hist[i] = t \}$, $\forall t \in \N$.
    By using standard hash map, $f_{\hist}$ can be computed with $O(d)$ time and space.
    Then it finds the set $\cI \doteq \{ t \in \N : f_{\hist}(t) \neq \varnothing \wedge t > \hist_{(k)} - \tau \}$, and store elements in $\cI$ as an array.
    It can be computed in $O(d)$ time. 
    Since an item in $\cI$ must equal one of the values of $\hist_{(1)}, \ldots, \hist_{(k)}, \hist_{(k)} - 1, \ldots, \hist_{(k)} - \tau + 1$, 
    it is easy to see that $\card{\cI} \le k + \tau$.
    So sorting the items in $\cI$ in decreasing order takes $O((k + \tau) \log (k + \tau))$ time. 
    Then the algorithm create an empty dynamic array $\vec{a}$, and a variable $\textsc{pos} = 0.$
    For each $j \in [i - 1]$, before the sampling step (Algorithm~\ref{algo: sampling step II explained}, line~\ref{line: iter end in algo Efficient Sequence Sampler}), we claim the following holds: 
    \begin{itemize}
        \item $\textsc{pos} = \argmax_{z} \cI[z] > \hist_{(j)} - r$
        \item $\vec{a} = \{\ell \in \DataDomain : \hist[\ell] > \hist_{(j)} - r \} \setminus \set{\topkSoln[1], \ldots, \topkSoln[j - 1]}$
    \end{itemize}
    This is true for $j = 1$. 
    Now, assume this is true for the $j$ and consider the case for $j + 1$.
    After the sampling step (Algorithm~\ref{algo: sampling step II explained}, line~\ref{line: iter end in algo Efficient Sequence Sampler}) in the $j$th iteration, we have  $\textsc{pos} = \argmax_{z} \cI[z] > \hist_{(j)} - r$ and $\vec{a} = \{\ell \in \DataDomain : \hist[\ell] > \hist_{(j)} - r \} \setminus \set{\topkSoln[1], \ldots, \topkSoln[j]}$.
    At the $(j + 1)$-th iteration, the inner loop (Algorithm~\ref{algo: sampling step II explained}, lines~\ref{line inner iter start in algo Efficient Sequence Sampler}-\ref{line: before sampling in algo Efficient Sequence Sampler}) 
    increases $\textsc{pos}$ from $z_1 \doteq \argmax_{z} \, \cI[z] > \hist_{(j)} - r)$ 
    to $z_2 \doteq \argmax_{z} \, \cI[z] > \hist_{(j + 1)} - r$, and expand $\vec{a}$ correspondingly. 
    Since $\set{ \cI[z_1], \ldots, \cI[z_2] }$ contains all $t \in \bracket{\hist_{(j + 1)} - r + 1, \hist_{(j)} - r}$ s.t., $f_{\hist} \bigparen{ t } \neq \varnothing$,  after this, $\vec{a}$ becomes
    \begin{align*}
        \vec{a} &= 
            \set{\ell \in \DataDomain : \hist[\ell] > \hist_{(j)} - r } \setminus \set{\topkSoln[1], \ldots, \topkSoln[j]} 
            \bigcup 
            \PAREN{
                \bigcup_{t = \hist_{(j + 1)} - r + 1}^{\hist_{(j)} - r} f_{\hist} \bigparen{ t }
            } \\
            &= 
            \set{\ell \in \DataDomain : \hist[\ell] > \hist_{(j)} - r } \setminus \set{\topkSoln[1], \ldots, \topkSoln[j]} 
            \bigcup \set{ 
                \ell \in [d] : 
                    \hist_{(j)} - r 
                    \ge \hist\bracket{\ell} 
                    > \hist_{(j + 1)} - r 
                } \\
            &= \set{\ell \in \DataDomain : \hist[\ell] > \hist_{(j + 1)} - r } \setminus \set{\topkSoln[1], \ldots, \topkSoln[j]} 
    \end{align*}
    Therefore the invaraints are maintained.

    {\it Case IV: Algorithm~\ref{algo: sampling step II}, line~\ref{line: sampling the tail iter}-\ref{line: sampling the tail}, when $r < \tau$.} 
    Observe that $\{\ell \in \DataDomain : \hist[\ell] \ge \hist_{(j)} - r \} =  \{\ell \in \DataDomain : \hist[\ell] > \hist_{(j)} - r - 1\}$.
    Hence we can use similar sampling technique to Case III.

\end{proof}

\begin{algorithm}[]
    \caption{}
    \label{algo: sampling step II explained}
    \begin{algorithmic}[1]
        \Procedure{Array Sampler~$\cA\cS$}{$\vec{a}$}
        \Require Dynamic array $\vec{a}$
            \State $L \leftarrow$ length of $\vec{a}$
            \State Sample $I \uniformSample [L]$
            \vspace{0.5mm}
            \State Swap $\vec{a}[I]$ and $\vec{a}[L]$
            \vspace{0.5mm}
            \State $ans \leftarrow \vec{a}[L]$
            \vspace{0.5mm}
            \State Remove $\vec{a}[L]$ from $\vec{a}$
            \vspace{0.5mm}
            \State {\bf return} $ans$.
        \EndProcedure
        
        \Statex \vspace{-1mm}
        \Procedure{Efficient Sequence Sampler}{}
        \Require Histogram $\hist$; Parameter $i \in [k]$.
            \State Compute the function $f_{\hist}: \N \rightarrow 2^\N$, s.t.,        $f_{\hist} (t) \doteq \{ \ell \in \cD : \hist[\ell] = t \}, \, \forall t \in \N$
            \State Compute $\cI \leftarrow \{ t \in \N : f_{\hist}(t) \neq \varnothing \, \wedge \, t > \hist_{(k)} - \tau \}$ and store it as an array
            \vspace{0.5mm}
            \State Sort $\cI$ in decreasing order 
            \vspace{0.5mm}
            \State $\vec{a} \leftarrow$ an empty dynamic array
            \vspace{0.5mm}
            \State $\textsc{pos} \leftarrow 0$
            \vspace{0.5mm}
            \For{$j \leftarrow 1$ to $i - 1$} 
            \label{line: iter start in algo Efficient Sequence Sampler}
                \vspace{0.5mm}
                \While{$\textsc{pos} <$ length of $\cI$ $\,\wedge\,$ $\cI \bracket{ \textsc{pos} + 1} > \hist_{(j)} - r$}
                \label{line inner iter start in algo Efficient Sequence Sampler}
                    \vspace{0.5mm}
                    \State Add the items $f_{\hist}\bigparen{\cI \bracket{ \textsc{pos} + 1}}$ to the back of $\vec{a}$
                    \vspace{0.5mm}
                    \State $\textsc{pos} \leftarrow \textsc{pos} + 1$
                \EndWhile
                \vspace{0.5mm}
                \label{line: before sampling in algo Efficient Sequence Sampler}    
                \State $\topkSoln[j] \leftarrow$ $\cA\cS(\vec{a})$
            \label{line: iter end in algo Efficient Sequence Sampler}
            \EndFor
            \State {\bf return} $\topkSoln[1], \ldots, \topkSoln[i - 1]$.
        \EndProcedure
    \end{algorithmic}
\end{algorithm}

\subsection{Vectorization}

Though \ourAlgoName is not implemented yet fully vectorized, we discuss its potential here. 
Given that \ourAlgoName has a runtime of $O(d + k^2 / \eps \cdot \ln d)$, the bottleneck lies in the $O(d)$ component for large datasets. 

The first $O(d)$ part involves computing the groups $f_{\hist}[t] \doteq \{ i \in \cD : \hist[i] = t \}$ for each unique value $t$ in $\hist$. This computation could be vectorized using an appropriate library.

The second $O(d)$ component in \subsetSequencesampling can be eliminated with a careful implementation. 
Recall that in Algorithm~\ref{algo: sampling step II}, a crucial step is to sample elements uniformly at random from the set 
\begin{equation}
    \label{eq: sampling sequence set}
    \{\ell \in \DataDomain : \hist[\ell] > \hist_{(j)} - r \} \setminus \set{\topkSoln[1], \ldots, \topkSoln[j - 1]},
\end{equation}
for a possible value of $r \in \set{ 1, 2, \ldots, \tau }$. 
Sampling from the set $\{\ell \in \DataDomain : \hist[\ell] \ge \hist_{(j)} - r \} \setminus \set{\topkSoln[1], \ldots, \topkSoln[j - 1]}$ can be handled similarly. 

To achieve this, we construct an array of at most $k + \tau$ buckets (the cost of constructing this array is covered by the initial $O(d)$ time cost):
$$
    \BigBracket{
        f_{\hist}[\hist_{(1)}], f_{\hist}[\hist_{(2)}], \ldots, f_{\hist}[\hist_{(k)}], f_{\hist}[\hist_{(k)} - 1], \ldots, f_{\hist}[\hist_{(k)} - \tau]
    }.
$$
Assume that each $f_{\hist}[t]$ in this array is itself managed by a dynamic array. 
Sampling from~\eqref{eq: sampling sequence set} is then equivalent to sampling uniformly from a prefix of buckets without replacement. 

The sampling process first selects a bucket with probability proportional to its size, then draws an element uniformly at random from that bucket. After sampling, the chosen element is removed from the bucket, which can be managed efficiently using a dynamic array. This approach removes the dependency on $d$ in the sampling step.

\clearpage

\section{Supplementary Plots}
\label{sec: supplementary plots}

In this section, we provide supplementary plots for our experiments:
\begin{itemize}[leftmargin=4.5mm, topsep=2pt, itemsep=2pt, partopsep=2pt, parsep=2pt]
    \item Figure~\ref{fig: count-diff} illustrates the gaps between large-score items for all tested datasets.
    \item Figure~\ref{fig: complete eps results} displays the algorithm's running time, \(\ell_\infty\) error, and \(\ell_1\) error versus \(\epsilon\).
    \item Figure~\ref{fig: complete delta results} showcases the algorithm's running time, \(\ell_\infty\) error, and \(\ell_1\) error versus \(\beta\).
    \item Figure~\ref{fig: joint without time results} depicts the running time of \joint (excluding time from the \subsetSequencesampling step) versus the running time of our proposed algorithm \ourAlgoName (including time from the \subsetSequencesampling step), over all tested datasets. 
    Given this, our algorithm still runs orders of magnitude faster than \joint. Due to time constraints, we only repeated the experiments 5 times to generate the plots. 
    This is acceptable since, according to the previous experiments, the running time of the algorithms is quite stable.
\end{itemize}

\vspace{3cm}
\begin{figure}[!h]
    \centering
    \makebox[\textwidth]{
        \includegraphics[scale=0.32]{plot-count-diff/games_data_dist.pdf}
        \hspace{-2mm}
        \includegraphics[scale=0.32]{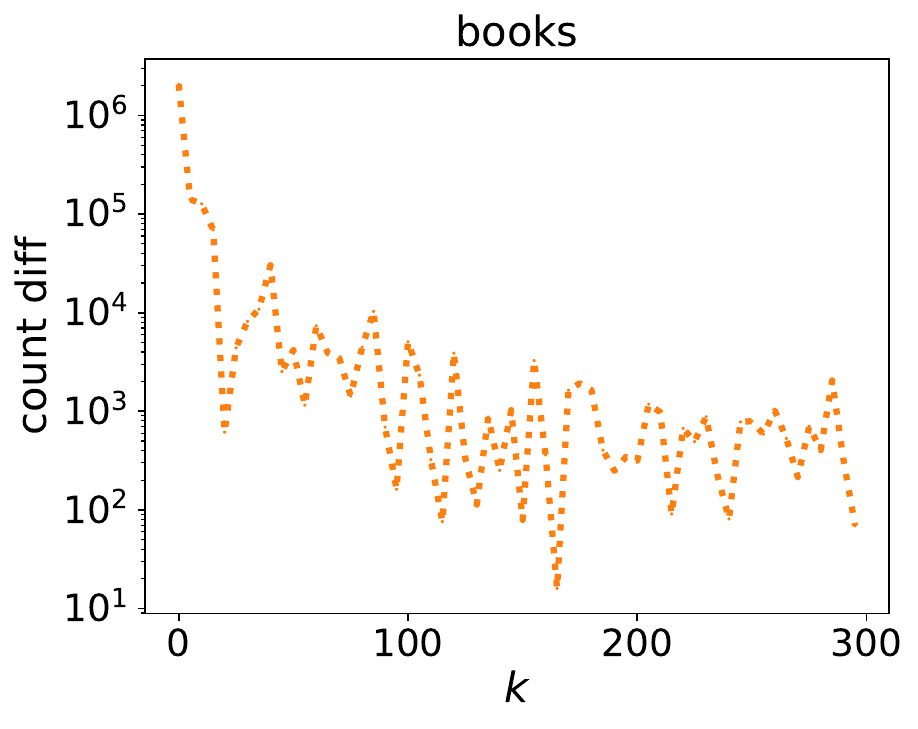}
        \hspace{-2mm}
        \includegraphics[scale=0.32]{plot-count-diff/news_data_dist.pdf} 
    } \\ \vspace{-10pt}
    \makebox[\textwidth]{
        \includegraphics[scale=0.32]{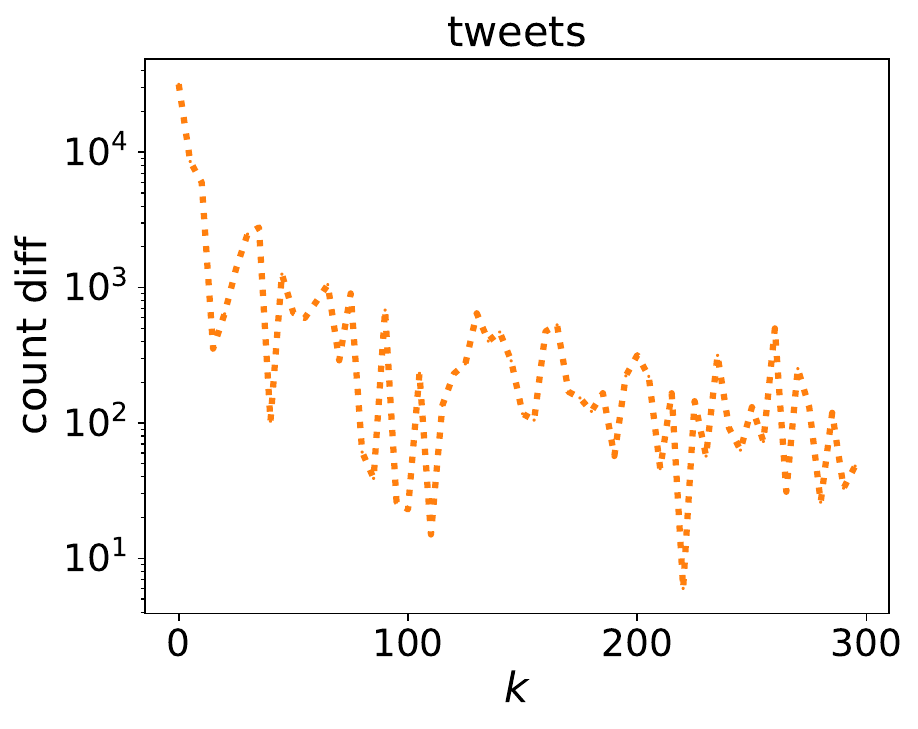}
        \hspace{-2mm}
        \includegraphics[scale=0.32]{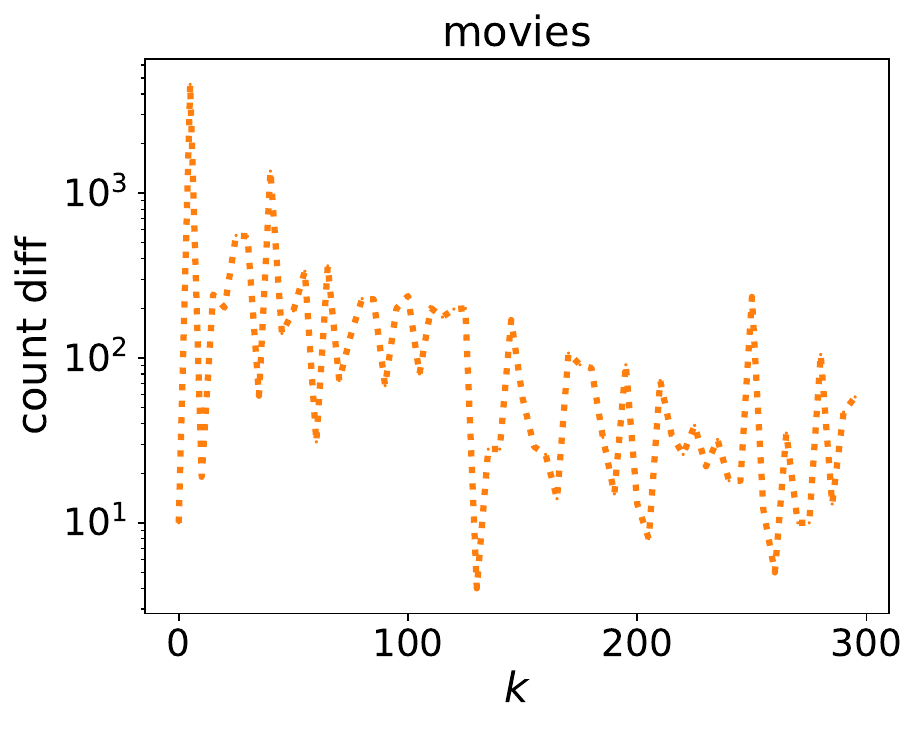}
        \hspace{-2mm}
        \includegraphics[scale=0.32]{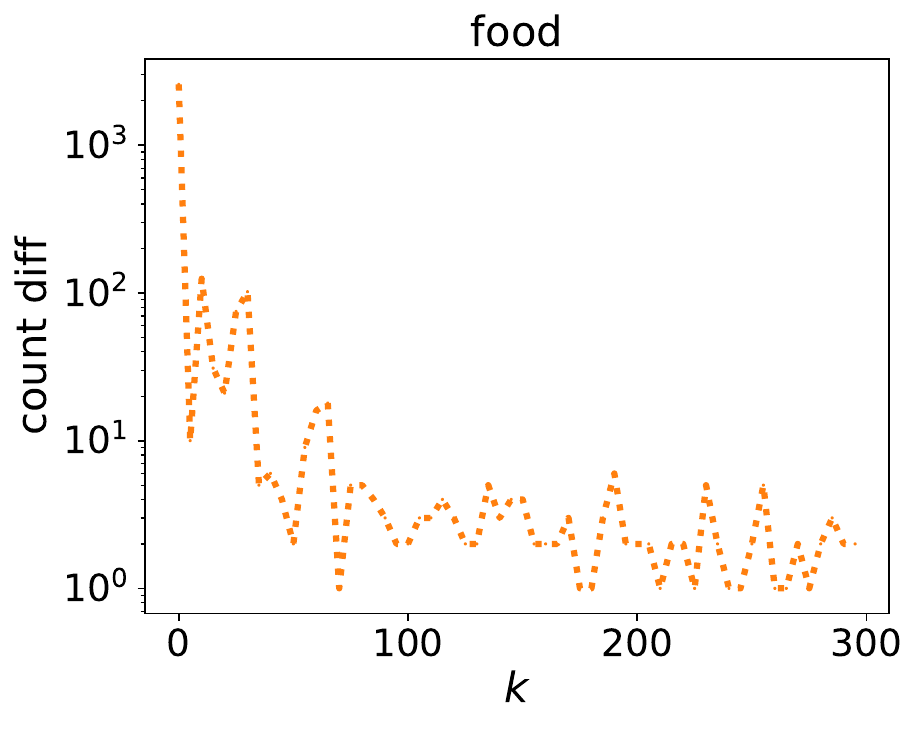} 
    }\\ \vspace{-10pt}
    \caption{
        The gaps between the top-$k$ scores (for $k = 300$) for all tested datasets. 
    }
    \label{fig: count-diff}
\end{figure}

\begin{figure}[]
    \centering
    \makebox[\textwidth]{
        \includegraphics[scale=0.32]{figures/games_var_eps_time.pdf}
        \hspace{-2mm}
        \includegraphics[scale=0.32]{figures/games_var_eps_L_INF.pdf}
        \hspace{-2mm}
        \includegraphics[scale=0.32]{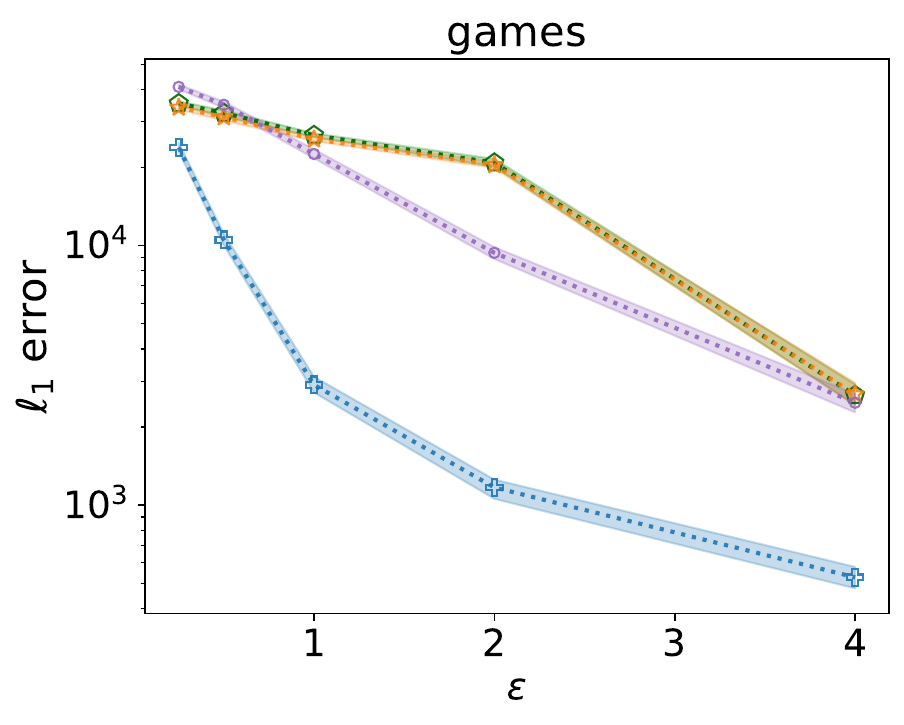} 
    } \\ \vspace{-10pt}
    \makebox[\textwidth]{
        \includegraphics[scale=0.32]{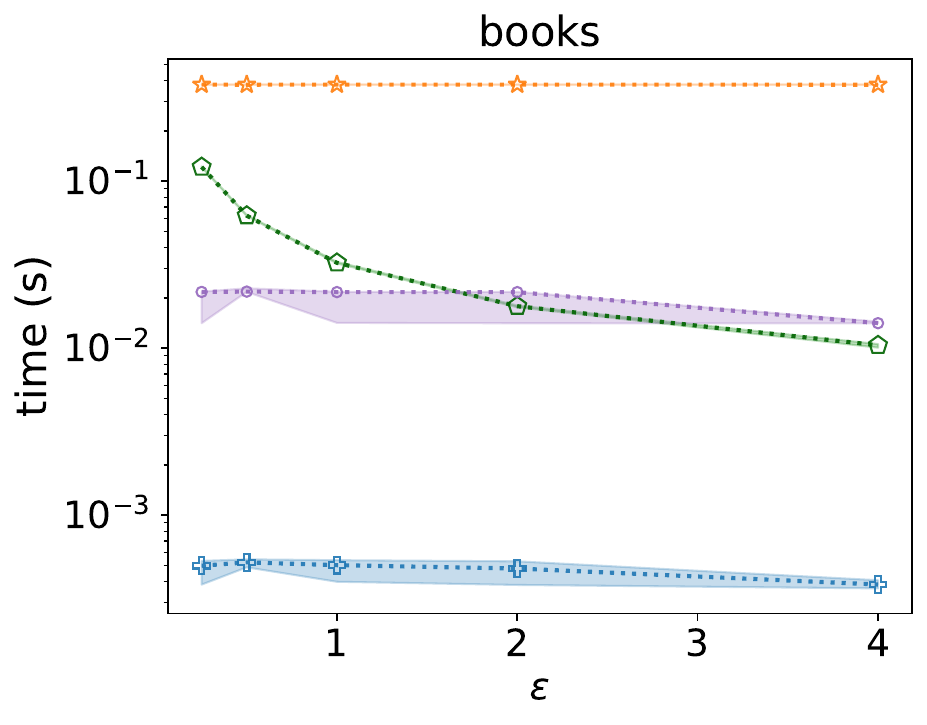}
        \hspace{-2mm}
        \includegraphics[scale=0.32]{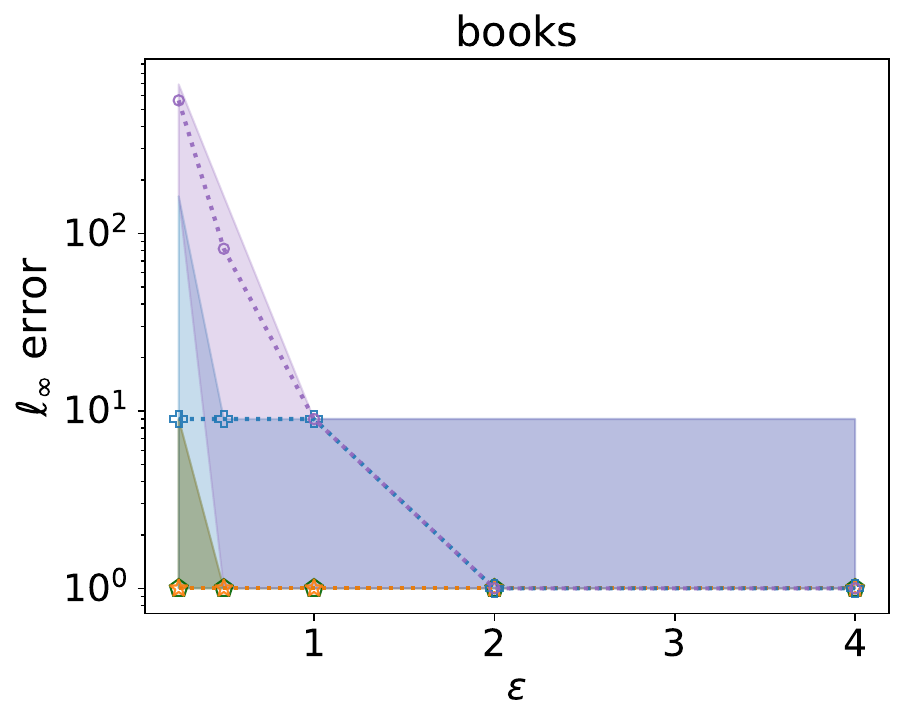}
        \hspace{-2mm}
        \includegraphics[scale=0.32]{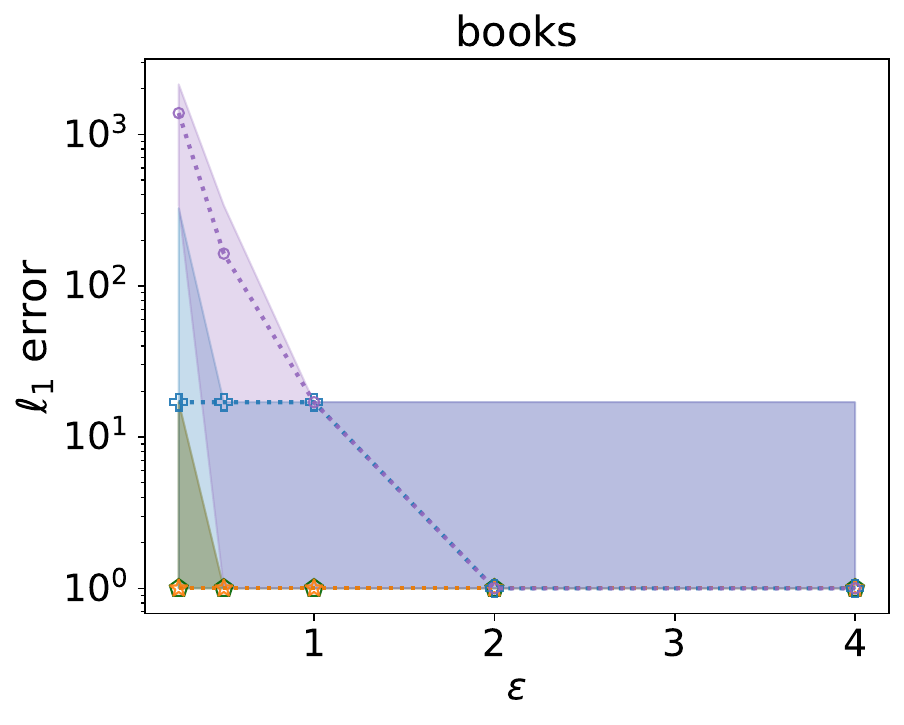} 
    }\\ \vspace{-10pt}
    \makebox[\textwidth]{
        \includegraphics[scale=0.32]{figures/news_var_eps_time.pdf}
        \hspace{-2mm}
        \includegraphics[scale=0.32]{figures/news_var_eps_L_INF.pdf}
        \hspace{-2mm}
        \includegraphics[scale=0.32]{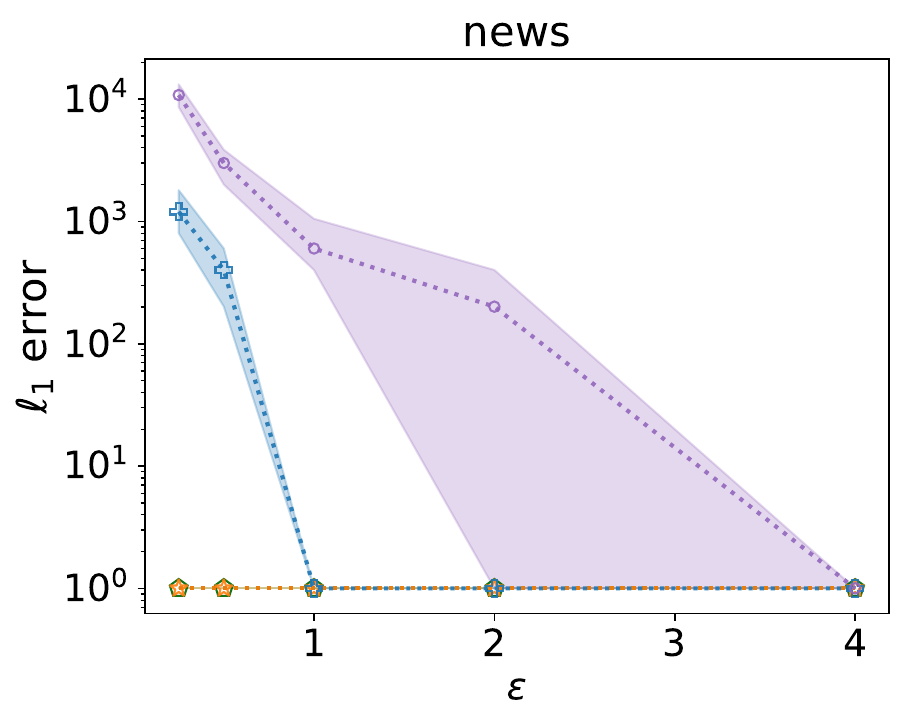} 
    } \\ \vspace{-10pt}
    \makebox[\textwidth]{
        \includegraphics[scale=0.32]{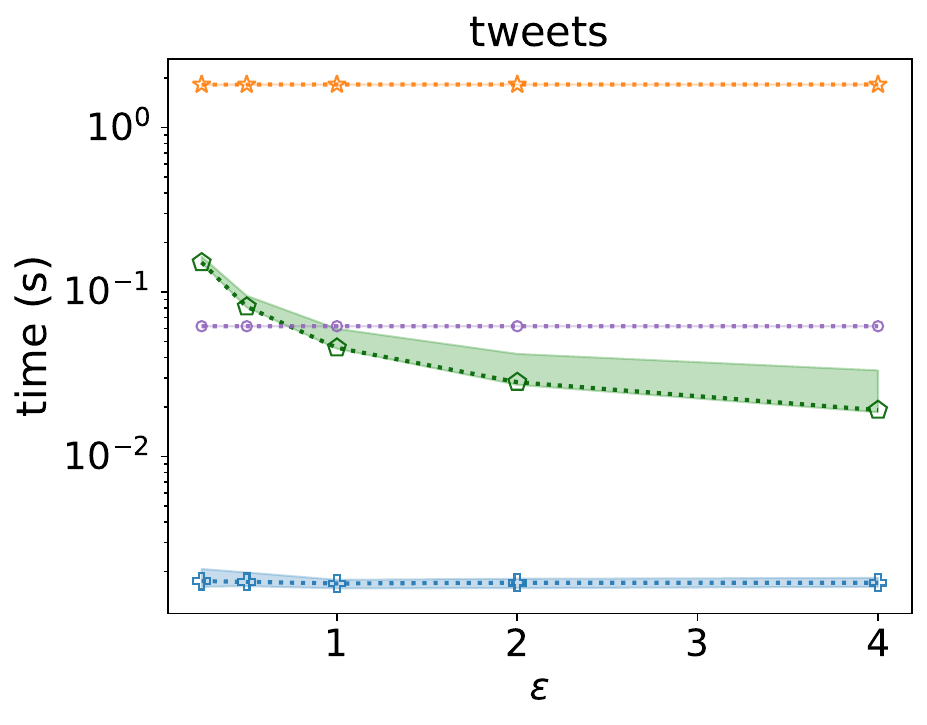}
        \hspace{-2mm}
        \includegraphics[scale=0.32]{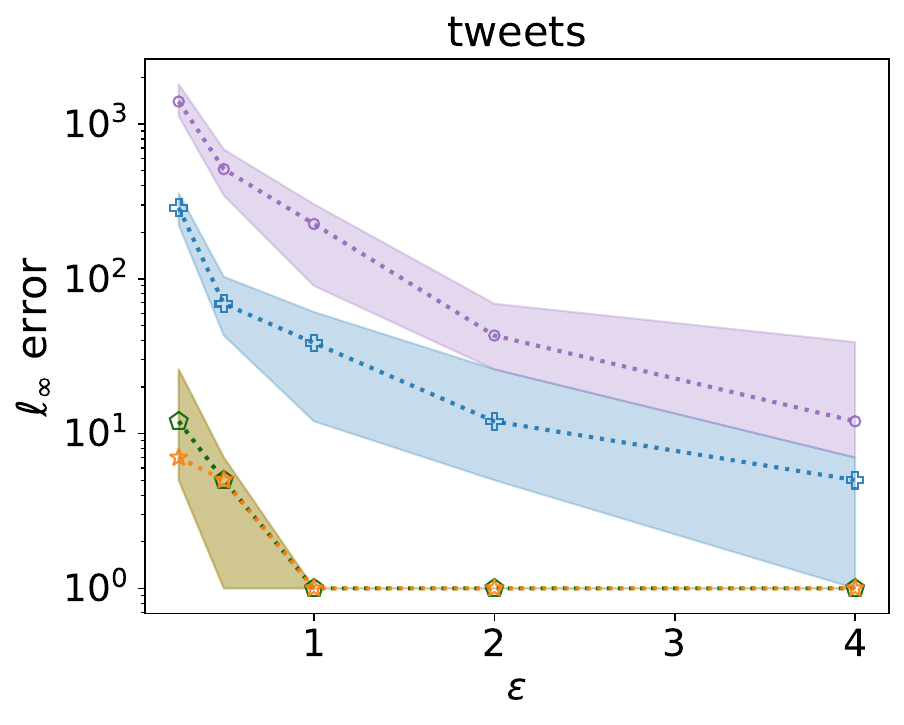}
        \hspace{-2mm}
        \includegraphics[scale=0.32]{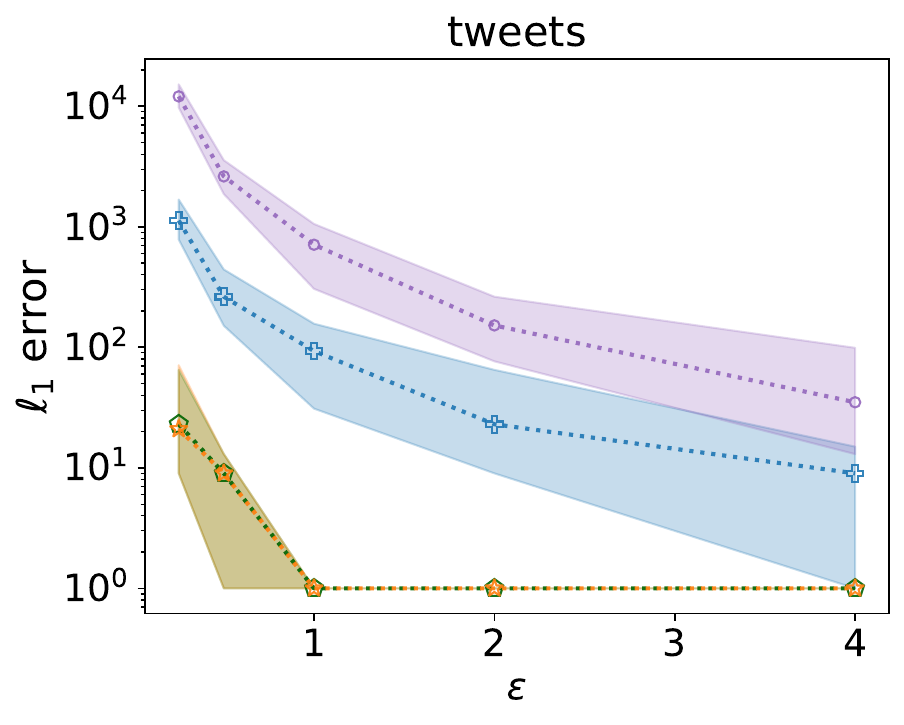} 
    } \\ \vspace{-10pt}
    \makebox[\textwidth]{
        \includegraphics[scale=0.32]{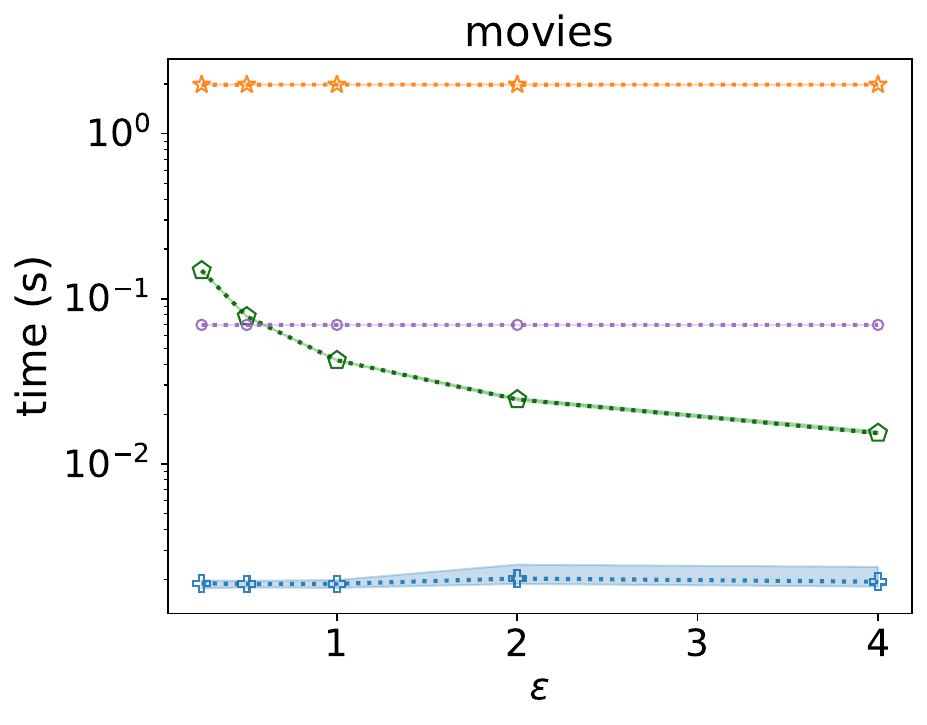}
        \hspace{-2mm}
        \includegraphics[scale=0.32]{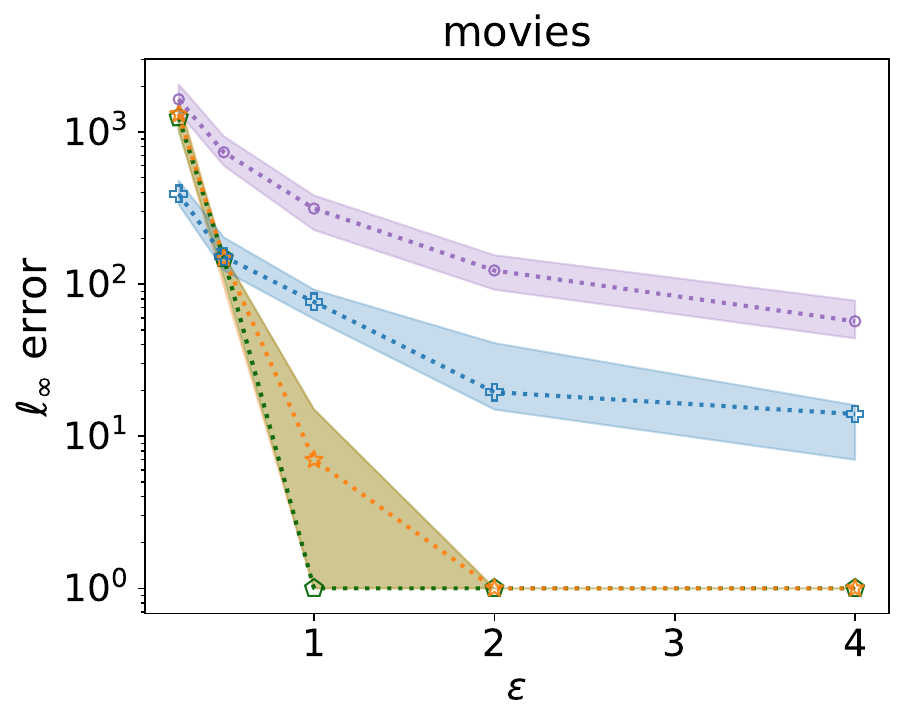}
        \hspace{-2mm}
        \includegraphics[scale=0.32]{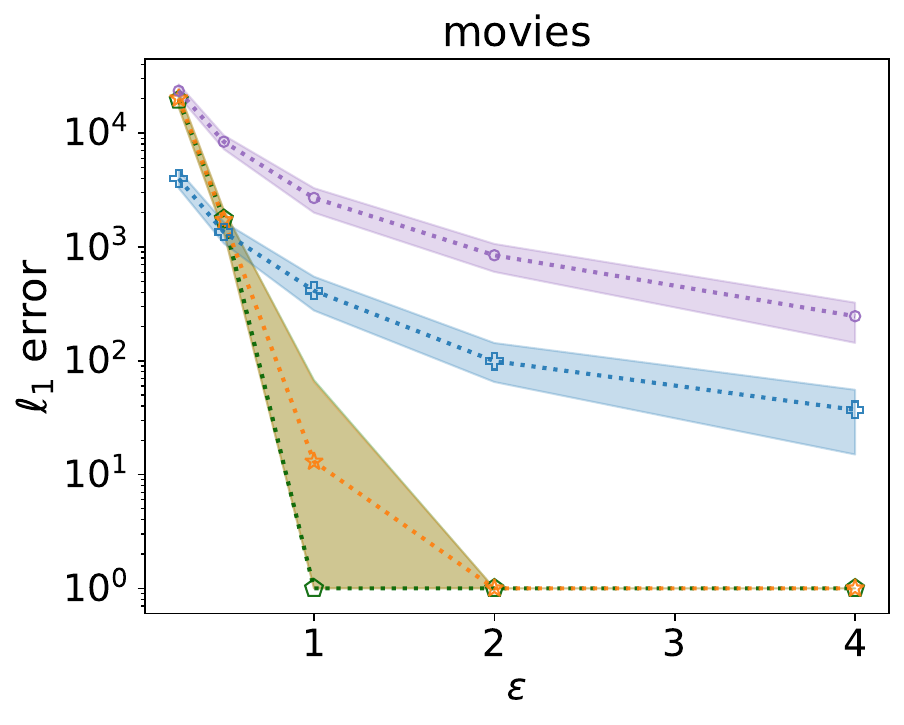} 
    } \\ \vspace{-10pt}
    \makebox[\textwidth]{
        \includegraphics[scale=0.32]{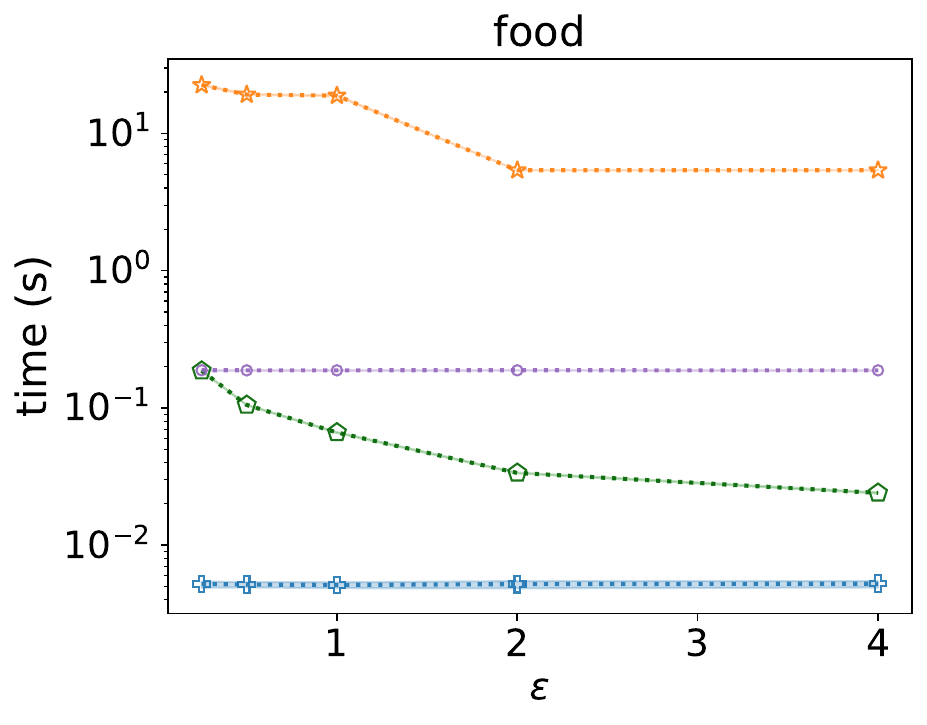}
        \hspace{-2mm}
        \includegraphics[scale=0.32]{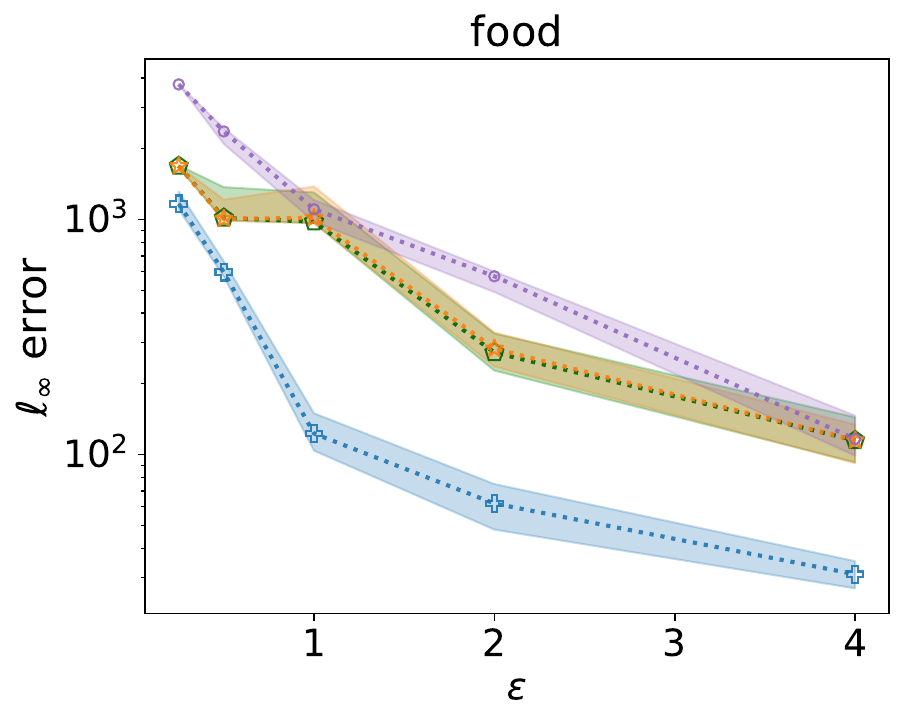}
        \hspace{-2mm}
        \includegraphics[scale=0.32]{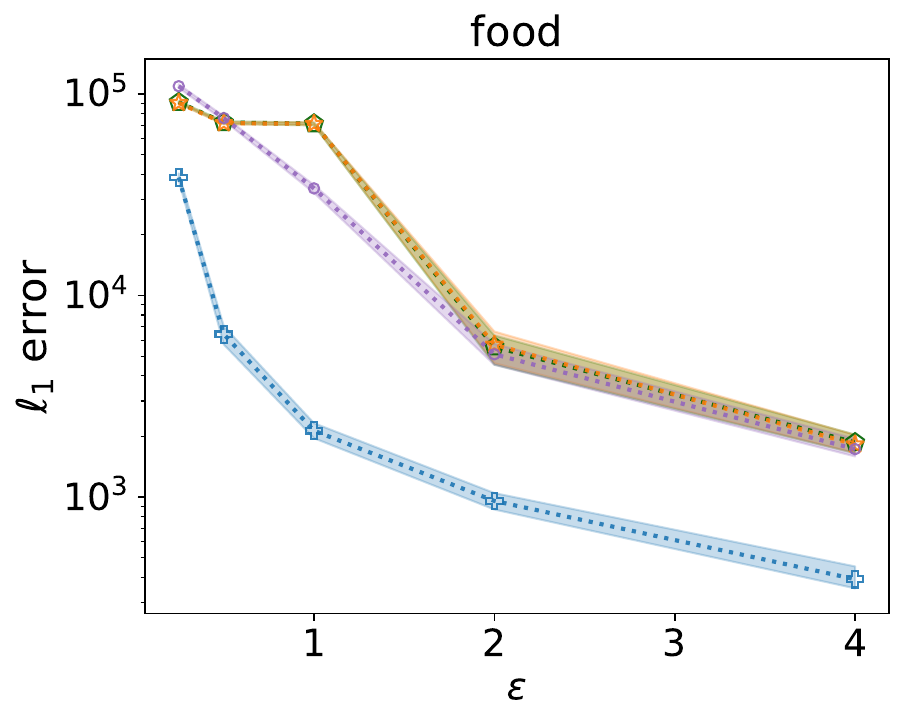} 
    } \\ 
    \includegraphics[scale=0.5, trim={0 9mm 0 9mm}, clip]{figures/legend.pdf}
    \vspace{-4pt}
    \caption{
        \centering
        \textbf{Left}: Running time vs $\eps$. \,
        \textbf{Center}: $\ell_\infty$ error vs $\eps$. \,
        \textbf{Right}: $\ell_1$ error vs $\eps$. \hspace{3cm}
        The $\ell_1/\ell_\infty$ plots are padded by $1$ to avoid $\log 0$ on the $y$-axis. 
    }
    \label{fig: complete eps results}
\end{figure}

\begin{figure}[]
    \centering
    \makebox[\textwidth]{
        \includegraphics[scale=0.32]{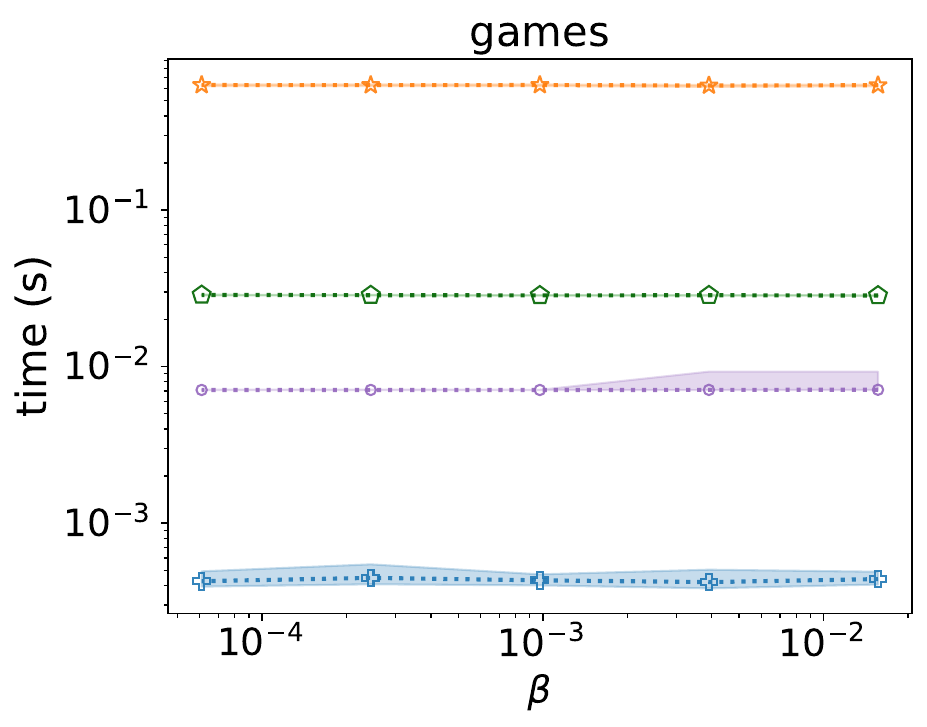}
        \hspace{-2mm}
        \includegraphics[scale=0.32]{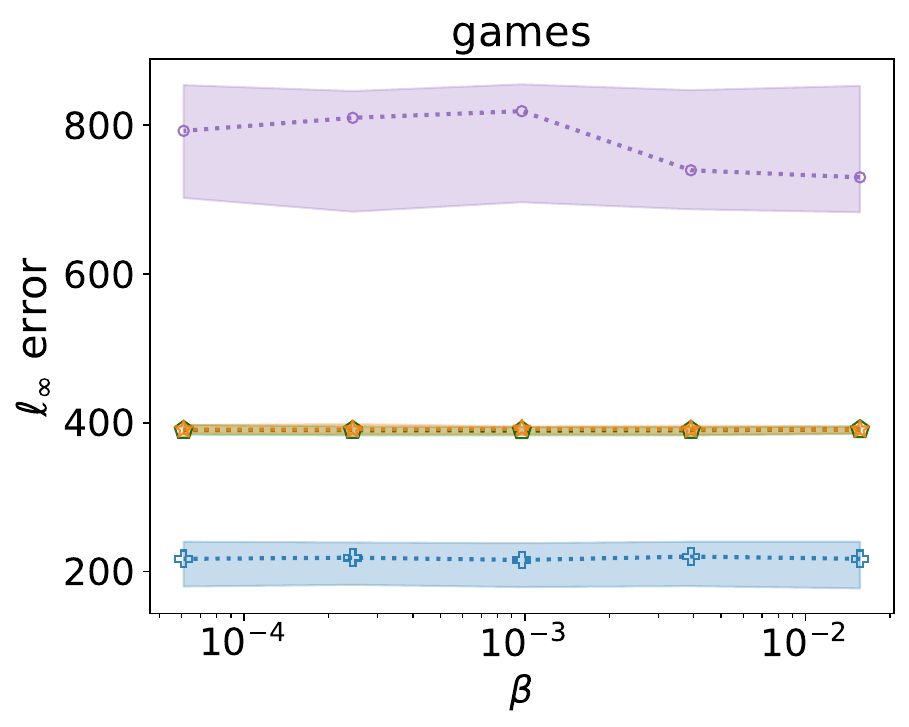}
        \hspace{-2mm}
        \includegraphics[scale=0.32]{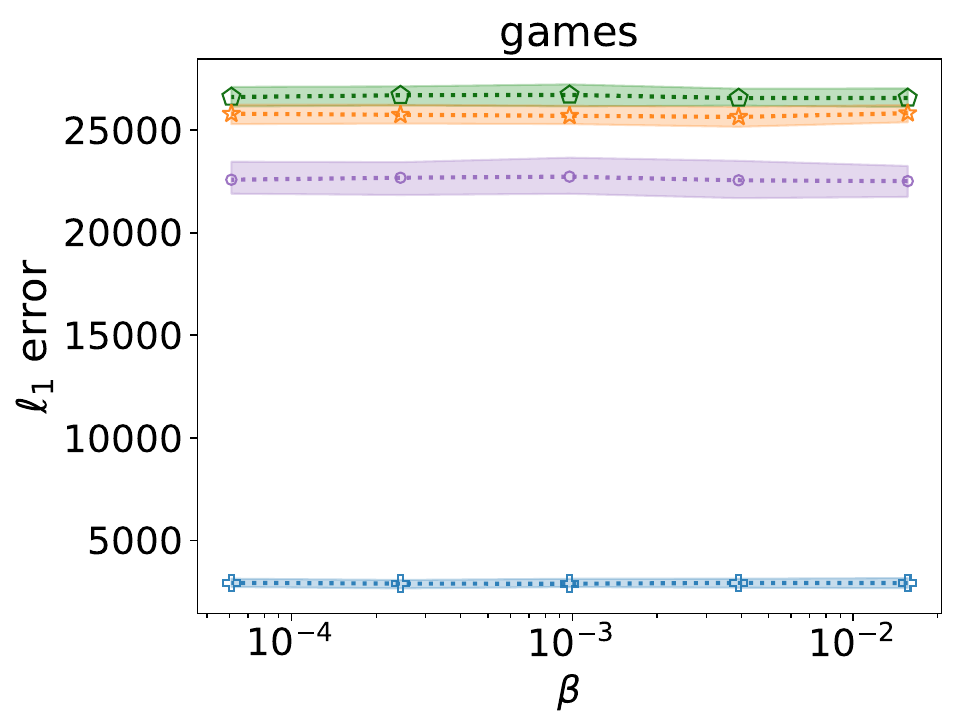} 
    } \\ \vspace{-10pt}
    \makebox[\textwidth]{
        \includegraphics[scale=0.32]{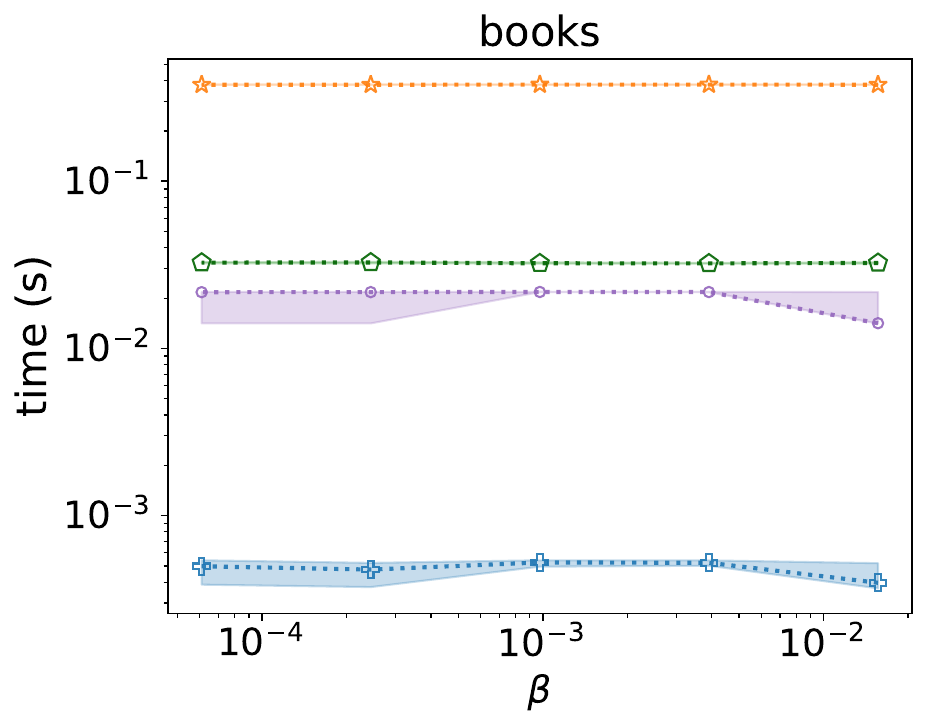}
        \hspace{1.8mm}
        \includegraphics[scale=0.32]{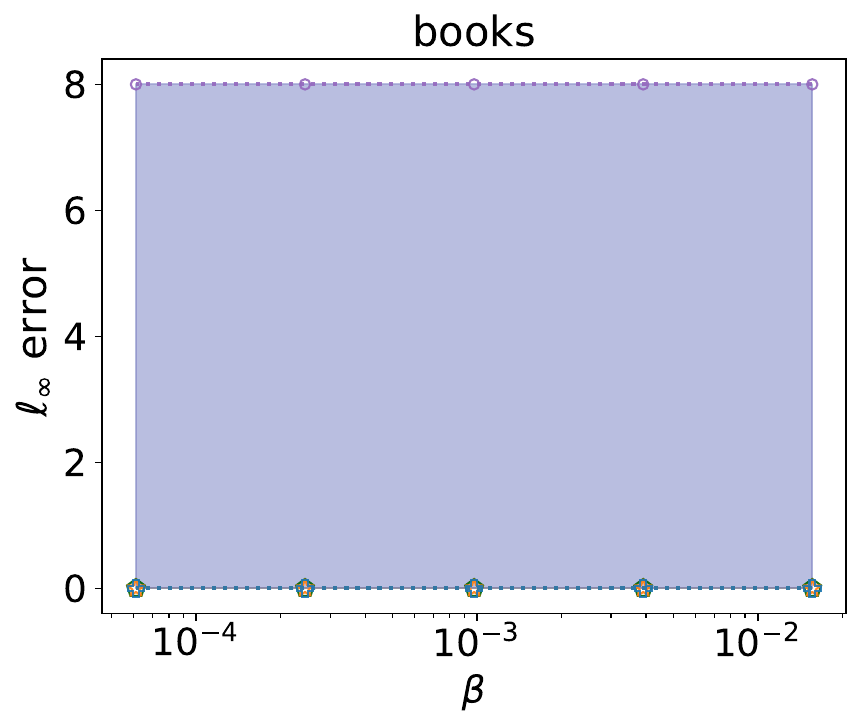}
        \hspace{2mm}
        \includegraphics[scale=0.32]{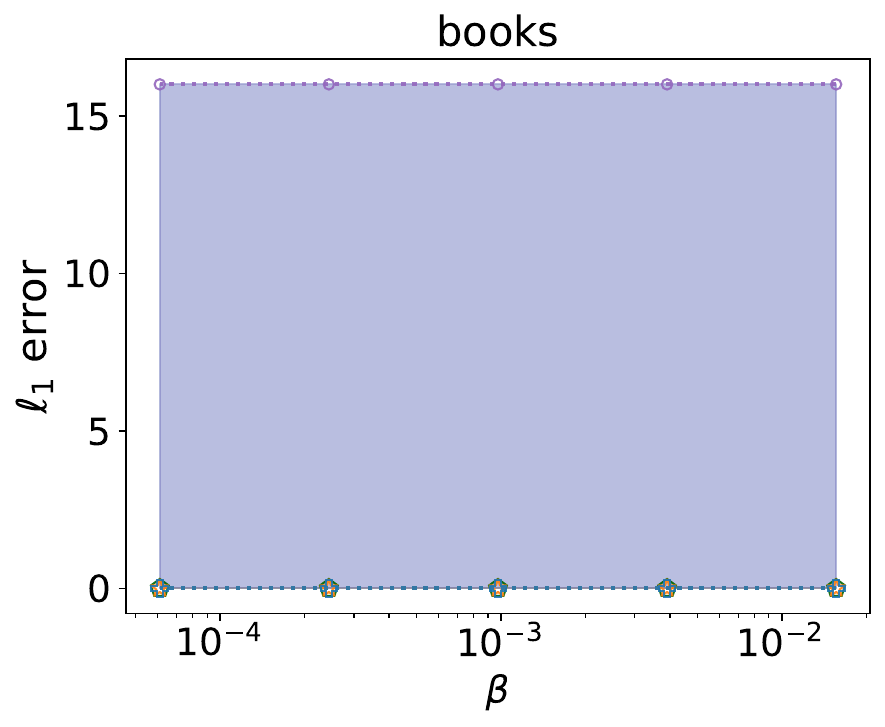} 
    }\\ \vspace{-10pt}
    \makebox[\textwidth]{
        \includegraphics[scale=0.32]{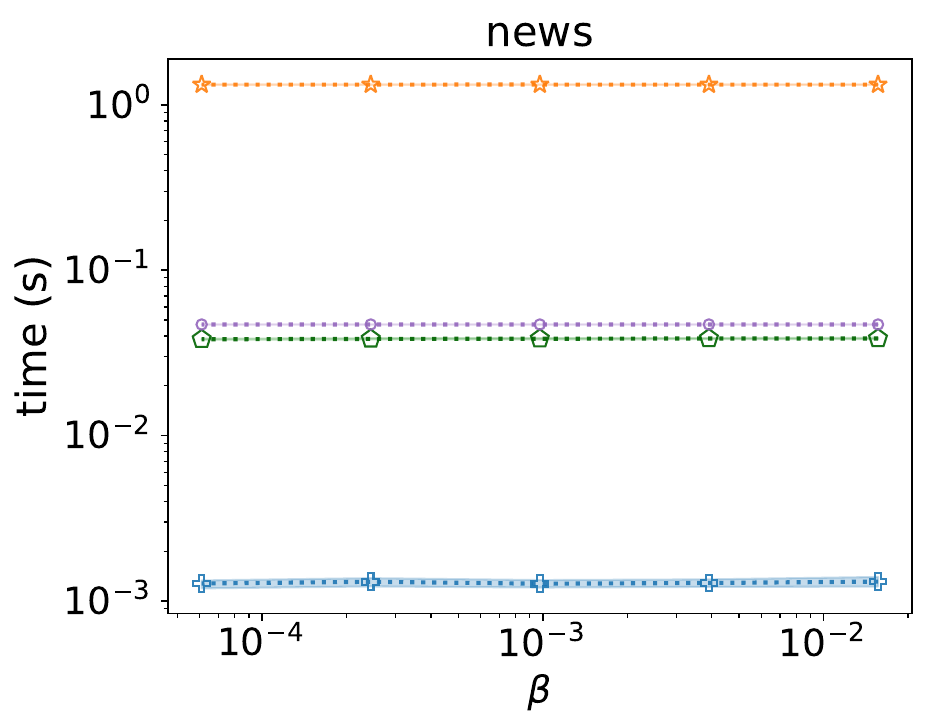}
        \hspace{-2mm}
        \includegraphics[scale=0.32]{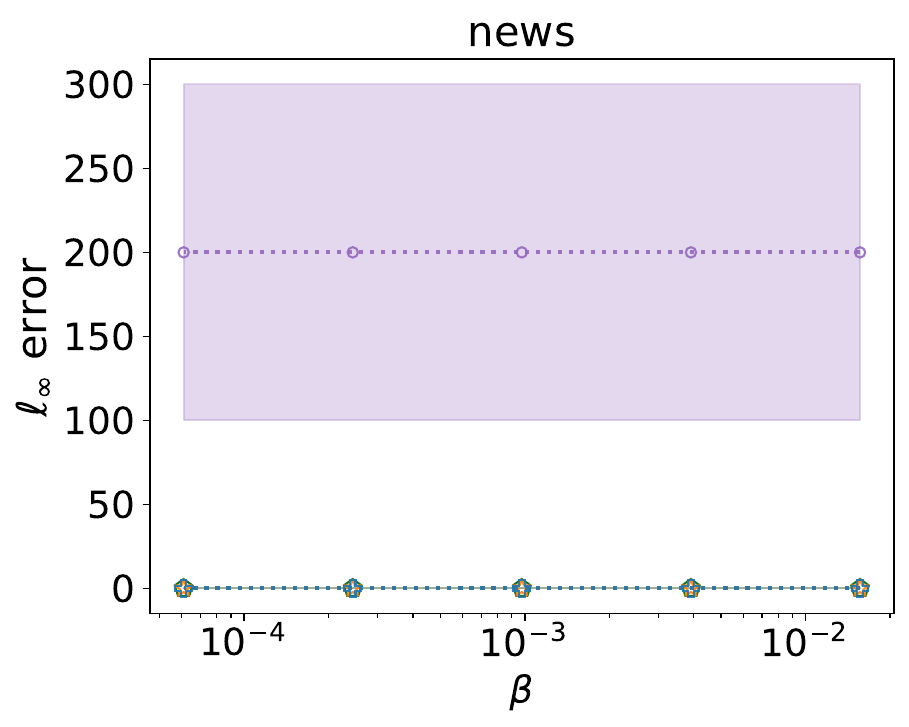}
        \hspace{-2mm}
        \includegraphics[scale=0.32]{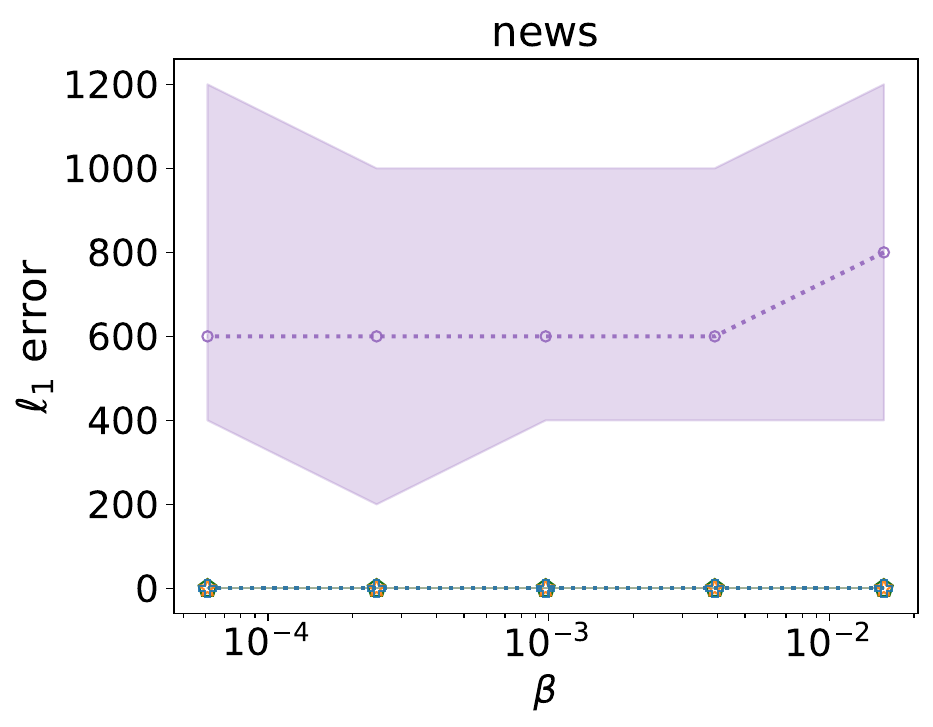} 
    } \\ \vspace{-10pt}
    \makebox[\textwidth]{
        \includegraphics[scale=0.32]{figures/tweets_var_beta_time.pdf}
        \hspace{-2mm}
        \includegraphics[scale=0.32]{figures/tweets_var_beta_L_INF.pdf}
        \hspace{-2mm}
        \includegraphics[scale=0.32]{figures/tweets_var_beta_L_1.pdf} 
    } \\ \vspace{-10pt}
    \makebox[\textwidth]{
        \includegraphics[scale=0.32]{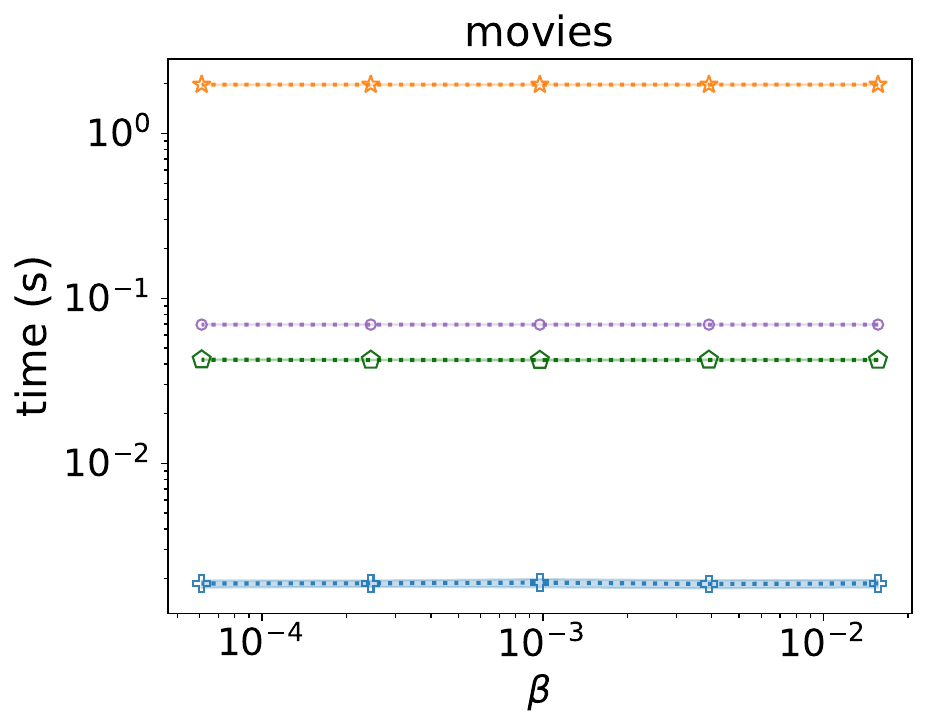}
        \hspace{-2mm}
        \includegraphics[scale=0.32]{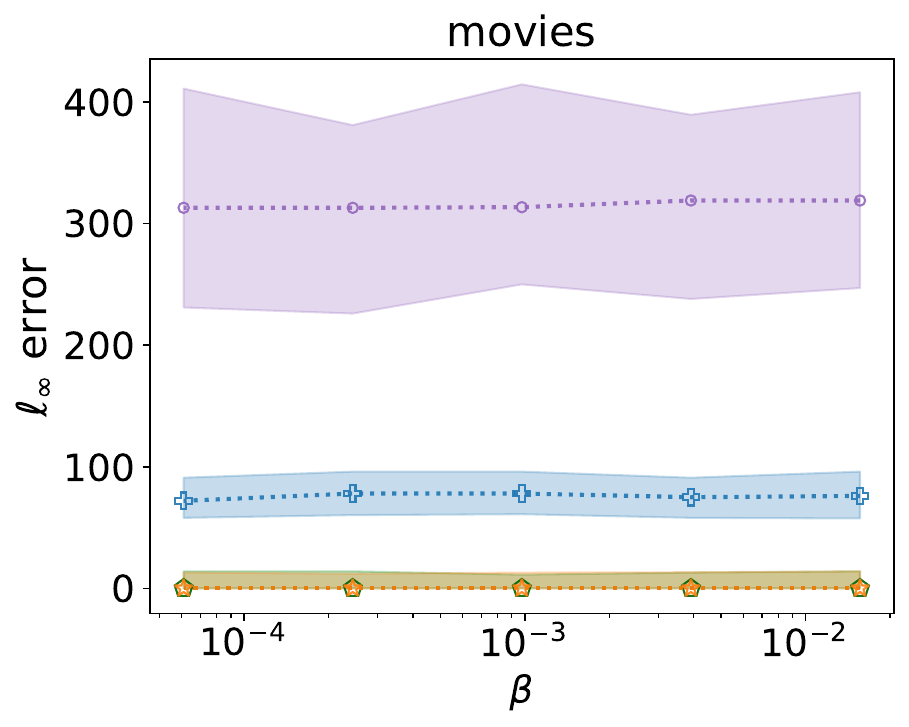}
        \hspace{-2mm}
        \includegraphics[scale=0.32]{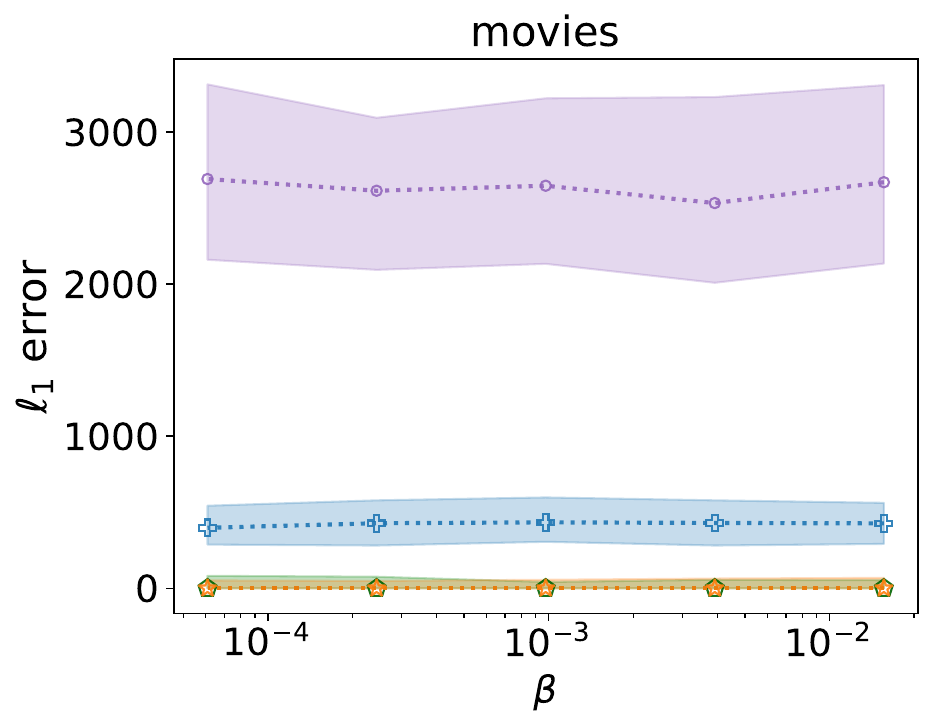} 
    } \\ \vspace{-10pt}
    \makebox[\textwidth]{
        \includegraphics[scale=0.32]{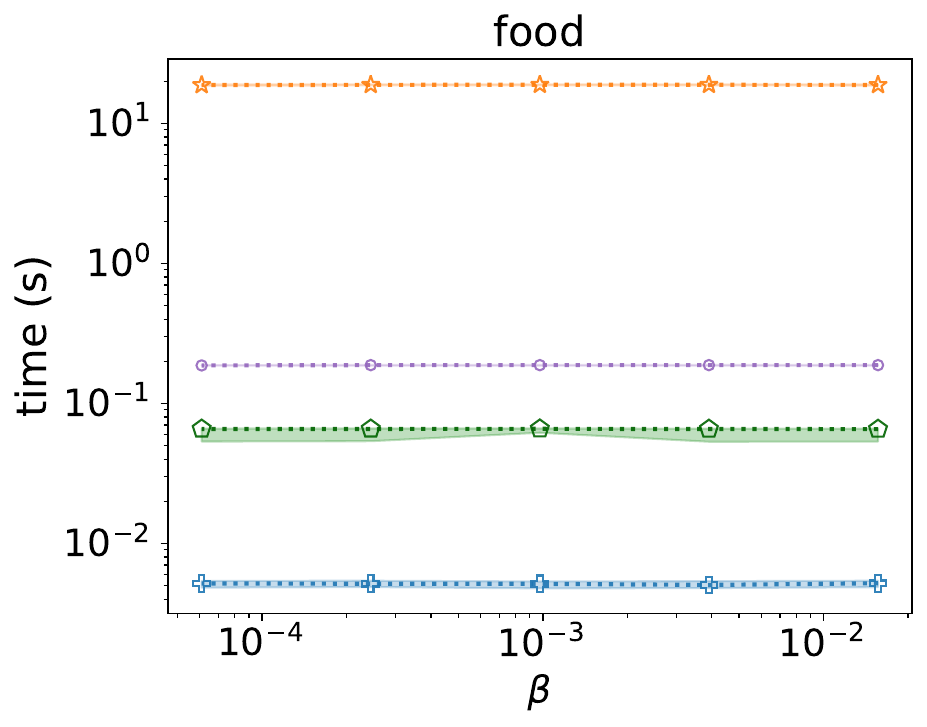}
        \hspace{-2mm}
        \includegraphics[scale=0.32]{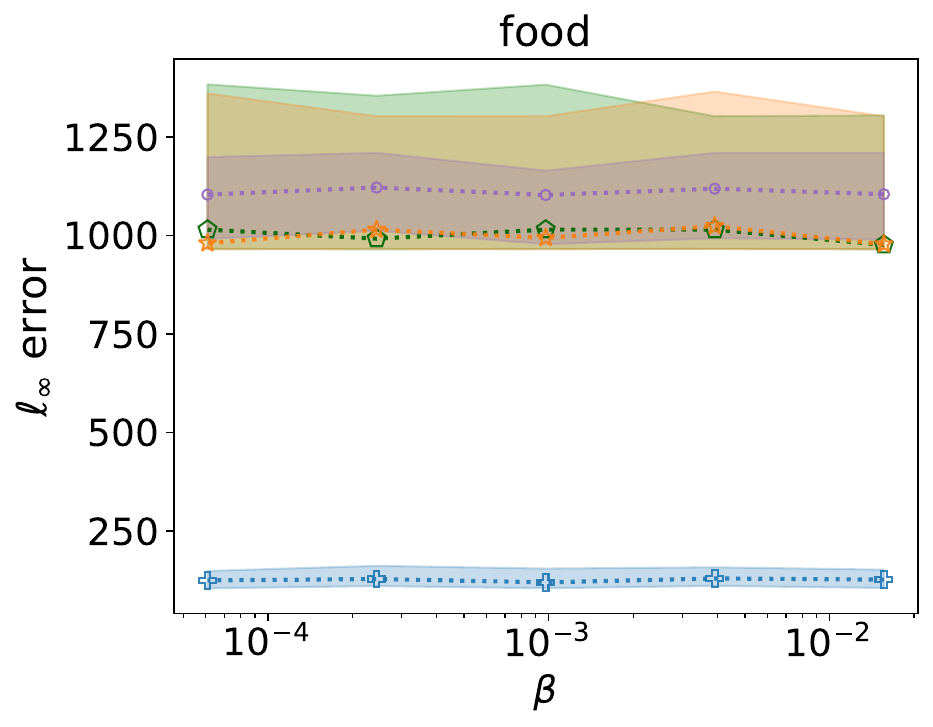}
        \hspace{-2mm}
        \includegraphics[scale=0.32]{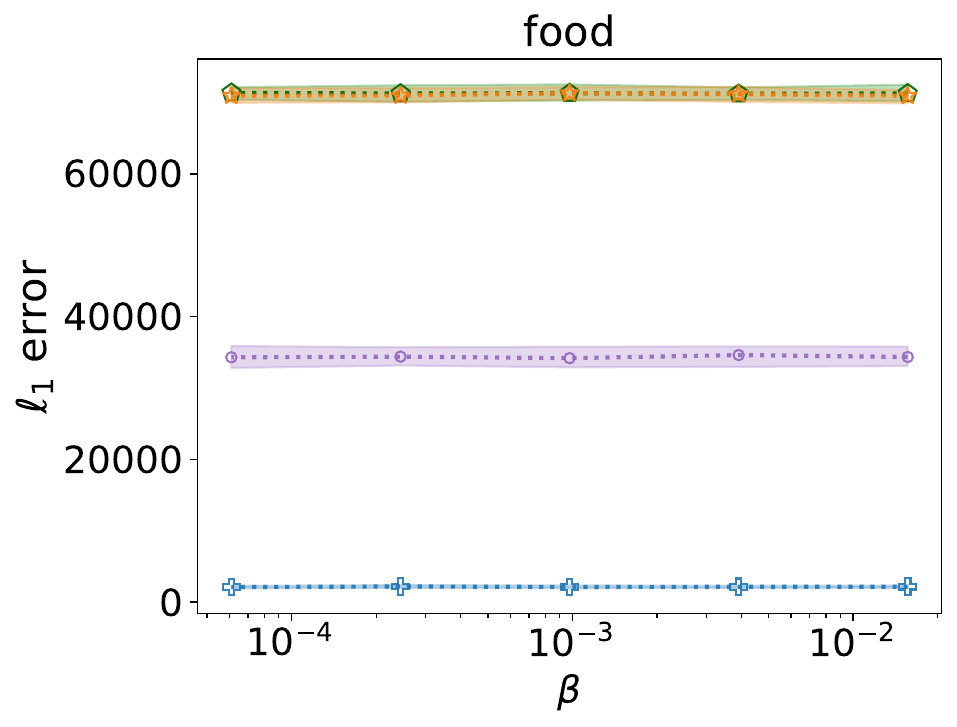} 
    } \\ \vspace{-2pt}
    \includegraphics[scale=0.5, trim={0 9mm 0 9mm}, clip]{figures/legend.pdf}
    \vspace{-6pt}
    \caption{
        \centering
        \textbf{Left}: Running time vs $\beta$. \,
        \textbf{Center}: $\ell_\infty$ error vs $\beta$. \,
        \textbf{Right}: $\ell_1$ error vs $\beta$. \hspace{3cm}
        The $\ell_1/\ell_\infty$ plots are padded by $1$ to avoid $\log 0$ on the $y$-axis. 
    }
    \label{fig: complete delta results}
\end{figure}

\begin{figure}[]
    \centering
    \makebox[\textwidth]{
        \includegraphics[scale=0.32]{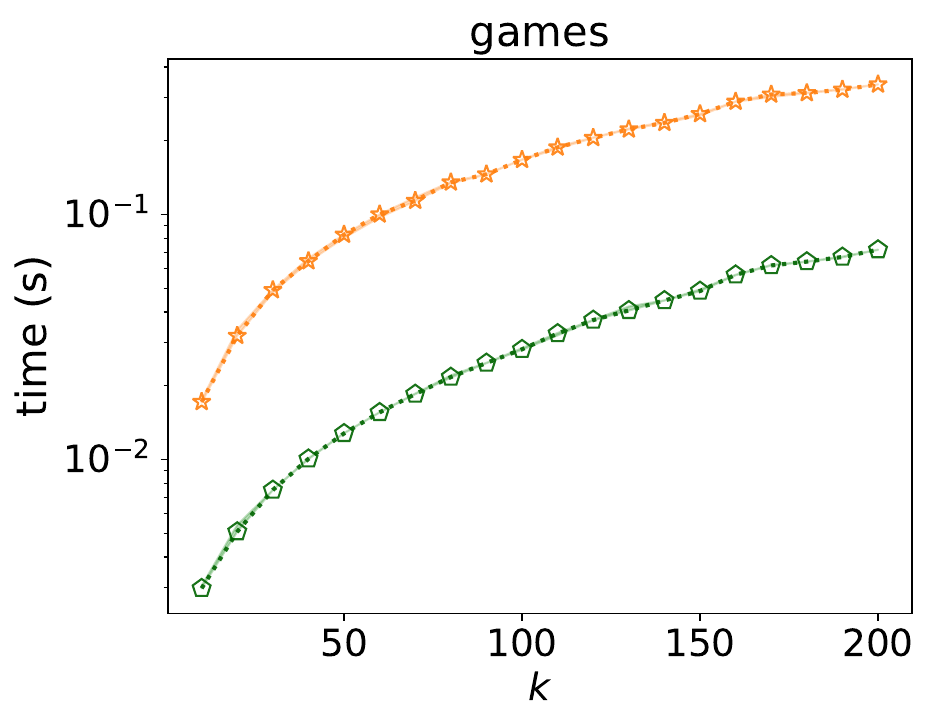}
        \hspace{-2mm}
        \includegraphics[scale=0.32]{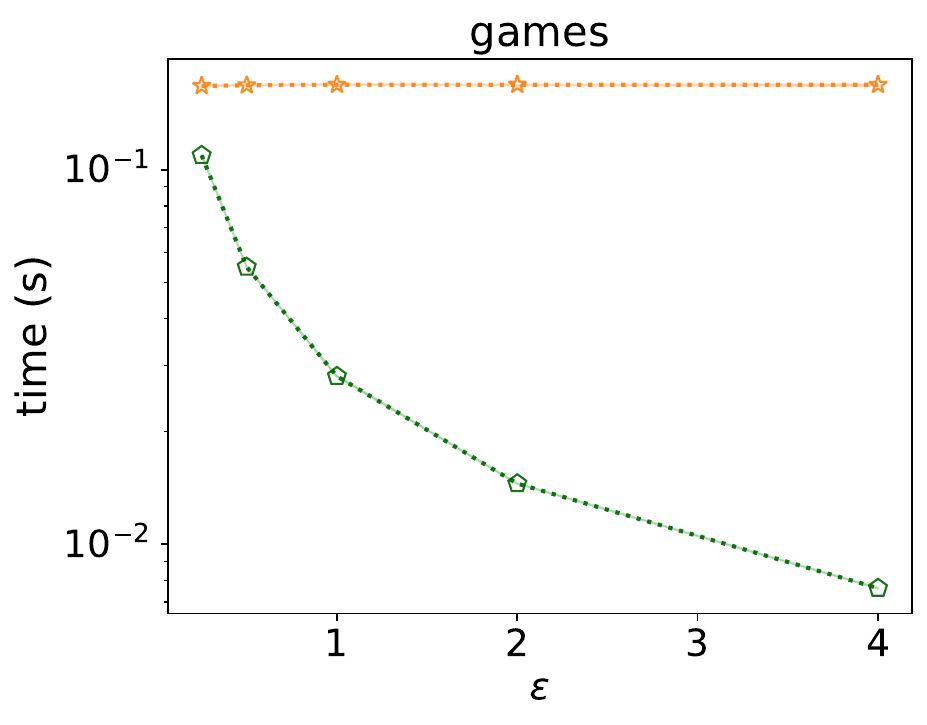}
        \hspace{-2mm}
        \includegraphics[scale=0.32]{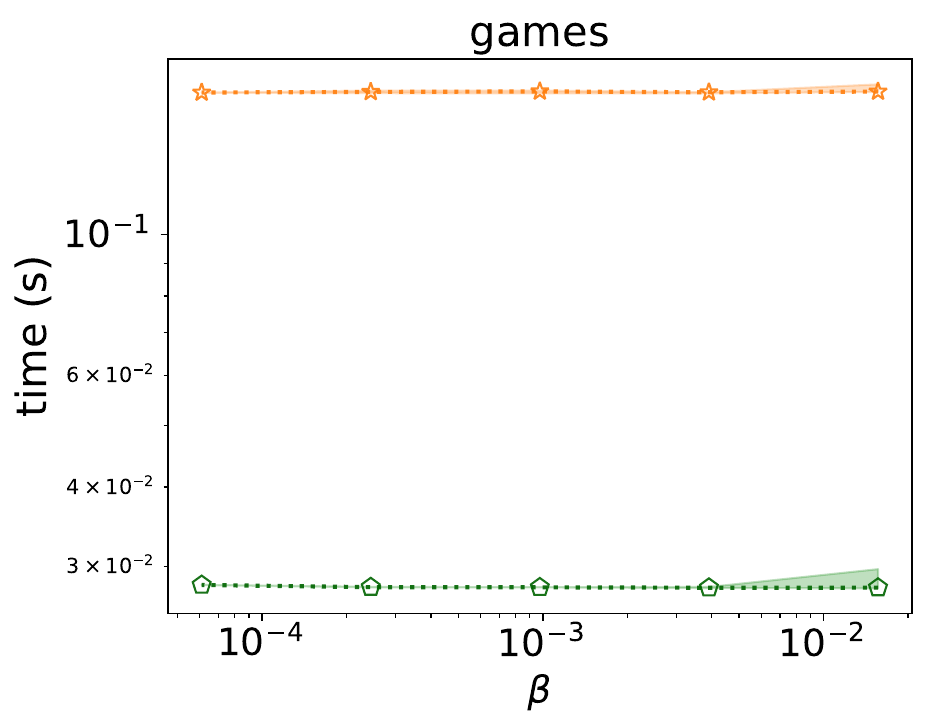} 
    } \\ \vspace{-10pt}
    \makebox[\textwidth]{
        \includegraphics[scale=0.32]{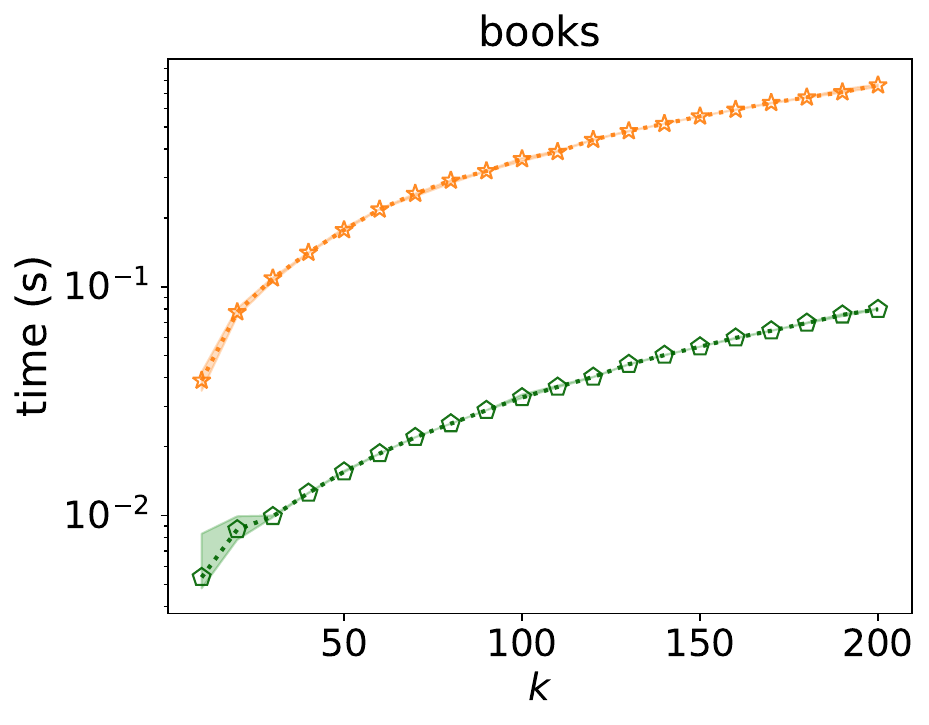}
        \hspace{-2mm}
        \includegraphics[scale=0.32]{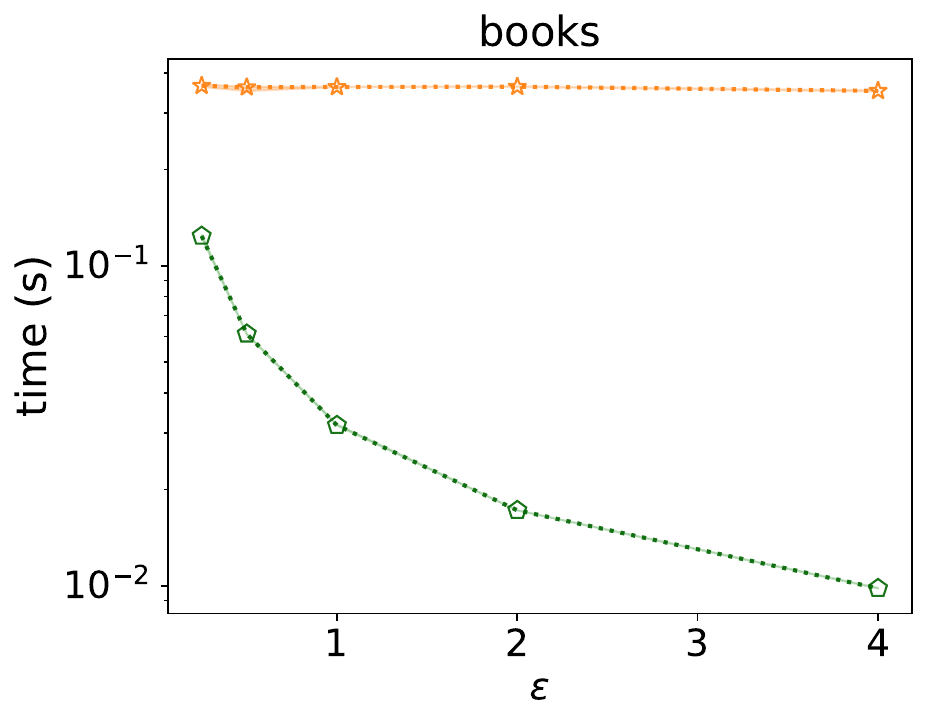}
        \hspace{-2mm}
        \includegraphics[scale=0.32]{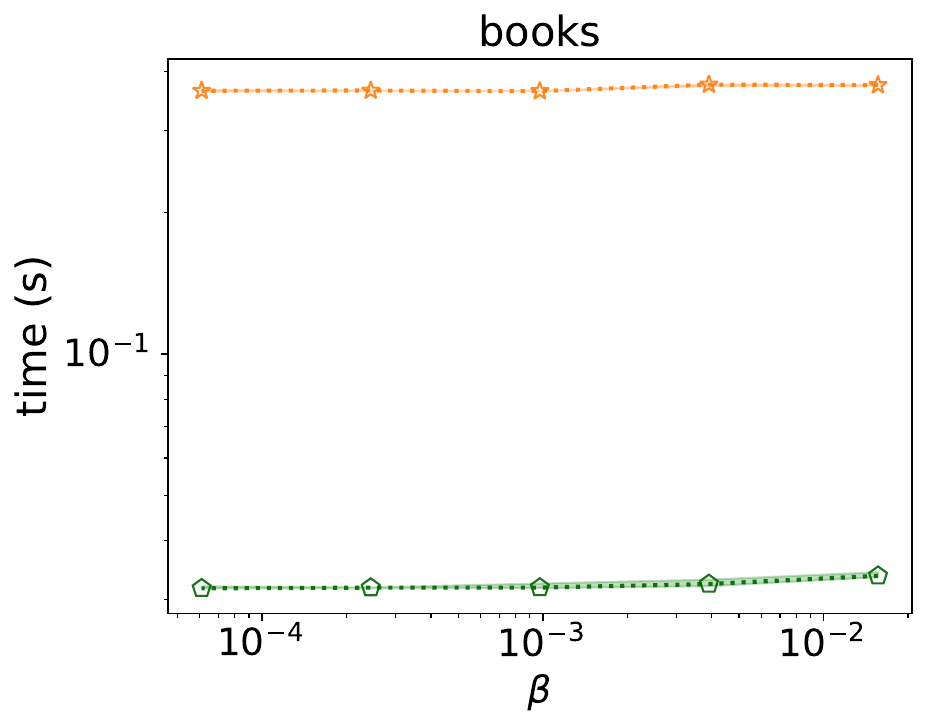} 
    }\\ \vspace{-10pt}
    \makebox[\textwidth]{
        \includegraphics[scale=0.32]{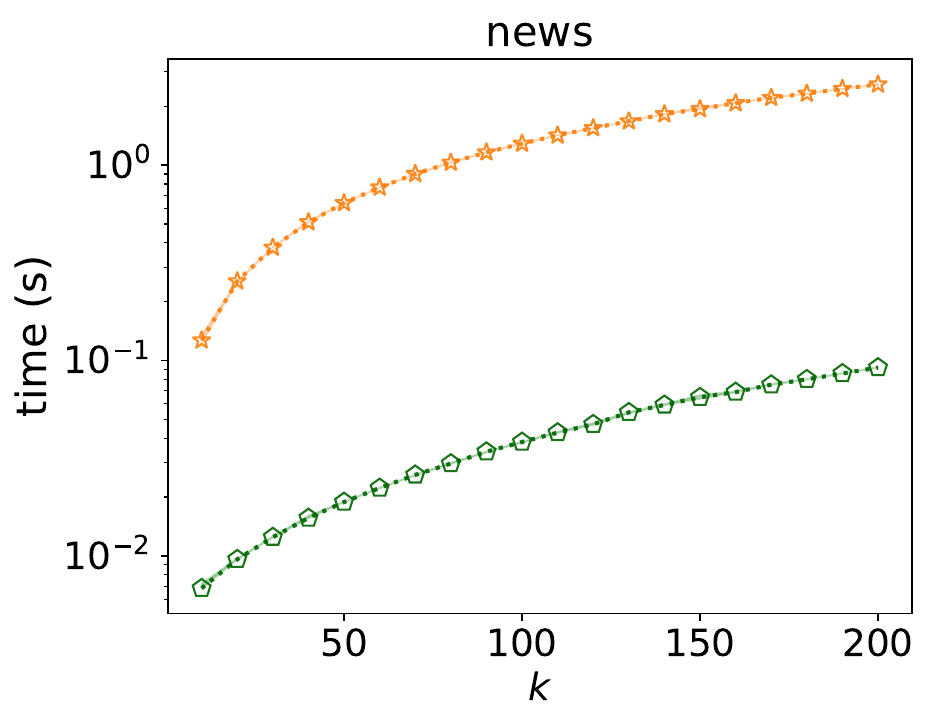}
        \hspace{-2mm}
        \includegraphics[scale=0.32]{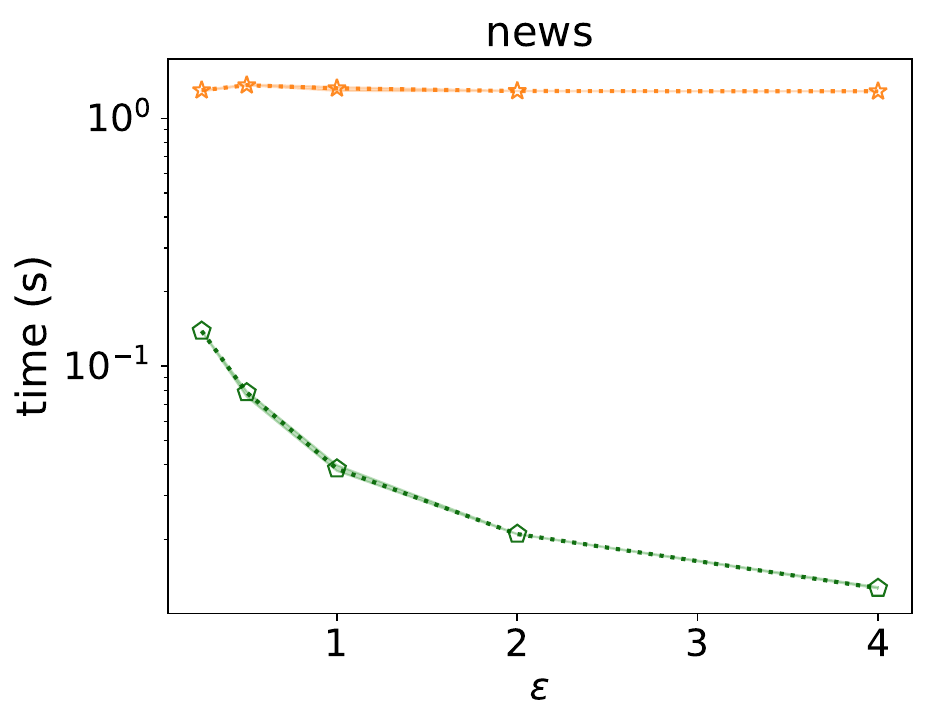}
        \hspace{-2mm}
        \includegraphics[scale=0.32]{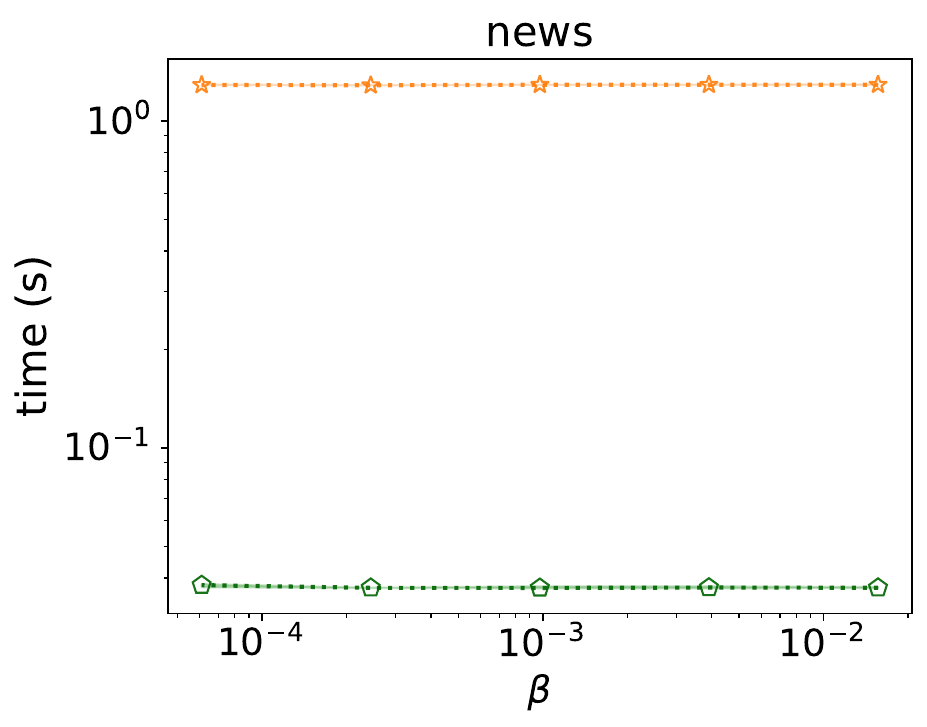} 
    } \\ \vspace{-10pt}
    \makebox[\textwidth]{
        \includegraphics[scale=0.32]{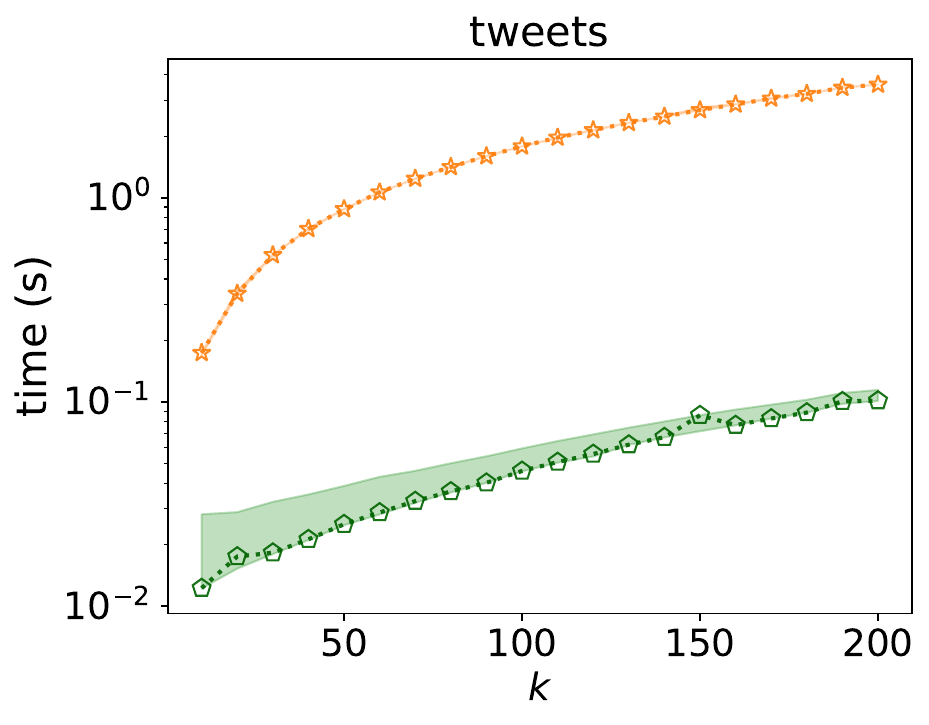}
        \hspace{-2mm}
        \includegraphics[scale=0.32]{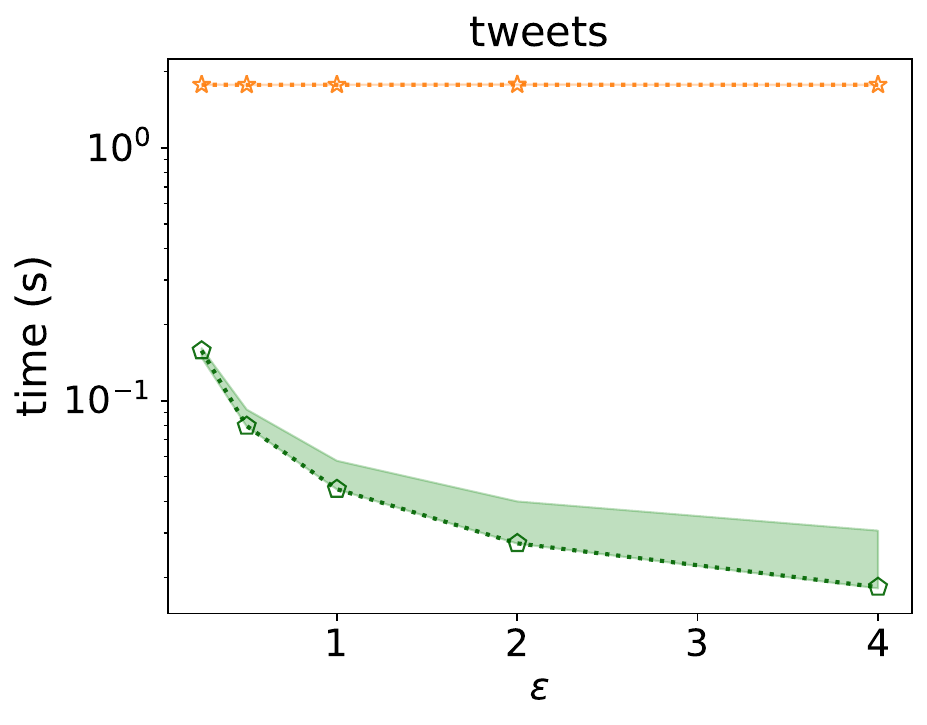}
        \hspace{-2mm}
        \includegraphics[scale=0.32]{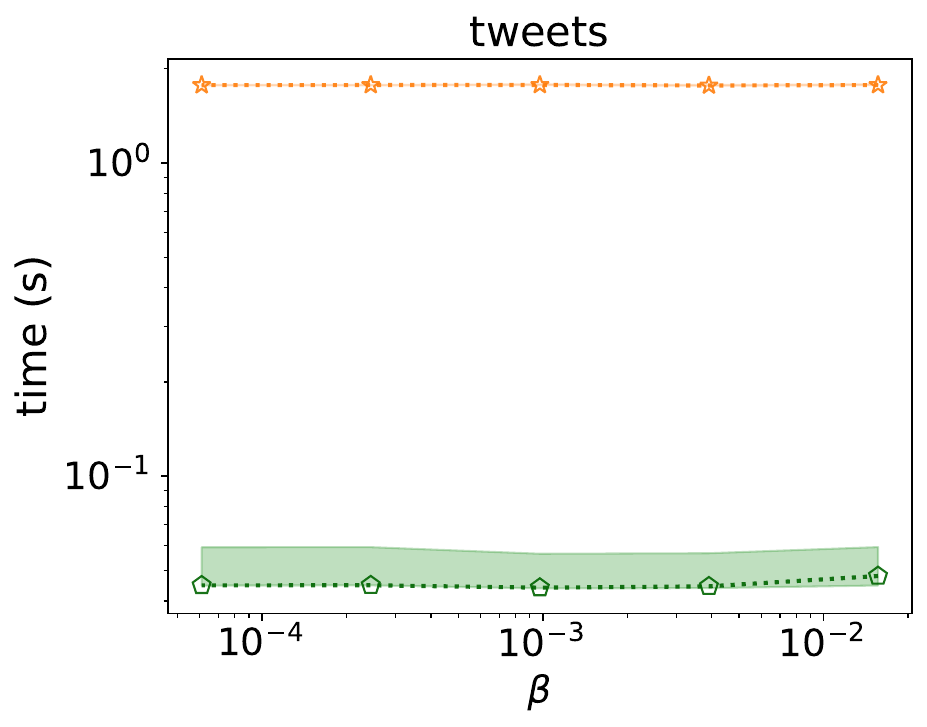} 
    } \\ \vspace{-10pt}
    \makebox[\textwidth]{
        \includegraphics[scale=0.32]{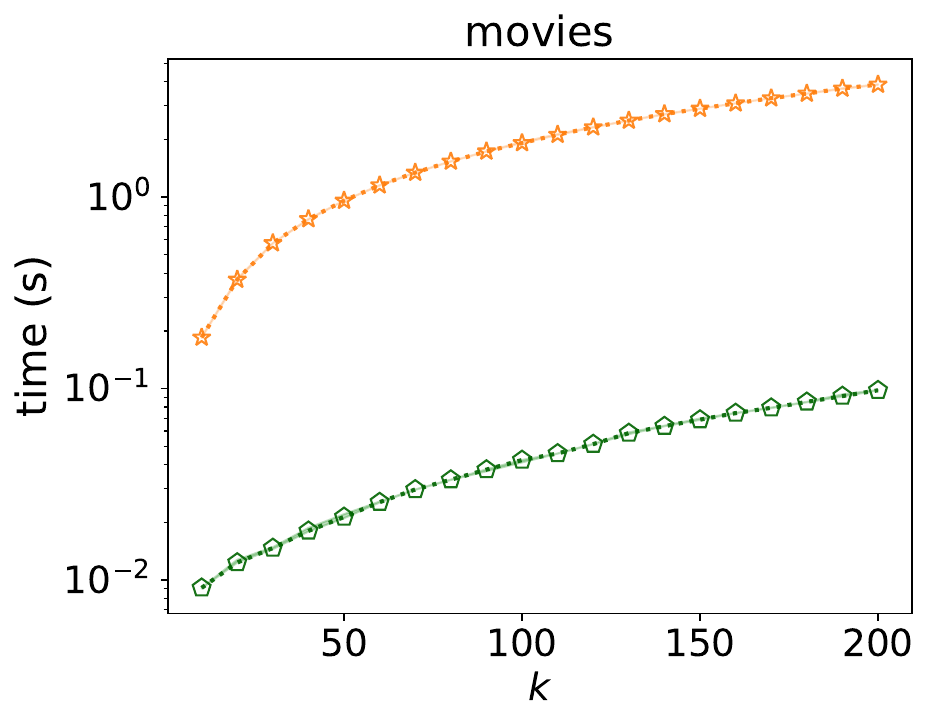}
        \hspace{-2mm}
        \includegraphics[scale=0.32]{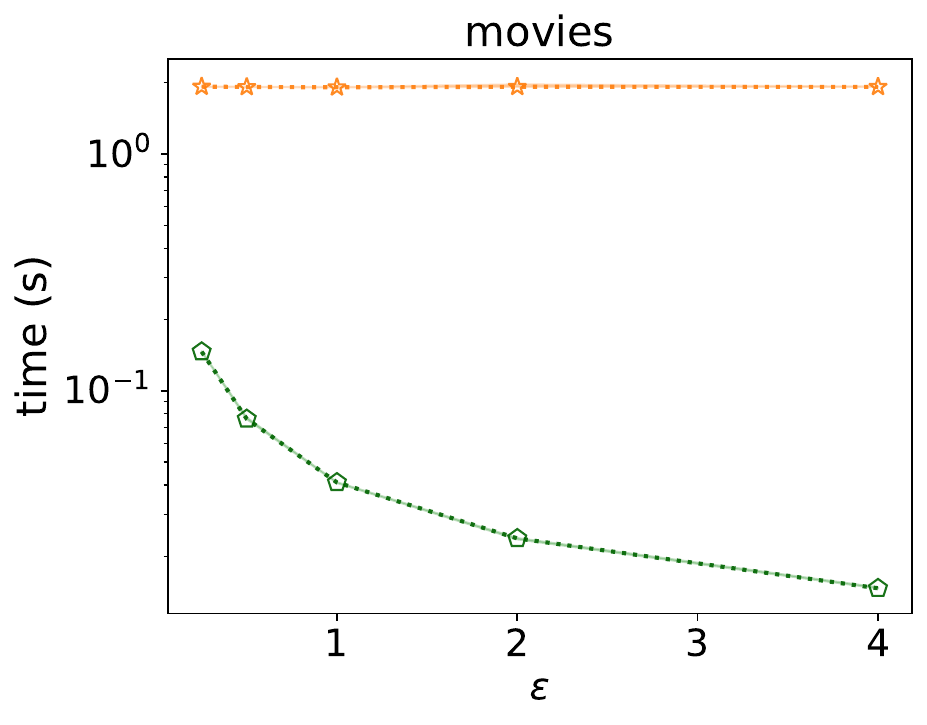}
        \hspace{-2mm}
        \includegraphics[scale=0.32]{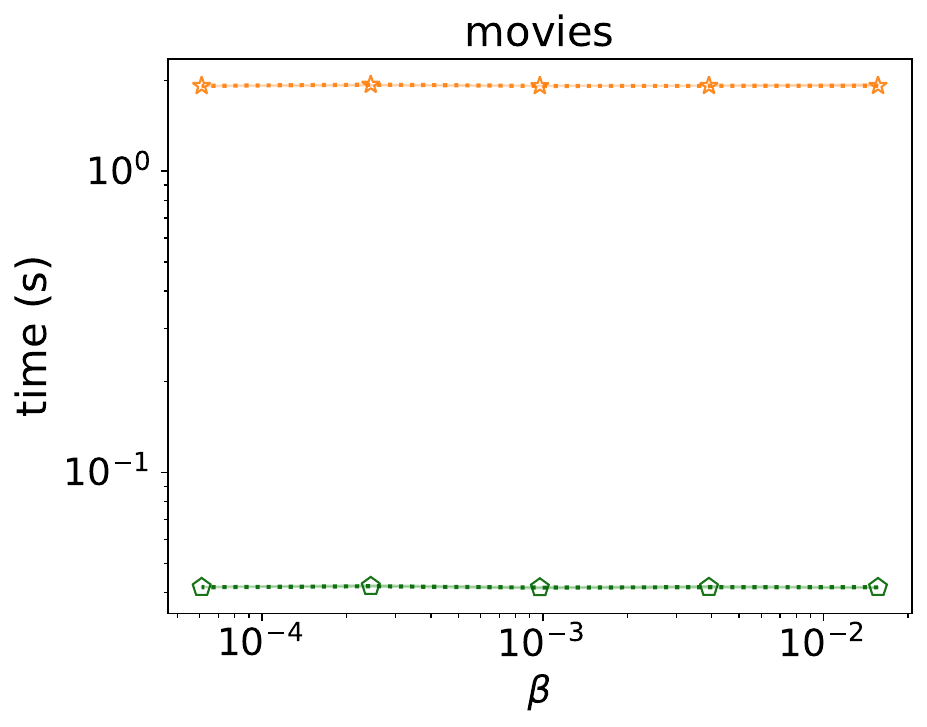} 
    } \\ \vspace{-10pt}
    \makebox[\textwidth]{
        \includegraphics[scale=0.32]{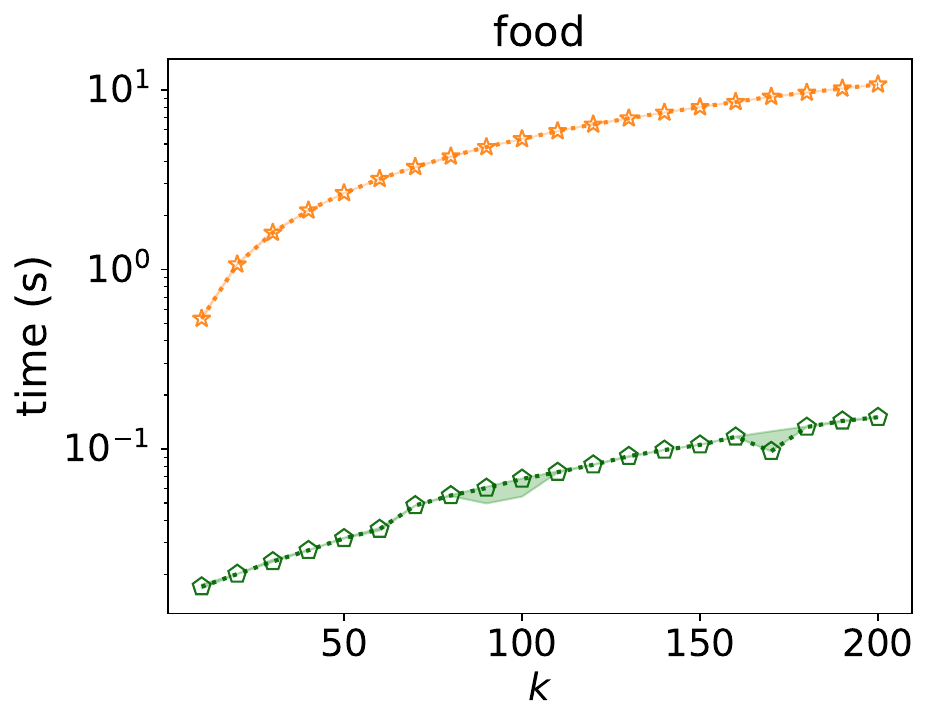}
        \hspace{-2mm}
        \includegraphics[scale=0.32]{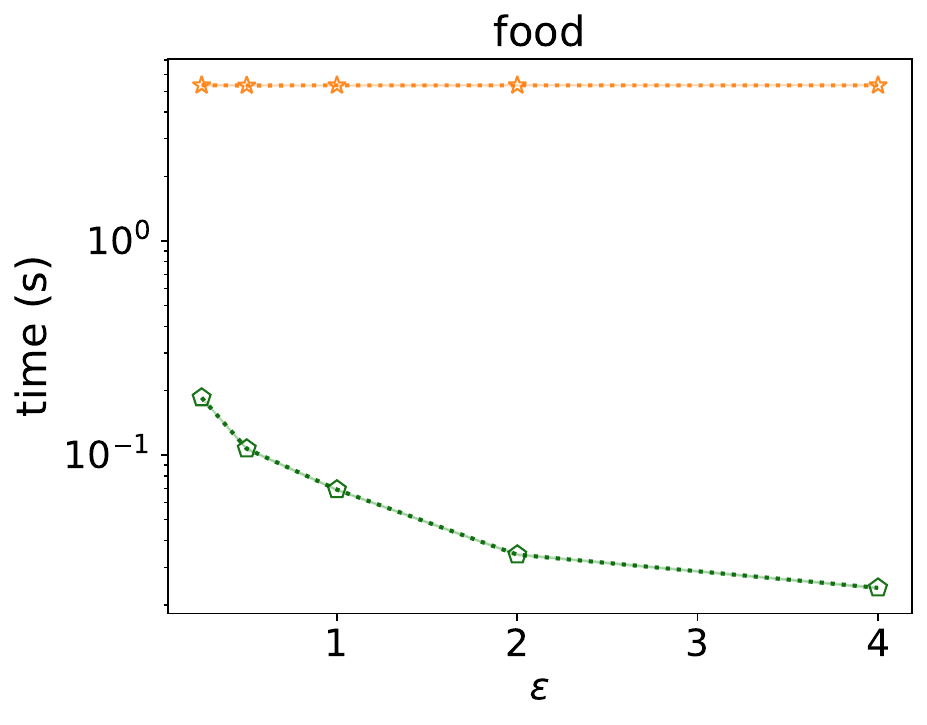}
        \hspace{-2mm}
        \includegraphics[scale=0.32]{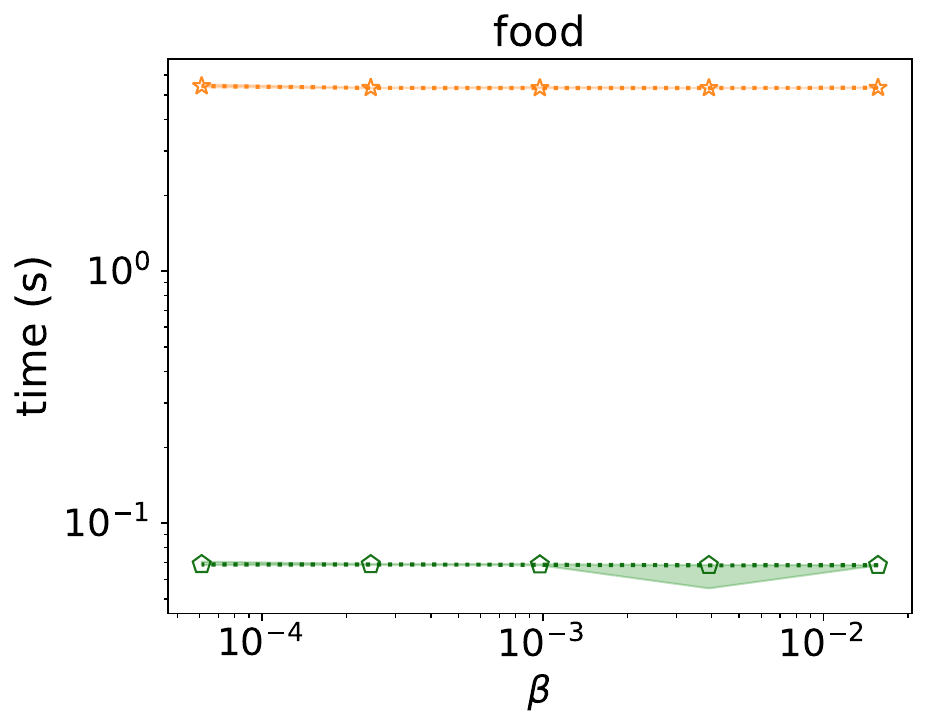} 
    } \\ \vspace{-10pt}
    \includegraphics[scale=0.5]{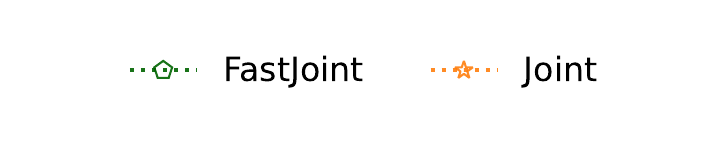}
    \vspace{-15pt}
    \caption{
        \centering
        \textbf{Left}: Running time vs $k$. \,
        \textbf{Center}: Running time vs $\eps$. \,
        \textbf{Right}: Running time vs $\beta$. 
    }
    \label{fig: joint without time results}
\end{figure}

\hfill

\end{document}